\documentclass[12pt,reqno,a4paper]{amspaper}
\usepackage{mathdef}
\usepackage{natbib}
\newcommand*\tageq{\refstepcounter{equation}\tag{\theequation}}
\graphicspath{{graphslocstatftsadj/}{../graphslocstatftsadj/}}
\usepackage[utf8]{inputenc}
\usepackage{comment}
\usepackage{setspace}
\usepackage{graphicx}
\usepackage{multirow}
\usepackage{lscape}
\usepackage{afterpage}
\usepackage{caption}
\usepackage{subcaption}
\usepackage[section]{placeins}
\usepackage{flafter}
\usepackage{rotating}
\usepackage{eufrak}
\usepackage{mathabx}
\usepackage{mathrsfs}
\usepackage{bbm}
\usepackage{url}

\usepackage{todonotes}

\DeclareMathAlphabet{\mathpzc}{OT1}{pzc}{m}{it}
\newcommand{\snorm}[1]{\vvvert{#1}\vvvert}
\newcommand{\bigsnorm}[1]{\big\vvvert{#1}\big\vvvert}
\newcommand{\Bigsnorm}[1]{\Big\vvvert{#1}\Big\vvvert}
\newcommand{\biggsnorm}[1]{\bigg\vvvert{#1}\bigg\vvvert}
\usepackage{enumitem}
\AtBeginDocument{%
  \addtolength\abovedisplayskip{-0.045\baselineskip}%
  \addtolength\belowdisplayskip{-0.045\baselineskip}%
}
\makeatletter
\newcommand{\leqnomode}{\tagsleft@true\let\veqno\@@leqno}
\newcommand{\reqnomode}{\tagsleft@false\let\veqno\@@eqno}
\makeatother
\newcommand{\E}{\mean}
\newcommand{\Var}{\var}
\newcommand{\Cov}{\cov}
\newcommand{\Hspace}{\mathbbm{H}}
\newcommand{\cumk}[1]{c_{t_1,\ldots,t_{{#1}-1}}}
\newcommand{\cumker}[1]{c_{t_1,\ldots,t_{{#1}-1}}(\tau_1,\ldots \tau_{#1})}
\newcommand{\cumksym}[1]{c_{t_1,\ldots,t_{{#1}-1},0}}
\newcommand{\cumop}[1]{\mathcal{C}_{t_1,\ldots,t_{{#1}-1}}}

\newcommand{\cumkersym}[1]{c_{t_1,\ldots,t_{{#1}-1},0}(\tau_1,\ldots \tau_{#1})}

\newcommand{\AtT}[1]{A^{(T)}_{t,#1}}
\newcommand{\AttT}[2]{A^{(T)}_{#1,#2}}
\newcommand{\AAtT}[1]{\mathcal{A}^{(T)}_{t,#1}}
\newcommand{\AAttT}[2]{\mathcal{A}^{(T)}_{#1,#2}}
\newcommand{\Au}[1]{{A}_{u,#1}}

\newcommand{\AAu}[1]{\mathcal{A}_{u,#1}}
\newcommand{\AAuu}[2]{\mathcal{A}_{#1,#2}}

\newcommand{\aatT}[1]{\mathpzc{a}^{(T)}_{t,#1}}
\newcommand{\aau}[1]{\mathpzc{a}_{u,#1}}

\newcommand{\XT}[1]{X_{#1,T}}
\newcommand{\XNT}[1]{X^{(N)}_{#1,T}}
\newcommand{\Xu}[1]{X^{(u)}_{#1}}
\newcommand{\Xuu}[2]{X^{(#1)}_{#2}}
\newcommand{\ZN}[1]{Z^{(N)}_{#1}}
\newcommand{\cT}{c^{(T)}}
\newcommand{\fT}{f^{(T)}}
\newcommand{\hatfT}{\hat f^{(T)}}
\newcommand{\FT}{\mathcal{F}^{(T)}}
\newcommand{\hatFT}{\hat{\mathcal{F}}^{(T)}}
\newcommand{\F}{\mathcal{F}}
\newcommand{\CCT}{\mathcal{C}^{(T)}}
\newcommand{\DT}{D^{(T)}}
\newcommand{\IT}{I^{(T)}}
\newcommand{\hN}[1]{h_{#1,N}}
\newcommand{\HN}{H_{N}}
\newcommand{\HNN}[1]{H_{#1,N}}
\newcommand{\IID}{\mbox{\itshape i.i.d.} }
\newcommand{\kappat}{\kappa_{\mathrm{t}}}
\newcommand{\kappaf}{\kappa_{\mathrm{f}}}
\newcommand{\bfT}{b_{\mathrm{f},T}}
\newcommand{\btT}{b_{\mathrm{t},T}}
\newcommand{\Kf}{K_{\mathrm{f}}}
\newcommand{\Kt}{K_{\mathrm{t}}}
\newcommand{\KfT}{K_{\mathrm{f},T}}
\newcommand{\KtT}{K_{\mathrm{t},T}}
\newcommand{\ET}{\hat E^{(T)}}
\newcommand{\Iso}{\mathop{\mathcal{T}}}

\begin{document}
\reqnomode

\title[Locally stationary functional time series]%
{Locally stationary functional time series}

\author{Anne van Delft}
\thanks{Ruhr-Universit{\"a}t Bochum, Fakult{\"a}t f{\"u}r Mathematik, 44780 Bochum, Germany} 
\thanks{{\em E-mail address:} anne.vandelft@rub.de (A.~van Delft)}
\author{Michael Eichler}
\thanks{Department of Quantitative Economics,
Maastricht University, P.O.~Box 616, 6200 MD Maastricht, The Netherlands}
\thanks{{\em E-mail address:} m.eichler@maastrichtuniversity.nl (M.~Eichler)}

\dedicatory{{\upshape\today}}
\begin{abstract}
The literature on time series of functional data has focused on processes of which the probabilistic law is either constant over time or constant up to its second-order structure. Especially for long stretches of data it is desirable to be able to weaken this assumption. This paper introduces a framework that will enable meaningful statistical inference of functional data of which the dynamics change over time. We put forward the concept of local stationarity in the functional setting and establish a class of processes that have a functional time-varying spectral representation. Subsequently, we derive conditions that allow for fundamental results from nonstationary multivariate time series to carry over to the function space. In particular, time-varying functional ARMA processes are investigated and shown to be functional locally stationary according to the proposed definition. As a side-result, we establish a  Cram{\'e}r representation for an important class of weakly stationary functional processes. Important in our context is the notion of a time-varying spectral density operator of which the properties are studied and uniqueness is derived. Finally, we provide a consistent nonparametric estimator of this operator and show it is asymptotically Gaussian using a weaker tightness criterion than what is usually deemed necessary.

\smallskip 

\end{abstract}

\maketitle

\section{Introduction}

In functional data analysis, the variables of interest take the form of smooth functions that vary randomly between repeated observations or measurements. Thus functional data are represented by random smooth functions $X(\tau)$, $\tau\in D$, defined on a continuum $D$. Examples of functional data are concentration of fine dust as a function of day time, the growth curve of children as functions of age, or the intensity as a function of wavelength in spectroscopy. Because functional data analysis deals with inherently infinite-dimensional data objects, dimension reduction techniques such as functional principal component analysis (FPCA) have been a focal point in the literature. Fundamental for these methods is the existence of a Karhunen-Lo{\`e}ve decomposition of the process \citep{Karhunen1947,Loeve1948}. Some noteworthy early contributions are \citet{Kleffe1973, Grenander1981, Dauxois1982, Besse1986}. For an introductory overview of the main functional data concepts we refer to \citet{Ramsay2005} and  \citet{Ferraty2006}.

Most techniques to analyze functional data are developed under the assumption of independent and identically distributed functional observations and focus on capturing the first- and second-order structure of the process. A variety of functional data is however collected sequentially over time. In such cases, the data can be described by a \textit{functional time series} $\{X_t(\tau)\}_{t\in \znum}$. Since such data mostly show serial dependence, the assumption of \IID repetitions is violated. Examples of functional time series in finance are bond yield curves, where each function is the yield of the bond as a function of time to maturity \citep[e.g.][]{Bowsher2008,Hays2012} or the implied volatility surface of a European call option as a function of moneyness and time to maturity. In demography, mortality and fertility rates are given as a function of age \citep[e.g.][]{Erbas2007,Hyndman2007,Hyndman2008}, while in geophysical sciences, magnometers record the strength and direction of the magnetic field every five seconds. Due to the wide range of applications, functional time series and the development of techniques that allow to relax the \IID assumption have received an increased interest in recent years.  

The literature on functional time series has mainly centered around stationary linear models \citep[][]{Mas2000,  Bosq2002, Dehling2005} and prediction methods \citep{Antoniadis2006,Bosq2007,Aue2015}. A general framework to investigate the effect of temporal dependence among functional observations on existing techniques has been provided by \citet{Hormann2010}, who introduce $L^p_m$ approximability as a moment-based notion of dependence. 

Violation of the assumption of identically distributed observations has been examined in the setting of change-point detection \citep[e.g.][]{Berkes2009,Hormann2010, Aue2009, Horvath2010,Gabrys2010}, in the context of functional regression by \citet{Yao2005, Cardot2006} and in the context of common principal component models by \citet{Benko2009}.

Despite the growing literature on functional time series, the existing theory has so far been limited to strongly or weakly stationary processes. With the possibility to record, store and analyze functional time series of an increasing length, the common assumption of (weak) stationarity becomes more and more implausible. \textcolor{black}{For instance, in meteorology the distribution of the daily records of temperature, precipitation and cloud cover for a region, viewed as three related functional surfaces, may change over time due to global climate changes. In the financial industry, implied volatility of an option as a function of moneyness changes over time. Other relevant examples appear in the study of cognitive functions such as high-resolution recordings from local field potentials, EEG and MEG.  It is widely known that these type of data have a time-varying spectral structure and their statistical treatment requires to take this into account.} While heuristic approaches such as localized estimation are readily implemented and applied, a statistical theory for inference from nonstationary functional time series is yet to be developed.

\textcolor{black}{The objective of the current paper is to develop a framework for inference of nonstationary functional time series that allows the derivation of large sample approximations for estimators and test statistics. For this, we extend the concept of locally stationary processes \citep{Dahlhaus1996a} to the functional time series setting. We show that fundamental results for multivariate time series can be carried over to the function space, which is a nontrivial task. Our work, which provides a basis for inference of nonstationary functional time series, focuses on frequency domain-based methods and therefore also builds upon the work by \citet{Panar2013b,PanarTav2013a}. Functional data carry infinite-dimensional intrinsic variation and in order to exploit this rich source of information, it is important to optimally extract defining characteristics to finite dimension via techniques such as functional PCA (FPCA). In the case of stationary dependent functional data, the shape and smoothness properties of the random curves are completely encoded by the spectral density operator, which has been shown to allow for an optimal lower dimension representation via dynamic FPCA \citep[see e.g.,][]{Panar2013b,Hormann2015}. Since the assumption of weak stationarity is often too restrictive, we aim to provide the building blocks for statistical inference of nonstationary functional time series and for the development of techniques such as time-varying dynamic FPCA. In particular, our framework will be essential for the development of optimal dimension reduction techniques via a {\em local} functional Cram{\'e}r-Karhunen-Lo{\`e}ve representation. Such a representation must not only take into account the between- and within curve dynamics but also that these are time-varying. Moreover, the frequency domain arises quite naturally in certain applications such as brain data imaging and is moreover very useful in nonparametric specifications. For example, \citet{avd16} use our framework to derive a test for stationarity of functional time series against nonstationary alternatives with slowly changing dynamics. 
}

The paper is structured as follows. In section \ref{section2}, we first introduce some basic notation and methodology for functional data and relate this in a heuristic manner to the concept of locally stationary time series and introduce the definition of a locally stationary functional time series.  In section \ref{section4}, we demonstrate that time-varying functional ARMA models have a causal solution and are functionally locally stationary according to the definition in section \ref{section2}. This hinges on the existence of stochastic integrals for operators that belong to a particular Bochner space. In section \ref{section3}, the time-varying spectral density operator is defined and its properties are derived. In particular, we will show uniqueness of the time-varying spectral density operator. In section \ref{expandcov}, we derive the distributional properties of a local nonparametric estimator of the time-varying spectral density operator and deduce a central limit theorem. The results are illustrated by application to a simulated functional autoregressive process in section \ref{simulation-sect}. Technical details and several auxiliary results that are of independent interest are proved in the Appendix.

\section{Locally stationary functional time series}
\label{section2}

Let $X=\{X_t\}_{t=1,\ldots,T}$ be a stochastic process taking values in the Hilbert space $H=L^2([0,1])$ of all real-valued functions that are square integrable with respect to the Lebesgue measure. While current theory for such processes is limited to the case where $\{X_t\}$ is either strictly or weakly stationary, we consider nonstationary processes with dynamics that vary slowly over time and thus can be considered as approximately stationary at a local level.

As an example, consider the functional autoregressive process $X$ given by
\begin{align*}\label{eq:tvfar1ex}
X_t (\tau) = B_t\big(X_{t-1}\big)(\tau) +\varepsilon_{t}(\tau), \quad  \tau \in [0,1],
\end{align*}
for $t=1,\ldots,T$, where the errors $\varepsilon_t$ are independent and identically distributed random elements in $H$ and $B_{t}$ for $t=1,\ldots,T$ are bounded operators on $H$. Assuming that the autoregressive operators $B_t$ change only slowly over time, we can still obtain estimates by treating the process as stationary over short time periods. However, since this stationary approximation deteriorates over longer time periods, standard asymptotics based on an increasing sample size $T$ do not provide suitable distributional approximations for the finite sample estimators. Instead we follow the approach by \citet{Dahlhaus1996a, d97} and define local stationary processes in a functional setting based on an infill asymptotics. The main idea of this approach is that for increasing $T$ the operator $B_t$ is still `observed' on the same interval but on a finer grid, resulting in more and more observations in the time period over which the process can be considered as approximately stationary. Thus we consider a family of functional processes
\begin{align*}
\XT{t}(\tau)
= B_{t/T}\big(\XT{t-1}\big)(\tau) +\varepsilon_{t}(\tau),\qquad\tau \in [0,1],\quad 1\leq t\leq T,
\end{align*}
indexed by $T\in\nnum$ that all depend on the common operators $B_{u}$ indexed by rescaled time $u=t/T$. Consequently, we in fact examine a triangular array of random functions that share common dynamics as provided by the continuous operator-valued function $B_{u}$, $u\in [0,1]$. For each $T$, a different `level' of the sequence is thus considered where the dynamics change more slowly for increasing values of $T$. We will establish a class of functional time series with a time-varying functional spectral representation that includes interesting processes such as the above example and higher order time-varying functional ARMA models. The framework as provided in this paper will allow to investigate how nonstationarity affects existing methods, such as (dynamic) FPCA, and  how these methods  should be adjusted in order to be robust for changing characteristics. Similarly as \citet{SubbaRao2006} and \citet{Vogt2012} in the case of ordinary time series, we call a functional time series locally stationary if it can be locally approximated by a stationary functional time series. In the following definition, $\norm{\cdot}_2$ denotes the $L^2$-norm of $H$.

\begin{definition}[Local stationarity]
\label{FLS}
A sequence of stochastic processes $\{\XT{t}\}_{t\in\znum}$ indexed by $T\in\nnum$ and taking values in $H$ is called locally stationary if for all rescaled times $u\in[0,1]$ there exists an $H$-valued strictly stationary process $\{X^{(u)}_t\}_{t\in\znum}$, such that
\[
\bignorm{\XT{t}-\Xu{t}}_{2}
\leq\big(\big|\tfrac{t}{T}-u\big|+\tfrac{1}{T}\big)\,P_{t,T}^{(u)}\qquad a.s. 
\]
for all $1\leq t\leq T$, where $P_{t,T}^{(u)}$ is a positive real-valued process such that for some $\rho>0$ and $C<\infty$ the process satisfies $\mean\big(\big|P_{t,T}^{(u)}\big|^\rho\big)<C$ for all $t$ and $T$ and uniformly in $u\in[0,1]$.
\end{definition}

\textcolor{black}{For the purpose of illustration, a very simple locally stationary functional time series is depicted in figure \ref{fig:Xlocstatexam}\,(A). Note that visual interpretation of a functional time series can be extremely difficult, especially when it is driven by many interacting components. The process in figure \ref{fig:Xlocstatexam}\,(A) is driven solely by two components and is generated as
\[
\XT{t}(\tau) = \xi_{t,T} \,\phi_1(\tau) +\chi_{t,T}\, \phi_2(\tau) \qquad \tau \in [0,1] \tageq \label{eq:XtTplot}
\]
where $\phi_1, \phi_2$ are basis functions of $H$ and the random coefficients $\xi_{t,T}, \chi_{t,T}$ are independent Gaussian time-varying AR(2). The parameters of the time-varying AR(2) models are chosen in such a way that the magnitude and phase of the roots of the characteristic polynomials of $\xi_{t,T}$ and $\chi_{t,T}$ vary cyclically with rescaled time $t/T$ but in opposite direction as time progresses. The corresponding coefficient curves 
\FloatBarrier
\begin{figure}[h!]
   \begin{subfigure}[t]{.45\textwidth}
    \centering
      \hspace*{-40pt}   \includegraphics[width=1.4\linewidth]{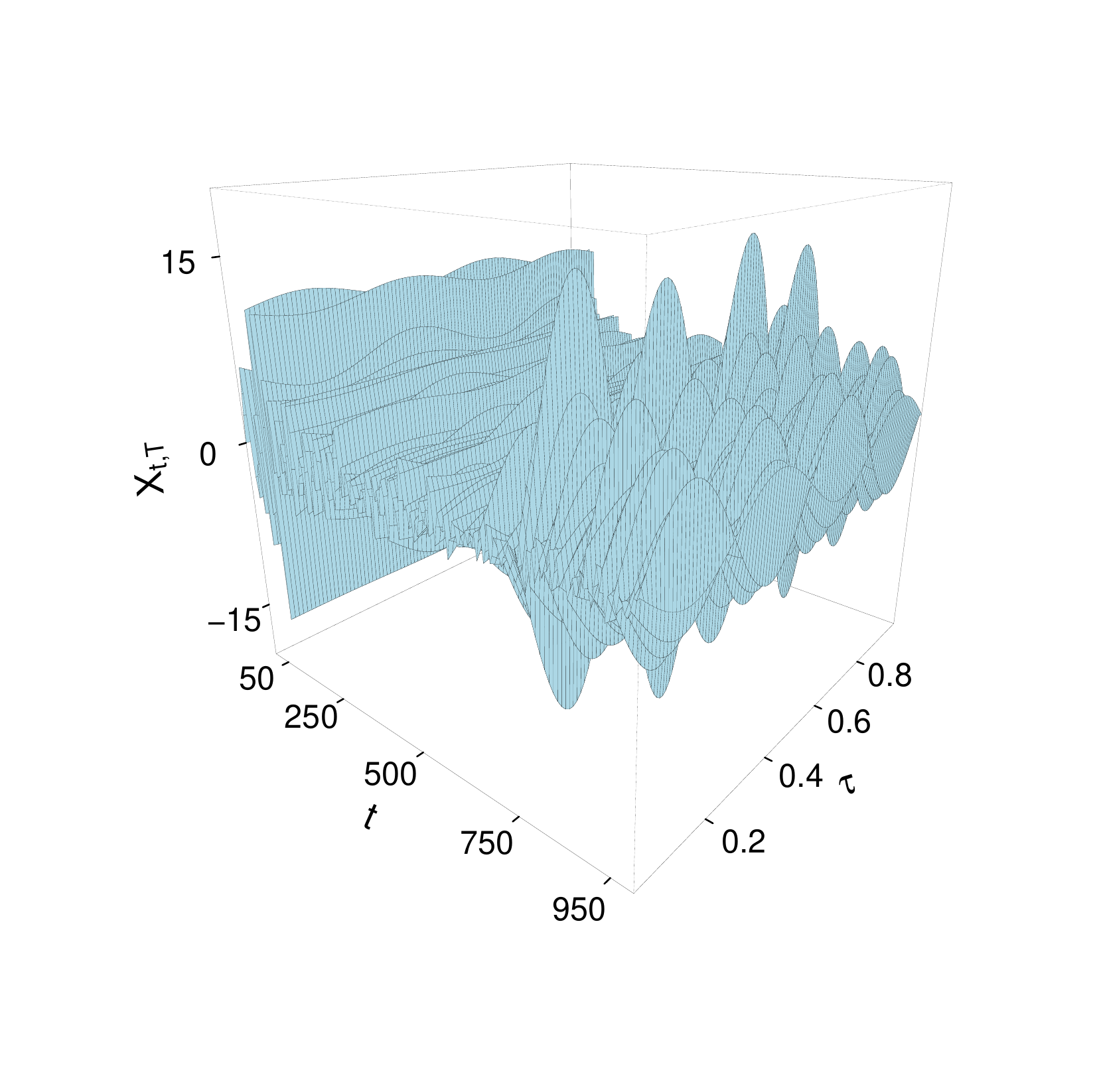}
 \vspace*{-50pt} \caption{Locally stationary process}
  \end{subfigure}
  \hfill
 \begin{subfigure}[t]{.45\textwidth}
    \centering
        \hspace*{-40pt}  \includegraphics[width=1.4\linewidth]{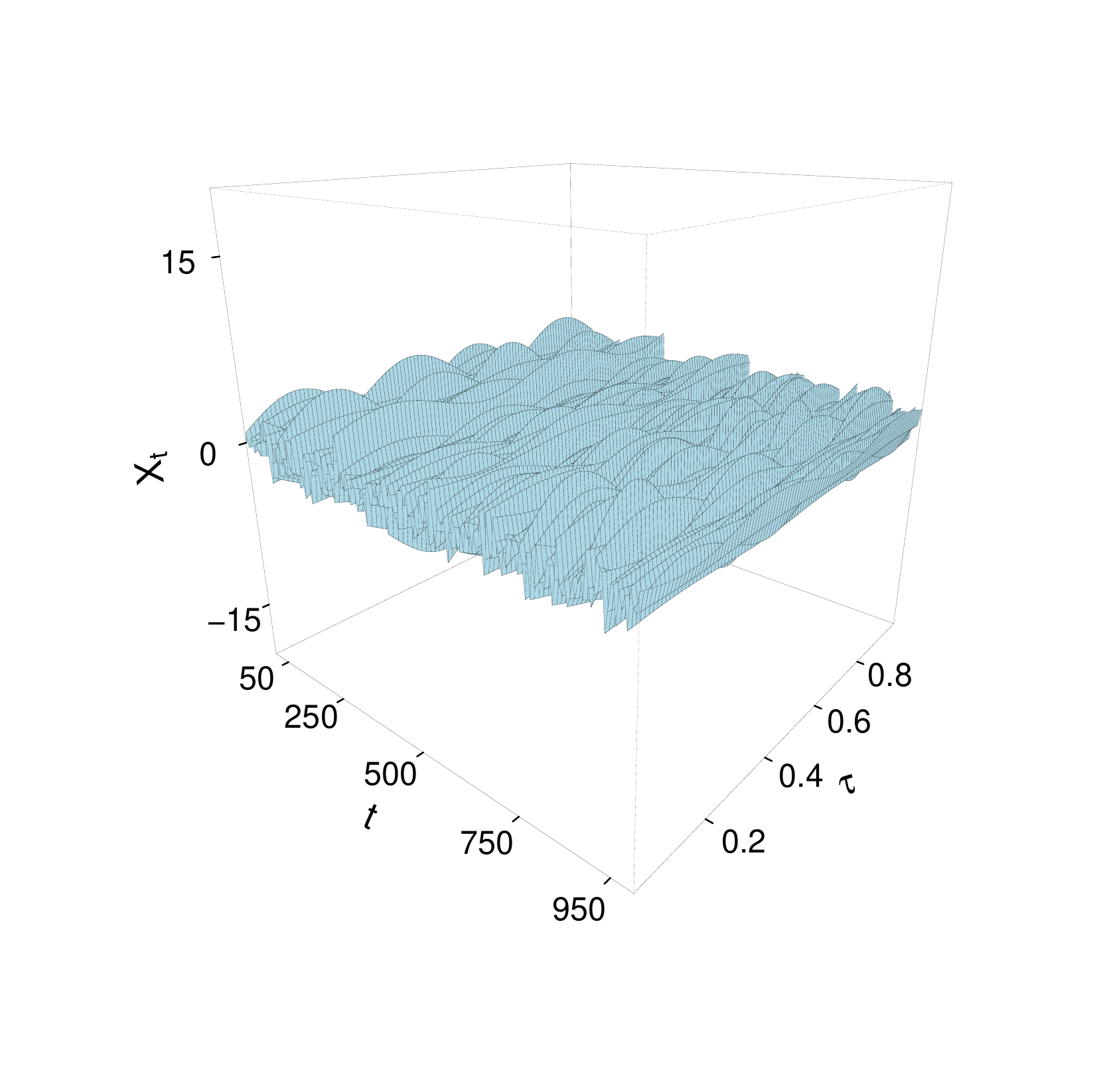}
   \vspace*{-50pt}   \caption{Stationary process}
  \end{subfigure}
\medskip
 \vspace*{-30pt}
  \begin{subfigure}[t]{.45\textwidth}
    \centering
       \hspace*{-10pt}   \includegraphics[width=\linewidth]{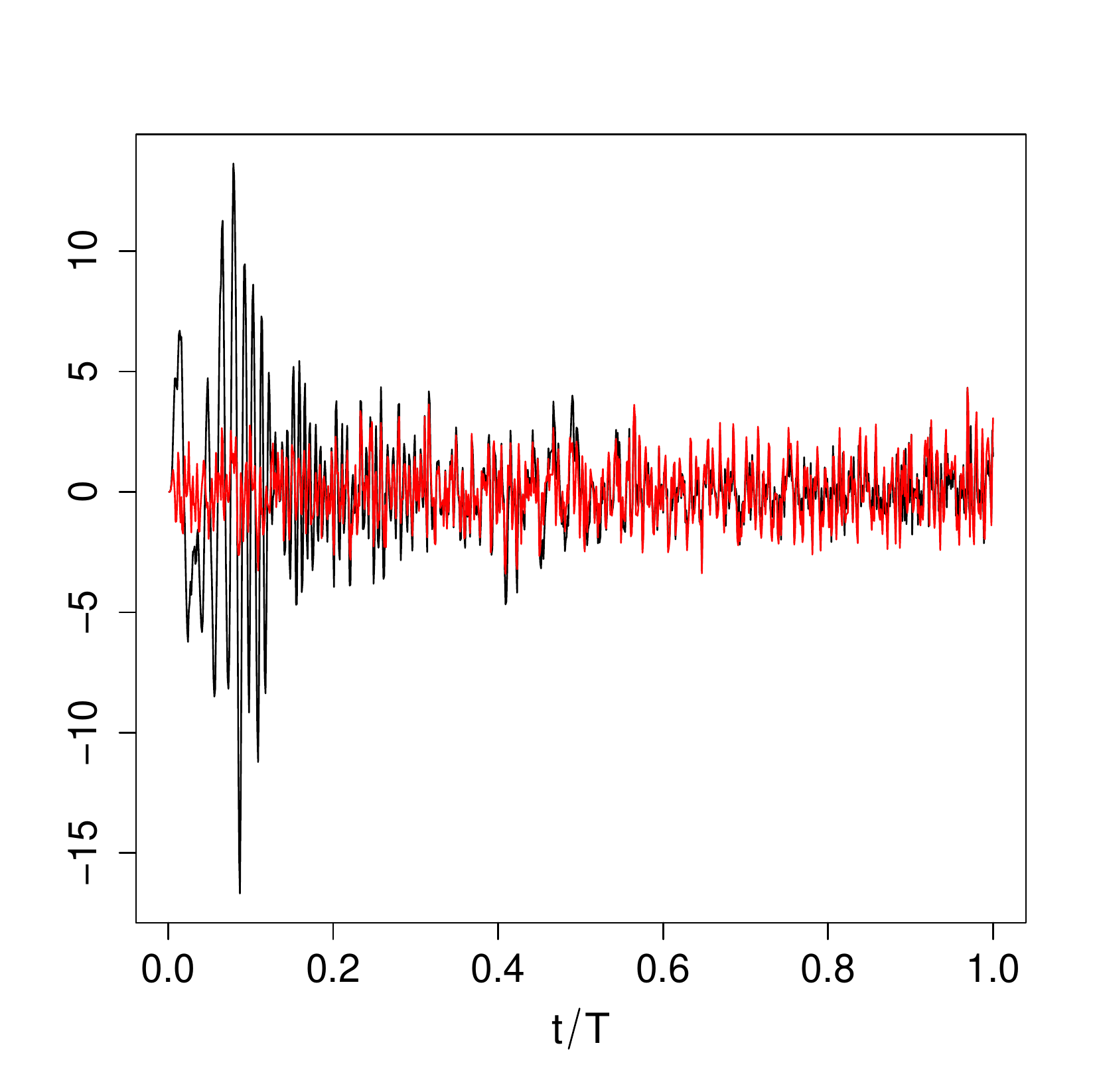}
 \vspace*{-5pt}    \caption{Time series of the coefficient $\xi_{t,T}$ of 
the locally stationary process in (A) (black curve) and of the 
coefficient $\xi_t$ of the stationary process in (B) (red curve)}
  \end{subfigure}
  \hfill
  \begin{subfigure}[t]{.45\textwidth}
    \centering
      \hspace*{-10pt}    \includegraphics[width=\linewidth]{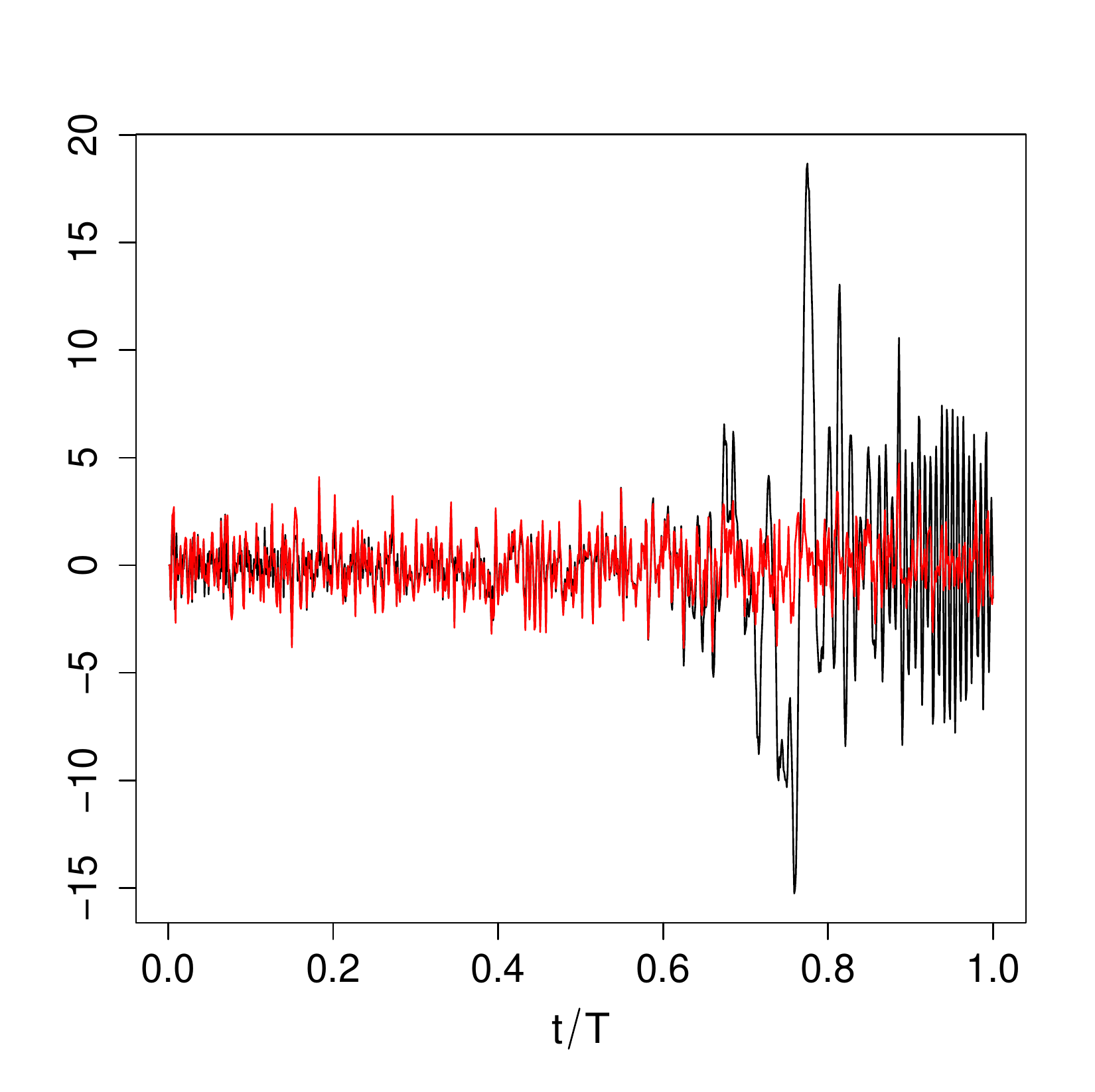}
    \vspace*{-5pt} \caption{Time series of the coefficient $\chi_{t,T}$ of 
the locally stationary process in (A) (black curve) and of the 
coefficient $\chi_t$ of the stationary process in (B) (red curve)}
  \end{subfigure}
  \caption{Comparison of stationary and local stationary functional time 
series}\label{fig:Xlocstatexam}
\end{figure}
\FloatBarrier
are plotted in figure \ref{fig:Xlocstatexam}\,(C) and \ref{fig:Xlocstatexam}\,(D), respectively. The dependence structure of the two driving components vary from independence to close to unit root behavior and this varying cyclical behavior is also clearly visible in the resulting functional process $\{\XT{t}\}$. In order to contrast this with behavior observed under stationarity, figure \ref{fig:Xlocstatexam}\,(B) depicts the closely related weakly stationary functional process of \eqref{eq:XtTplot} where the random coefficients are generated using two stationary AR(2) with parameters specified as the time average of $ \xi_{t,T}$ and $\chi_{t,T}$, respectively. A comparison of the plots for the locally stationary and the stationary case shows a clear difference in the dynamics of the two processes, which is particularly discernible in the projections on the Fourier components. The example thus indicates the effect of falsely misspecifying a locally stationary process as stationary on statistical inference.}

Definition \ref{FLS} is broad and is further investigated in \citet{avd16}. It will allow for the development of statistical inference procedures for nonstationary functional time series and in particular encompasses nonlinear functional models. Nonlinear functional time series is a topic that is relatively unexplored. Possible relevant models that are worth investigating are, for instance, time-varying additive functional regression \citep[][]{Muller2008} and time-varying functional ARCH models \citep[][]{hhr2013}.
\textcolor{black}{However, as the focus of this paper is on frequency domain based methods, it is more appropriate to work with an alternative characterization of local stationarity in terms of spectral representations, which we discuss below. We start by introducing the necessary terminology on operators and spectral representations for stationary functional time series.}


\subsection{Functional spaces and operators: notation and terminology}

First, we introduce some basic notation and definitions on functional spaces and operators. Let $(T,\mathcal{B})$ be a measurable space with $\sigma$-finite measure $\mu$. Furthermore, let $E$ be a Banach space with norm $\norm{\cdot}_E$ and equipped with the Borel $\sigma$-algebra. We then define $L^p_E(T,\mu)$ as the Banach space of all strongly measurable functions $f:T\to E$ with finite norm
\[
\norm{f}_{p}=
\norm{f}_{L^p_E(T,\mu)}=\Big(\int \norm{f(\tau)}_E^p\,d\mu(\tau)\Big)^{\tfrac{1}{p}}
\]
for $1\leq p<\infty$ and with finite norm
\[
\norm{f}_{\infty}=
\norm{f}_{L^\infty_E(T,\mu)}=\inf_{\mu(N)=0}\sup_{\tau\in T\without N}\norm{f(\tau)}_E
\]
for $p=\infty$. We note that two functions $f$ and $g$ are equal in $L^p$, denoted as $f \overset{L^p}{=} g$, if $\norm{f-g}_p=0$. If $E$ is a Hilbert space with inner product $\innerprod{\cdot}{\cdot}_E$ then $L^2_E(T,\mu)$ is also a Hilbert space with inner product
\[
\innerprod{f}{g}=
\innerprod{f}{g}_{L^2_E(T,\mu)}=\int\innerprod{f(\tau)}{g(\tau)}_E\,d\mu(\tau).
\]
For notational convenience, we use the shorter notation $\norm{f}_p$ and $\innerprod{f}{g}$ whenever no ambiguity about the space $L^p_E(T,\mu)$ is possible. Similarly, if $T\subset\rnum^k$ and $\mu$ is the Lebesgue measure on $T$, we omit $\mu$ and write $L^p_E(T)$, and if $E=\rnum$ we write $L^p(T,\mu)$.

Next, an operator $A$ on a Hilbert space $H$ is a function $A:H\to H$. An operator $A$ is said to be compact if the image of each bounded set under $A$ is relatively compact. If $H$ is separable, there exist orthonormal bases $\{\phi_n\}$ and $\{\psi_n\}$ of $H$ and a monotonically decreasing sequence of non-negative numbers $s_n(A)$, $n\in\nnum$ converging to zero, such that
\begin{align}\label{eq:svd}
A\,f= \lsum_{n=1}^\infty s_n(A)\,\innerprod{f}{\psi_n}\,\phi_n
\end{align}
for all $f\in H$. The values $s_n(A)$ are called the {\em singular values} of $A$ and \eqref{eq:svd} is the {\em singular value decomposition} of $A$. For operators on $H$, we denote the {\em Schatten p-class} by $S_p(H)$ and its norm by $\snorm{\cdot}_p$. More specifically, for $p=\infty$, the space  $S_{\infty}(H)$ indicates the space of bounded linear operators equipped with the standard operator norm, while for $1 \leq p < \infty$ the {\em Schatten p-class} is the subspace of all compact operators $A$ on $H$ such that the sequence $s(A)=\big(s_n(A)\big)_{n\in\nnum}$ of singular values of $A$ belongs to $\ell^p$; the corresponding norm is given by $\snorm{A}_p = \norm{s(A)}_p$. For $1 \le p \le q \le \infty$, we have the inclusion $S_{p}(H) \subseteq S_{q}(H)$. Two important classes are the trace-class and the Hilbert-Schmidt operators on $H$, which are given by $S_1(H)$ and $S_{2}(H)$, respectively. More properties of Schatten-class operators and in particular of Hilbert-Schmidt operators are provided in Appendix \ref{properties}. Finally, the adjoint of $A$ is denoted by $A^{\dagger}$ while the identity and zero operator are given by $I_H$ and $O_H$, respectively. As usual, the complex conjugate of $z \in \cnum$ is denoted by $\overline{z}$ and the imaginary number by $\im$.\\

The main object of this paper are functional time series $X=\{X_t\}$ that take values in the Hilbert space $H=L^2([0,1])$. More precisely, for some underlying probability space $(\Omega,\fclass,\prob)$, let $\mathbb{H}=L^2_H(\Omega,\prob)$ be the Hilbert space of all $H$-valued random variables $X$ with finite second moment $\mean\norm{X}^2_2<\infty$. \textcolor{black}{To avoid ambiguities between the norms of $H$ and $\mathbb{H}$, we write $\norm{X}_{\mathbb{H}}$ for the norm in $\mathbb{H}$ and reserve the notation $\norm{X}_2$ for the more frequently used norm in $H$.} Throughout the paper, we assume that $X_t\in\mathbb{H}$. For the spectral representation and Fourier analysis of functional time series $\{X_t\}$, we also require the corresponding spaces $H_\cnum=L^2_\cnum([0,1])$ and $\mathbb{H}_\cnum=L^2_{H_\cnum}(\Omega,\prob)$.
We recall some basic properties of functional time series. First, a functional time series $X$ is called strictly stationary if, for all finite sets of indices $J \subset \mathbb{Z}$, the joint distribution of $\{X_{t+j}\given j \in J\}$ does not depend on $t\in\znum$. Similarly, $X$ is weakly stationary if its first- and second-order moments exist and are invariant under translation in time. In that case, the mean function $m$ of $X$ is defined as the unique element of $H$ such that 
\[
\langle m, g \rangle =\E\langle X_t, g \rangle, \qquad g \in H.
\] 
Furthermore, the $h$--th lag covariance operator $\mathcal{C}_{h}$ is given by
\[
\innerprod{\mathcal{C}_{h}g_1}{g_2}
= \mean\big[\innerprod{g_1}{X_0-m}\,\innerprod{X_h-m}{g_2}\big], \qquad
g_1,g_2 \in H,\]
and belongs to $S_2(H)$. Since $S_2(H)$ is isomorphic to the tensor product, we call $\mathcal{C}_h$ also {\em autocovariance tensor}. The covariance operator $\mathcal{C}_{h}$ can alternatively be described by its kernel function $c_h$ satisfying 
\begin{align*}
\innerprod{\mathcal{C}_h\,g_1}{g_2}
=\int_0^1 \int_{0}^{1}
c_h(\tau,\sigma)\,g_1(\sigma)\,g_2(\tau)\,d\sigma\,d\tau,
\qquad g_1,g_2\in H.
\end{align*}
In analogy to weakly stationary multivariate time series, where the covariance matrix and spectral density matrix form a Fourier pair, the {\em spectral density operator} or {\em tensor} $\mathcal{F}_{\omega}$ is given by the Fourier transform of $\mathcal{C}_h$,
\begin{align}
\label{eq:Fomega}
\mathcal{F}_{\omega}
= \frac{1}{2\pi}\lsum_{h \in \mathbb{Z}} \mathcal{C}_h\,e^{-\im\omega h}.
\end{align}
A sufficient condition for the existence of $\mathcal{F}_{\omega}$ in $S_p(H_\cnum)$ is $\sum_{h\in\znum} \snorm{\mathcal{C}_h}_p < \infty$. 
Since the setting of this paper allows for higher order dependence among the functional observations, we also require the notion of higher order cumulant tensors. The necessary derivations and definitions are given in Appendix \ref{cumprops}. Throughout the remainder of this paper, time points in $\{1,\ldots,T\}$ will be denoted by $t,s$ or $r$, while rescaled time points on the interval $[0,1]$ will be given by $u$ and $v$. Additionally, angular frequencies are indicated with $\lambda, \alpha, \beta$ or $\omega$ and functional arguments are denoted by $\tau, \sigma$.  

Finally, we require the notion of stochastic integrals with respect to operator-valued functions. To this end, let $\mathcal{B}_{\infty}$ denote the Bochner space $\mathcal{B}_{\infty}=L^2_{S_{\infty}(H_\cnum)}([-\pi,\pi],\mu)$ of all strongly measurable functions $U:[-\pi,\pi]\to S_{\infty}(H_\cnum)$ such that 
\[
\|U\|^2_{\mathcal{B}_{\infty}} =\int_{-\pi}^{\pi} \snorm{U_{\omega}}^2_{\infty} d\mu(\omega)<\infty,
\]
where $\mu$ is a measure on the interval $[-\pi,\pi]$ given by $\mu(A)=\int_{A}\snorm{\mathcal{F}_{\omega}}_1\,d\omega$ 
for all Borel sets $A\subseteq[-\pi,\pi]$. The subspace $\mathcal{B}_{2}$ is then defined similarly with Hilbert-Schmidt norms replacing the operator norms. We distinguish explicitly between the two spaces as the latter space allows for stronger results to be obtained but excludes interesting processes such as functional autoregressive processes.

\subsection{Assumptions}
\textcolor{black}{
In this section, we collect for better reference the assumptions required in subsequent sections. We start by the main assumptions needed for a frequency domain characterization of local stationarity. In contrast to \citet{Panar2013b}, who only consider transfer functions in $\mathcal{B}_{2}$, we also prove the more general case of transfer functions in $\mathcal{B}_{\infty}$ as it includes the important case of functional autoregressive processes. The necessary results are proved in section \ref{proofStochInt} of the Appendix. Throughout the assumptions and the paper, we refer to $S_p(H_\cnum)$ and $\mathcal{B}_{p}$ with $p=2$ or $p=\infty$ to make the distinction between the two cases.}
\textcolor{black}{
\begin{Alist}{A}
\item
(i) $\{\veps_t\}_{t\in\znum}$ is a weakly stationary white noise process taking values in $H$ with spectral representation $\veps_t=\int_{-\pi}^{\pi} e^{\im\omega t}\,dZ_{\omega}$, where $Z_{\omega}$ is a $2\pi$-periodic orthogonal increment process taking values in $H_\cnum$;\\
(ii) the functional process $\XT{t}$ with $t=1,\ldots,T$ and $T\in\nnum$ is given by
\[
\XT{t} = \int_{-\pi}^{\pi} e^{\mathrm{i}\omega t}\, \AAtT{\omega}\,dZ_{\omega} \quad \text{a.e. in } \Hspace 
\]
with transfer operator $\AAtT{\omega}\in\mathcal{B}_p$ and an orthogonal increment process $Z_{\omega}$.
\item\label{A-A}
There exists $\mathcal{A}:[0,1]\times[-\pi,\pi]\to S_p(H_\cnum)$ with $\AAu{\cdot}\in\mathcal{B}_p$ and $\AAu{\omega}$ being continuous in $u$ such that for all $T\in\nnum$
\[
\sup_{\omega,t}
\bigsnorm{\AAtT{\omega}-\AAuu{\frac{t}{T}}{\omega}}_p
=O\big(\tfrac{1}{T}\big).
\]
\item
The function $\mathcal{A}_{\cdot,\cdot}$ is H{\"o}lder continuous of order $\alpha > 1/2$ in $u$ and $\omega$.
\item
The function $\mathcal{A}_{\cdot,\cdot}$ is twice continuously differentiable in $u$ and $\omega$ with second derivatives being uniformly bounded in $u$ and $\omega$.
\savecounter{alphcount}
\end{Alist}
We note that a functional Cram{\'e}r representation such as in (A1)(ii) can also be obtained when the spectral density operator is not well-defined; we refer the reader to  \citet{vde2017}, in which a functional version of Herglotz Theorem is proved and frequency domain representations for stationary time series on the function space are further generalized. For the derivation of asymptotic results for kernel estimators of the spectral density operator, we require additional assumptions on the moments of $k$--order of the process $\veps_t$ in (A1). The following assumption will be imposed for $k\leq 4$ or $k<\infty$.
\begin{Alist}{A}
\restorecounter{alphcount}
\item
\label{itm:cumWN}
The process $\{\varepsilon_t\}_{t\in\znum}$ satisfies $\mean\norm{\varepsilon_0}_2^k<\infty$ and $\sum_{t_1,\ldots,t_{k-1}=-\infty}^{\infty}\snorm{\cumop{k}^{\varepsilon}}_2<\infty$.
\savecounter{alphcount}
\end{Alist}
Finally, the following assumptions formulate the conditions imposed on the taper functions, the kernel functions, and the bandwidths used for kernel smoothing.
\begin{Alist}{A}
\restorecounter{alphcount}
\item
The function $h : \rnum \to \rnum^+$ is symmetric with compact support on $[0,1]$ and is of bounded variation.
\item
The function $\Kf:\rnum\to\rnum^+$ is symmetric, has bounded variation and compact support $[-1,1]$, and satisfies
\begin{romanlist} 
\item
$\DS\int_\rnum \Kf(\omega)\,d\omega=1$;
\item
$\DS\int_\rnum \omega\,\Kf(\omega)\,d\omega=0$.
\end{romanlist}
\item
The sequences $\bfT$ and $\btT$ satisfy for $T\to\infty$ (i) $\bfT \to 0$ and $\btT \to 0$; (ii) $\bfT\,\btT\,T \to \infty$; (iii) $\bfT\,\log(\btT\,T) \to 0$; and (iv) $\btT^2\,\bfT \to 0$.
\savecounter{alphcount}
\end{Alist}}

\subsection{Local stationarity in the frequency domain}

The original definition of local stationary processes by \citet{Dahlhaus1996a} has been formulated in the frequency domain. The following proposition can be viewed a generalization of \citet{Dahlhaus1996a} to the functional setting. 

\begin{proposition}\label{FTVspec}
Suppose that assumptions (A1) and (A2) hold. Then $\{\XT{t}\}$ is a locally stationary process in $H$.
\end{proposition}

\begin{proof}
For $u\in[0,1]$, we define the approximating stationary functional process $\{\Xu{t}\}_{t\in\znum}$ by
\[
\Xu{t}=\int_{-\pi}^{\pi} e^{\im\omega t}\,\AAu{\omega}\,dZ_{\omega}.
\]
Then we have
\[
\bignorm{\XT{t}-\Xu{t}}_{2}
=\Bignorm{\int_{-\pi}^{\pi} e^{\im\omega t}\,
\big(\AAtT{\omega}-\AAu{\omega}\big)\,dZ_{\omega}}_2
\leq c\,P^{(u)}_{t,T}
\]
with
\begin{align*}
c &= \sup_{\omega}\bigsnorm{\AAtT{\omega}-\AAu{\omega}}_\infty\\
&\leq\sup_{\omega}\bigsnorm{\AAtT{\omega}-\AAuu{t/T}{\omega}}_\infty
   +\sup_{\omega}\bigsnorm{\AAuu{t/T}{\omega}-\AAu{\omega}}_\infty
=O\big(\tfrac{t}{T}+\big|\tfrac{t}{T}-u\big|\big)
\end{align*}
and
\[
P^{(u)}_{t,T}
=\frac{1}{c}\,\biggnorm{\int_{-\pi}^{\pi}e^{\im\omega t}\,
\big(\AAtT{\omega}-\AAu{\omega}\big)\,dZ_{\omega}}_2.
\]
Since
\[
\mean|P^{(u)}_{t,T}|^2
\leq\frac{1}{c^2}\,\int_{-\pi}^{\pi}
\bigsnorm{\AAtT{\omega}-\AAu{\omega}}^2_\infty\,
\bigsnorm{\mathcal{F}_\omega}_1\,d\omega
\leq\int_{-\pi}^{\pi}\bigsnorm{\mathcal{F}_\omega}_1\,d\omega,
\]
the process satisfies the conditions of Definition \ref{FLS} with $\rho=2$. 
\end{proof}
As in the time series setting, we need the existence of a {\em transfer operator} $\AAu{\omega}$ that is continuous in $u\in[0,1]$ to guarantee locally an approximately stationary behavior without sudden changes. In order to include interesting cases such as autoregressive processes for which a time-varying functional spectral representation with a common continuous transfer operator $\AAu{\omega}$ does not exist, we require that such a representation only holds approximately by condition (A2). 

The previous result leads us to consider time-varying processes of the form 
\begin{equation}
\label{eq:LP} 
\XT{t}=\lsum_{s \in \mathbb{Z}}\AtT{s} \varepsilon_{t-s},
\end{equation}
where $\{\veps_s\}_{s\in\znum}$ is a weakly stationary functional white noise process in $H$ and $\{\AtT{s}\}_{s \in\znum}$ are sequences of linear operators for $t=1,\ldots,T$ and $T\in\nnum$. The following result states the conditions under which such a process satisfies condition (A1).

\begin{proposition}
\label{tvMAbochnerinf}
Suppose that $\{\veps_t\}_{t \in \znum}$ satisfies assumption (A1) and, for $p=2$ or $p=\infty$, let $\{\AtT{s}\}_{s\in\znum}$ be a sequence of operators in $S_{p}(H)$ satisfying $\sum_{s}\snorm{\AtT{s}}_{p}<\infty$ for all $t=1,\ldots,T$ and $T \in\nnum$. Then the process
\begin{equation}
\XT{t}=\lsum_{s\in\znum}\AtT{s}\,\veps_{t-s}
\end{equation}
satisfies assumption (A1) with $\AAtT{\omega}\in \mathcal{B}_{p}$.
\end{proposition}

The proof is relegated to section \ref{proofssection2} of the Appendix. For $p=2$, the proposition yields a time-varying version of the corresponding result of \citet{Panar2013b}. The more general case $p=\infty$ also includes linear models introduced by \citet{Bosq} and \citet{Hormann2010} as well as the important class of time-varying functional autoregressive processes, which we discuss in detail in the next section.

\textcolor{black}{\begin{remark}
For processes of the form \eqref{eq:LP}, we can alternatively verify Definition \ref{FLS} in the time domain provided we impose some regularity conditions on the decay of the sequence of filter operators $\{\AtT{s}\}_{s\in\znum}$. For example, sufficient conditions for (A1)-(A2) to be satisfied would be to assume that there exists a positive monotonically decreasing sequence  $\{{\ell(s)}\}_{s \in \mathbb{Z}}$ that satisfies $\sum_{s \in \mathbb{Z}} |s|\ell(s)< \infty$ such that $\sup_{t,T}\snorm{\AtT{s}}_{p} < K \ell(s)$ and that there exists a sequence $\{A_{s,u}\}_{s \in \mathbb{N}}$ that satisfies
\[\sup_{t,T}\snorm{\AtT{s}-A_{s,u}}_{p} \le \frac{K \ell(s)}{T},\]
for some constant $K$ independent of $T$. The local asymptotic theory derived later in this paper relies however on additional smoothness conditions of the approximate transfer operators such as condition (A4). These could then be replaced by $\sup_{u}\snorm{\frac{\partial^2}{\partial u^2} A_{s,u}}_{p} <K \ell(s)$, where $\frac{\partial^2}{\partial u^2} A_{s,u}$ denotes the second-order derivative of the function $u \mapsto A_{s,u}$. Depending on the application, different conditions could be considered. The investigation of necessary restrictions on the time domain filter operators are beyond the scope of this paper and are left for future work. 
\end{remark}}
\section{Locally stationary functional autoregressive processes}\label{section4}

Due to its flexibility as well as its simplicity, functional autoregressive processes have been found useful in numerous applications such as economics and medicine, especially for prediction purposes \citep[see e.g.,][for early work]{Damon1982,Besse1986,Antoniadis2003}. Despite of being linear in the function space, the filter operators act on a Hilbert space of which the elements can still exhibit arbitrary degrees of nonlinearity and can therefore be seen to be highly nonlinear in terms of scalar records. Most estimation techniques are however still based on the assumption of \IID functional errors. This assumption has been relaxed by \citet{Bosq}, where the assumption of independence of the errors of the causal solution is relaxed to uncorrelatedness in an appropriate  sense, and by \citet{Hormann2010} for functional AR(1) processes within the framework of $L^p$-$m$-approximability.
 
In this section, we introduce a class of time-varying functional autoregressive processes for which inference and forecasting methods can be developed in a meaningful way. More specifically, we will show that time-varying functional autoregressive processes as well as the more general time-varying functional ARMA processes are locally stationary and that stationary functional ARMA($m$,$n$) processes are a special case. For this we first need to establish that a causal solution exists for time-varying functional AR($m$) processes. This is done in the theorem stated below. 

\begin{theorem} \label{Sol}
Let $\{ \varepsilon_t\}_{t \in \mathbb{Z}}$ be a white noise process in $H$ and let $\{\XT{t}\}$ be a sequence of time-varying functional AR($m$) given by
\begin{equation}
\label{implicitFAR}
\XT{t} = \lsum_{j=1}^{m} B_{\frac{t}{T},j}(\XT{t-j})+\varepsilon_{t}
\end{equation}
with $B_{u,j}=B_{0,j}$ for $u<0$ and $B_{u,j}=B_{1,j}$ for $u>1$. Furthermore, suppose that
\begin{romanlist}
\item
the operators $B_{u,j}$ are continuous in $u \in [0,1] $ for all $j= 1,\ldots,m$; 
\item
for all $u \in [0,1]$, the operators satisfy
$\sum_{j=1}^{m}\snorm{B_{u,j} }_{\infty} < 1$. 
\end{romanlist}
Then \eqref{implicitFAR} has a unique causal solution of the form
\begin{align} \label{eq:MArep}
\XT{t} = \lsum_{l=0}^{\infty} \AtT{l}(\varepsilon_{t-l})
\end{align}
for all $t\in\nnum$ with $\sup_{t,T}\sum_{l=0}^{\infty}\bigsnorm{ \AtT{l}}_{\infty}<\infty$.
\end{theorem}

In order to prove the theorem, note that we can represent the functional AR($m$) process in state space form
\begin{equation} \label{eq:blockfar}
 \underbrace{\begin{pmatrix}
  \XT{t}  \\
  \XT{t-1}  \\
  \vdots \\
  \XT{t-m+1}
 \end{pmatrix}}_{\boldsymbol{X}^{*}_{t,T}}=
 \underbrace{\begin{pmatrix}
  B_{\frac{t}{T},1} & B_{\frac{t}{T},2}  & \cdots & B_{\frac{t}{T},m}  \\
 I_H & & &  O_H \\
&  \ddots & & \vdots  \\
& & I_H  & O_H
 \end{pmatrix}}_{\boldsymbol{B}^{*}_{\frac{t}{T}}} \underbrace{\begin{pmatrix}
  \XT{t-1}  \\
  \XT{t-2}  \\
  \vdots \\
  \XT{t-m}
 \end{pmatrix}}_{\boldsymbol{X}^{*}_{t-1,T}}+ \underbrace{\begin{pmatrix}
  \varepsilon_{t}  \\
  O_H \\
\vdots \\
  O_H
 \end{pmatrix}}_{\boldsymbol{\varepsilon}^{*}_t}.
\end{equation}
Here, $\boldsymbol{X}^{*}_{t,T}$ is a $m$-dimensional random vector taking values in the Hilbert space $H^m$ with inner product $\innerprod{x}{y}=\sum_{i=1}^{m}\innerprod{x_i}{y_i}_H$. Furthermore, $\boldsymbol{B}^{*}_{u}$ denotes a matrix of operators and thus is itself an operator on $H^m$. Consequently, we can write the functional AR($m$) process more compactly as
\[
\boldsymbol{X}^{*}_{t,T}=\boldsymbol{B}_{\frac{t}{T}}^{*}(\boldsymbol{X}^{*}_{t-1,T})+\boldsymbol{\varepsilon}^{*}_t
\]
with $\boldsymbol{\varepsilon}^{*}_t \in L^2_{H^m}(\Omega,\prob)$. 

\begin{proof}[Proof of Theorem \ref{Sol}] 
 In order to show that a causal solution exists in the locally stationary setting, we require the following result which is proved in Appendix  \ref{proofsection3}.
\begin{lemma}
\label{oplem}
For $u \in [0,1]$, the assumption $\sum_{j=1}^{m}\snorm{B_{u,j} }_{\infty} < 1$ implies that the operator $\boldsymbol{B}^{*}_{u}$ satisfies $\snorm{{\boldsymbol{B}^{*}_u}^{k_o} }_{\infty} < 1$  for some $k_o \ge 1, k_o \in \mathbb{Z}$.
\end{lemma}
We note that this is a weaker assumption than $\snorm{ \boldsymbol{B}^{*}_{u}}_{\infty} <1$. Although $\snorm{ \boldsymbol{B}^{*k_0}_{u}}_{\infty} <1$ is usually stated as the condition for a causal solution in the stationary case, the condition $\sum_{j=1}^{m}\snorm{B_{u,j} }_{\infty} < 1$ is easier to check in practice. As a consequence of this lemma, the assumption $\sum_{j=1}^{m}\snorm{B_{u,j} }_{\infty} < 1$ for all $u\in[0,1]$ implies that the spectral radius of $\boldsymbol{B}^{*}_{u}$ satisfies \begin{equation}
\label{eq:srpaper}
r(\boldsymbol{B}^{*}_{u})
=\sup_{\lambda\in S_u}|\lambda|
= \lim_{k \to\infty}\bigsnorm{\boldsymbol{B}^{*k}_{u}}^{1/k}_{\infty}
< \frac{1}{1+\delta}
\end{equation}
for some $\delta >0$. Observe then that by recursive substitution
\begin{align*}
\boldsymbol{X}^{*}_{t,T} = \sum_{l=0}^{\infty} \Big(\lprod_{s=0}^{l-1}\boldsymbol{B}^{*}_{\frac{t-s}{T}}\Big)\boldsymbol{\varepsilon}^{*}_{t-l}.
\end{align*}
From \eqref{eq:blockfar}, this implies a solution is given by 
\begin{align} \label{eq:solFAR}
\XT{t}& = \sum_{l=0}^{\infty} \Big[\lprod_{s=0}^{l-1} \boldsymbol{B}^{*}_{\frac{t-s}{T}}\Big]_{1,1}(\varepsilon_{t-l}),
\end{align}
where $[\cdot]_{1,1}$ refers to the upper left block element of the corresponding block matrix of operators. In order to prove the theorem we shall proceed in a similar manner as \citet{Kunsch} and derive that 
\begin{align*} 
\underset{t,T}{\text{sup}} \Bigsnorm{ \Big[\prod_{s=0}^{l-1} \boldsymbol{B}^{*}_{\frac{t-s}{T}}\big]_{1,1}}_{\infty} < c \rho^l
\end{align*}
for some constant $c$ and $\rho <1$. For this, we require the following lemma.
\begin{lemma} \label{murphy}
Let $B(H)$ be the algebra of bounded linear operators on a Hilbert space. Then for each $A \in B(H)$ and each $\varepsilon >0 $, there exists an invertible element $M$ of $B(H)$ such that $r(A) \le \snorm{MAM^{-1}}_{\infty} \le r(A)+ \varepsilon$. 
\end{lemma}
Since $B(H)$ forms a unital $C^{*}$-algebra, this lemma is a direct consequence of a result in \citet{Murphy}[p.74].
Lemma \ref{murphy} together with (\ref{eq:srpaper}) imply we can specify for fixed $u$ a new operator $M(u) \in B(H)$ such that 
\begin{equation*}
 \snorm{M(u) \boldsymbol{B}^{*}_{u} M^{-1}(u)}_{\infty} < \frac{1}{1+\delta/2}.
\end{equation*} 
Because of the continuity of the autoregressive operators in $u$, we have that for all $u \in [0,1]$, there exists a neighborhood $\mathcal{V}(u)$ such that
\begin{equation*}
 \snorm{M(u) \boldsymbol{B}^{*}_{v} M^{-1}(u)}_{\infty} < \frac{1}{1+\delta/3} < 1 \quad \text{for} \quad v \in \mathcal{V}(u).
\end{equation*}
Define now the finite union $\lcup^{r}_{i = 1} \mathcal{V}(u_i)$ with $\mathcal{V}(u_i) \cap \mathcal{V}(u_l) = \varnothing$ for $i \neq l$. Due to compactness and the fact that $\boldsymbol{B}^{*}_{u} =\boldsymbol{B}^{*}_0$ for $u \le 0$ this union forms a cover of $(-\infty, 1]$. The preceding then implies that there exists a constant $c$ such that 
\begin{equation*}
\snorm{\boldsymbol{B}^{*}_v}_{\infty} \le c \snorm{M(u_i) \boldsymbol{B}^{*}_v M^{-1}(u_i)}_{\infty}, \quad i =1,\ldots,r.
\end{equation*}
Now, fix $t$ and $T$ and define the set $J_{i,l} = \{s \ge 0: \frac{t-s}{T} \in \mathcal{V}(u_i)\} \cap \{0,1,\ldots,l-1\}$. Then specify $\rho = \frac{1}{1+\delta/3}$ to obtain 
\begin{align*}
\Bigsnorm{ \big(\lprod_{s=0}^{l-1} \boldsymbol{B}^{*}_{\frac{t-s}{T}}\big)_{1,1} }_\infty
&\leq\Bigsnorm{\lprod_{s=0}^{l-1} \boldsymbol{B}^{*}_{\frac{t-s}{T}} }_\infty 
\le\lprod_{i =1}^{r} \Bigsnorm{ \lprod_{s \in J_{i,l}} \boldsymbol{B}^{*}_{\frac{t-s}{T}} }_\infty\\
& \le c^m \lprod_{i =1}^{r} \lprod_{s \in J_{i,l}}  \bigsnorm{M(u_i) \boldsymbol{B}^{*}_{\frac{t-s}{T}}M^{-1}(u_i) }_\infty\\
&\le c^r  \lprod_{i =1}^{r} \rho^{|J_{i,l}|} = c^r \rho^{l}, 
\end{align*}
which gives the result.
\end{proof}

\textcolor{black}{Theorem \ref{Sol} will be used to show that time-varying functional ARMA models for which a functional spectral representation exists satisfy conditions (A1) and (A2) and hence by Proposition \ref{FTVspec} are locally stationary. Before we can consider general time-varying functional ARMA models we first need the following result, which shows that for time-varying functional autoregressive processes there exists a common continuous transfer operator $\AAu{\omega}$ that satisfies condition (A2). }

\begin{theorem} \label{FAR} 
Let $\{\varepsilon_t\}_{t \in \znum}$ be a white noise process in $\Hspace$ and let $\{\XT{t}\}$ be a sequence of functional autoregressive processes given by 
\begin{align} \label{eq:DE}
\sum_{j=0}^{m} B_{\frac{t}{T},j}(\XT{t-j})= C_{\frac{t}{T}}(\varepsilon_{t})
\end{align}
with $B_{u,j} = B_{0,j}$, $C_{u} = C_{0}$ for $u<0$ and $B_{u,j} = B_{1, j}$, $C_{u} = C_{1}$ for $u>1$. If the process satisfies, for all $u\in[0,1]$ and $p=2$ or $p=\infty$, the conditions
\begin{romanlist}
\item
$C_u$ is an invertible element of $S_{\infty}(H)$;
\item
$B_{u,j} \in S_{p}(H)$  for $j=1,\ldots, m$ with $\sum_{j=1}^{m} \snorm{  B_{u,j}}_l< 1$ and $B_{u,0}=I_H$;
\item
the mappings $u \mapsto B_{u,j}$ for $j= 1,..,m$ and $u\mapsto C_u, $ are continuous in $u \in [0,1]$ and differentiable on $u \in (0,1)$ with bounded derivatives,
\end{romanlist}
then the process $\{\XT{t}\}$ satisfies (A2) with
\begin{align}\label{eq:smoothf}
\AAttT{\frac{t}{T}}{\omega} = \frac{1}{\sqrt{2 \pi}}\,\bigg( \sum_{j=0}^{m} e^{-\im\omega j}\,B_{\frac{t}{T},j}\bigg)^{-1}\,C_{\frac{t}{T}}
\end{align}
and thus is locally stationary.
\end{theorem}

The proof of Theorem  \ref{FAR} is relegated to Appendix  \ref{proofsection3}. As shown in Theorem \ref{Sol}, a sufficient condition for the difference equation  \eqref{eq:DE} to have a causal solution is $\sum_{j=1}^m\snorm{B_{u,j} }_{\infty} < 1$ or $\snorm{\boldsymbol{B}^{*k_0}_u}_{\infty}<1$ for some $k_0 \ge 1$. The moving average operators will then satisfy $\sum_{l=0}^{\infty} \snorm{\AtT{l}}_{\infty} < \infty$, and Proposition \ref{tvMAbochnerinf} shows that $\XT{t}$ satisfies assumption (A1) with $\AAtT{\omega}  \in \mathcal{B}_{\infty}$. It follows from \eqref{eq:solFAR} that time-varying functional AR($m$) processes that have a causal solution with moving average operators satisfying $\sum_{l=0}^{\infty}\snorm{\AtT{l}}_{2} < \infty$ do not exist. Instead we need at least $\AtT{0}$ to be an invertible element of $S_{\infty}(H)$ together with $\sum_{j=1}^{m}\snorm{B_{u,j}}_2 <1$. By Proposition \ref{Neumann}, this case is covered by Proposition \ref{FTVspec} with $\AAuu{\frac{t}{T}}{\omega} \in S_2(H_\cnum)$ in condition (A2). For stationary functional AR($m$) processes this is straightforward to verify using back-shift operator notation and by solving for the inverse of the autoregressive lag operator. Under slightly more restrictive assumptions it is possible to obtain  uniform convergence results for processes with transfer operators $\AAtT{\omega} \in \mathcal{B}_2$. We will come back to this in sections \ref{section3} and \ref{expandcov}, in which we consider capturing the changing second-order dependence structure via the time-varying spectral density operator.  

Using Theorem \ref{FAR}, it is now straightforward to establish that the time-varying functional ARMA processes are locally stationary in the sense of Proposition \ref{FTVspec}. A time-varying functional moving average process of order $n$ has transfer operator 
\begin{align*}
\AAtT{\omega} = \frac{1}{\sqrt{2 \pi}} \sum_{j=0}^{n} \Phi_{\frac{t}{T},j}\,e^{-\im\omega j},
\end{align*} 
where $\Phi_{t/T,j} \in S_p(H)$ are the moving average filter operators. This follows from the spectral representation of the $\varepsilon_t$. 
Setting $\AAuu{\frac{t}{T}}{\omega}= \AAtT{\omega}$ gives the result. Finally, we can combine this with Theorem \ref{FAR}, to obtain that conditions (A1) and (A2) are satisfied for time-varying functional ARMA($m$,$n$) with common continuous transfer operator given by 
\begin{align}
\AAuu{\frac{t}{T}}{\omega} = \frac{1}{\sqrt{2 \pi}}\,C_{\frac{t}{T}}\,
\bigg( \sum_{j=0}^{m} e^{-\im\omega j}\, B_{\frac{t}{T},j}\bigg)^{-1}   \sum_{l=0}^{n} \Phi_{\frac{t}{T},l}\,e^{-\im\omega l}.
\end{align}

For operators that do not depend on $t$, this result proves the existence of a well-defined functional  Cram{\'e}r representation for weakly stationary functional ARMA($m$,$n$) processes as discussed in \citet{Bosq} or as in \citet{Hormann2010}. The latter is easily seen by means of an application of the dominated convergence theorem and by defining the $m$-dependent coupling process by
\begin{align*}
X^{(m)}_{t,T}=g_{t,T}(\varepsilon_{t},\ldots,\varepsilon_{t-m+1},\varepsilon^{*}_{t-m}, \varepsilon^{(*)}_{t-m-1},\ldots)
\end{align*}
for measurable functions $g_{t,T}:{H}^{\infty}\to H$ with $t=1,\ldots,T$ and $T\in\nnum$ and where $\{\varepsilon^{*}_{t}\}$ is an independent copy of $\{\varepsilon_{t}\}$.

\section{Time-varying spectral density operator}
\label{section3}

We will now introduce the time-varying spectral density operator and its properties. We will show that the uniqueness property of the time-varying spectral density established by \citet{Dahlhaus1996a} also extends to the infinite dimension. Let $\XT{t}$ satisfy conditions (A1) and (A2) with $\AAtT{\omega} = \AAttT{1}{\omega} $ for $t<1$ and $\AAtT{\omega}=\AAttT{T}{\omega}$ for $t>T$. We define the {\em local autocovariance operator} as the cumulant tensor 
\begin{align} \label{eq:rtvop}
\mathcal{C}^{(T)}_{u,s} = \cov(X_{\lfloor uT-s/2\rfloor,T},X_{\lfloor uT+s/2\rfloor,T}),
\end{align}
where $\lfloor s \rfloor$ denotes the largest integer not greater than $s$. This operator belongs to $S_2(H)$ and hence has a {\em local autocovariance kernel} $\cT_{u,s} \in L^2([0,1]^2)$ given by 
\begin{align}
\label{eq:rtv}
\innerprod{\mathcal{C}^{(T)}_{u,s} g_1}{g_2} = \int \int \cT_{u,s}(\tau,\sigma) g_1(\sigma) \overline{g_2(\tau)} d\sigma d\tau
\qquad g_1,g_2 \in H. \end{align}
\begin{proposition} \label{propl2rtv}
Suppose (A1) and (A2) are satisfied. Then the local autocovariance operator defined in \eqref{eq:rtvop} satisfies $\sum_{s\in\znum}\snorm{\mathcal{C}^{(T)}_{u,s} }_{2} < \infty$.
\end{proposition}

The proof can be found in section \ref{proofsection4} of the Appendix. Proposition \ref{propl2rtv} implies that the Fourier transform of \eqref{eq:rtvop} is a well-defined element of $S_2(H_{\cnum})$ and is given by 
\begin{align}
\label{eq:SDOdef}
\FT_{u,\omega}
=\SSS{\frac{1}{2\pi}} \lsum_{s}\CCT_{u,s}e^{-\im\omega s}.
\end{align}
For fixed $T$, this operator can be seen as a functional generalization of the Wigner-Ville spectral density matrix \citep[][]{Martin1985} and we shall therefore refer to it as the {\em Wigner-Ville spectral density operator} $\FT_{u,\omega} $. It is easily shown that the Fourier transform of the autocovariance kernel $\cT_{u,s}$, for fixed $t$ and $T$, forms a Fourier pair in $L^2$ with the kernel of $\FT_{u,\omega}$, referred to as the  {\em Wigner-Ville spectral density kernel}
\begin{align}
\label{eq:wigvilsdk}
\fT_{u,\omega}(\tau,\sigma)
= \frac{1}{2\pi} \sum_{s \in \znum}\cT_{u,s}(\tau, \sigma) e^{-\im\omega s}.
\end{align}
More specifically, given $\sum_{s\in\znum}\norm{\cT_{u,s}}_{p}< \infty$ for $p=2$ or $p=\infty$, the spectral density kernel is uniformly bounded and uniformly continuous in $\omega$ with respect to $\| \cdot \|_p$. Additionally, the inversion formula \begin{align} \label{eq:invform}
\cT_{u,s}(\tau,\sigma)
=\int_{-\pi}^{\pi}\fT_{u,\omega}(\tau,\sigma)\,e^{\im s\omega}\,d\omega
\end{align}
holds in $\norm{\cdot}_p$ for all $s$, $u$, $T$, $\tau$, and $\sigma$. This formula and its extension to higher order cumulant kernels is direct from an application of the dominated convergence theorem. Under additional assumptions, certain results presented in this paper will hold uniformly rather than in mean square. Sufficient would be to assume that the functional white noise process $\{\veps_t\}_{t\in\znum}$ is mean square continuous 
and that the sequence of operators $\{\AtT{s}\}_{s \in \znum}$ is Hilbert-Schmidt with continuous kernels for all $t=1,\ldots,T$ and $T \in \nnum$. The process $\{\XT{t}\}$ is then itself mean square continuous and a slight adjustment in the proof of Proposition \ref{propl2rtv} demonstrates that $\sum_{s \in \znum}\bignorm{\cT_{u,s}}_{\infty}<\infty$. This is also expected to be the case for yet-to-be-developed concepts such as a time-varying Cram{\'e}r-Karhunen-Lo{\`e}ve representation. The results obtained in the previous section however demonstrate that a representation under these stronger conditions excludes time-varying functional AR($m$) models. We will therefore not impose them but merely remark where stronger results could be obtained.

The pointwise interpretation of the $L^2$-kernels makes it easy to verify that the Wigner-Ville spectral operator $\FT_{u,\omega}$ is $2 \pi$-periodic in $\omega$ and self-adjoint. Namely, $\cT_{u,-s}(\sigma,\tau)=\cT_{u,s}(\tau,\sigma)$ implies ${f^{\dagger}}^{(T)} _{u,\omega} (\sigma,\tau)= \overline{\fT_{u,\omega}(\tau,\sigma)}$, where $f^{\dagger}$ the kernel function of the adjoint operator $\F^{\dagger}$. Moreover, $\F^{\varepsilon}_{\omega}$ is trace-class by Parseval's identity and therefore Proposition \ref{holderoperator} implies that \eqref{eq:SDOdef} is actually an element of $S_1(H_\cnum)$. We will show in the following that \eqref{eq:SDOdef} converges in integrated mean square to the {\em time-varying spectral density operator} defined as
\begin{align} \label{eq:tvSDO}
\F_{u, \omega}
= \AAu{\omega}\,\F^{\varepsilon}_{\omega}\,\AAu{\omega}^\dagger. \displaybreak[0]
\end{align}
The time-varying spectral density operator satisfies all of the above properties and is non-negative definite since for every $\psi \in L^2_\cnum([0,1])$,
\begin{align*}
\innerprod{\AAu{\omega}\,\mathcal{F}^{\varepsilon}_{\omega}\, \AAu{\omega}^{\dagger}\,\psi}{\psi}=
\innerprod{\sqrt{\F^{\varepsilon}_{\omega}}\,\AAu{\omega}^{\dagger}\,\psi}{\sqrt{\F^{\varepsilon}_{\omega}} \AAu{\omega}^{\dagger}\,\psi}\geq 0,
\end{align*}
which is a consequence of the non-negative definiteness of $\F^{\varepsilon}_{\omega}$. For any two elements $\psi, \varphi$ in $L^2_\cnum([0,1])$, one can interpret the mapping $\omega \mapsto \innerprod{\psi}{\F_{u,\omega}\,\varphi}=\innerprod{\F_{u, \omega}\, \psi}{\varphi}\in \cnum $ to be the local cross-spectrum of the sequences $\{ \innerprod{\psi}{\Xu{t}}\}_{t \in \znum}$ and $\{\innerprod{\varphi}{ \Xu{t}}\}_{t \in \znum}$. In particular, $\omega \mapsto \innerprod{\psi}{\F_{u, \omega}\,\psi}\geq 0$ can be interpreted as the local power spectrum of $\{ \innerprod{\psi}{\Xu{t}}\}_{t \in \znum}$ for all $u \in [0,1]$. In analogy to the spectral density matrix in multivariate time series, we will show below that the local spectral density operator completely characterizes the limiting second-order dynamics of the family of functional processes $\{\XT{t}:t =1,\ldots,T\}_{T \in \mathbb{N}}$.

\begin{theorem}
\label{SDO}
Suppose that assumptions (A1) to (A3) hold. Then, for all $u \in (0,1)$,
\begin{align}
\label{eq:SDObound}
\int_{-\pi}^{\pi} \bigsnorm{\FT_{u,\omega} - \F_{u,\omega}}^2_2\,d\omega = o(1) \end{align}
as $T\to \infty$.
\end{theorem}

\begin{proof}
By definition of the Wigner-Ville operator and Lemma \ref{bochnerlemma},
\begin{align*}
\FT_{u,\omega}
&=\frac{1}{2\pi} \sum_{s} \cov\big(\int_{-\pi}^{\pi} e^{\im\lambda \lfloor uT-s/2\rfloor}\,\AAttT{\lfloor uT-s/2\rfloor}{\lambda}\,dZ_{\lambda},
\int_{-\pi}^{\pi} e^{\im\beta \lfloor uT+s/2\rfloor}\,\AAttT{\lfloor uT+s/2\rfloor}{\beta}\,dZ_{\beta}\big)\,e^{-\im\omega s} \\
&=\frac{1}{2\pi} \sum_{s}\int_{-\pi}^{\pi}  e^{\im\lambda s}  \AAttT{\lfloor uT-s/2\rfloor}{\lambda}\,\F^{\varepsilon}_{\lambda}\,\big(\AAttT{\lfloor uT+s/2\rfloor}{\lambda}\big)^\dagger\,d\lambda\,e^{-\im\omega s}.
\end{align*}
Using identity \eqref{eq:tensorabc}, we have that
\[
\AAttT{\lfloor uT-s/2\rfloor}{\lambda}\,\F^{\varepsilon}_{\lambda} \big(\AAttT{\lfloor uT+s/2\rfloor}{\lambda}\big)^\dagger
=\big(\AAttT{\lfloor uT-s/2 \rfloor}{\lambda} \otimes \AAttT{\lfloor uT+s/2 \rfloor}{\lambda} \big)\,\F^{\varepsilon}_{\lambda}.
\]
Similarly, 
\begin{align*}\F_{u,\omega} =
\frac{1}{2\pi} \sum_{s}\int_{-\pi}^{\pi}  e^{\im\lambda s}\,\big(\AAu{\lambda} \otimes  \AAu{\lambda}^{\dagger}\big)\,\F^{\varepsilon}_{\lambda}\,d\lambda\, e^{-\im\omega s}.
\end{align*}
We can therefore write the left-hand side of \eqref{eq:SDObound} as
\begin{align*}
\int_{-\pi}^{\pi} \biggsnorm{\frac{1}{2\pi} \sum_{s}\int_{-\pi}^{\pi}  e^{\im\lambda s}\big(\AAttT{\lfloor uT-s/2 \rfloor}{\lambda}\otimes \AAttT{\lfloor uT+s/2 \rfloor}{\lambda}-\AAu{\lambda}\otimes \AAu{\lambda}\big) \,\F^{\varepsilon}_{\lambda}\, d\lambda\, e^{-\im\omega s}}^2_2\, d\omega. 
\end{align*}
Consider the operator 
\begin{align*}
G^{(u,T)}_{s,\lambda}
= \big(\AAttT{\lfloor uT-s/2 \rfloor}{\lambda}\otimes
\AAttT{\lfloor uT+s/2 \rfloor}{\lambda} -\AAu{\lambda}\otimes \AAu{\lambda}\big)\,\F^{\varepsilon}_{\lambda}
\end{align*}
and its continuous counterpart 
\begin{align*}
G_{\frac{s}{2T},\lambda} = \big(\AAuu{(u-\frac{s}{2T})}{\lambda} \otimes \AAuu{(u+\frac{s}{2T})}{\lambda} -\AAu{\lambda} \otimes \AAu{\lambda}\big)\,\F^{\varepsilon}_{\lambda}. 
\end{align*}
By H{\"o}lder's inequality for operators (Proposition \ref{holderoperator}), both are trace-class and hence Hilbert-Schmidt. Another application of H{\"o}lder's inequality together with assumption (A2) yields
\begin{align*} \label{eq:Thm32_1}
& \Bigsnorm{\big(\AAttT{\lfloor uT\mp s/2 \rfloor}{\lambda}\otimes \big[\AAttT{\lfloor uT \pm s/2 \rfloor}{\lambda}
-\AAuu{u \pm \frac{s}{2T}}{\lambda}\big]\big)
\mathcal{F}^{\varepsilon}_{\lambda} }_2^2  \\
& \le\bigsnorm{\AAttT{\lfloor uT\mp s/2 \rfloor}{\lambda}}^2_{\infty}\, \bigsnorm{\AAttT{\lfloor uT \pm s/2 \rfloor}{\lambda} -\AAuu{u \pm \frac{s}{2T}}{\lambda}}^2_{\infty}\,\bigsnorm{\F^{\veps}_{\omega}}_2^2 = O\Big(\frac{1}{T^2}\Big). \tageq
\end{align*}
By Minkowski's inequality, we obtain
\[
\int_{-\pi}^{\pi} \bigsnorm{\F^{(T)}_{u,\omega} - \F_{u,\omega}}^2_2\,d\omega
= \int_{-\pi}^{\pi} \bigsnorm{G_{\frac{s}{2T},\omega}}^2_2\,d\omega+ o(1). 
\]
It is therefore sufficient to derive a bound on
\[\label{eq:Gint}
\int_{-\pi}^{\pi} \bigsnorm{G_{\frac{s}{2T},\omega}}^2_2\,d\omega. \tageq \]
A similar argument as in \eqref{eq:Thm32_1} shows that
\begin{align*}
& \Bigsnorm{\Big(\AAuu{(u-\frac{s}{2T})}{\omega} \otimes
\big(\AAuu{(u +\frac{s}{2T})}{\omega}
-\AAu{\omega}\big)+\big(\AAuu{(u -\frac{s}{2T})}{\omega}-\AAu{\omega}\big) \otimes \AAu{\omega} \Big)\F^{\varepsilon}_{\omega}}_2 \leq C\,\big|\SSS{\frac{s}{2T}}\big|^{\alpha}
\end{align*}
for some constant $C>0$.  The operator-valued function $G_{u,\omega}$ is therefore H{\"o}lder continuous of order $\alpha > 1/2$ in $u$. \textcolor{black}{Note that \eqref{eq:invform} implies the inverse Fourier transform of this operator is a well-defined element of $S_2(H_\cnum)$. We can therefore write \eqref{eq:Gint} as}
\begin{align*}
& \frac{1}{(2\pi)^2}\int_{-\pi}^{\pi} 
\sum_{s,s'} e^{-\im\omega (s-s')} \biginnerprod{ \int_{0}^{2 \pi}   G_{\frac{s}{2T},\lambda}e^{\im s \lambda}\, d\lambda}{ \int_{0}^{2 \pi}   {G_{\frac{s'}{2T},\lambda'}} e^{\im s' \lambda'}d\lambda' }
\\
&  =  \frac{1}{2\pi}\sum_{s\in\znum}
\snorm{ \tilde{G}_{\frac{s}{2T}}\,}^2_2,
\end{align*}
where 
$\tilde{G}_s$ can be viewed as the $s$-th Fourier coefficient operator of $G_{\frac{s}{2T},\lambda}$. Because of H{\"o}lder continuity, these operators satisfy $\snorm{\tilde{G}_s}_2 \le \|\pi^{\alpha+1} \snorm{G_{u,\omega}}_{2}\,|s|^{-\alpha} = O(s^{-\alpha})$. Hence,
\begin{align*}
\lsum_{s=n}^{\infty}  \snorm{ \tilde{G}_{\frac{|s|}{2T}}\,}^2_2  = O( n^{1-2\alpha}).
\end{align*}
Concerning the partial sum $ \sum_{s=0}^{n-1}  | \hat{g}_s(\tau,\sigma)|^2$, we proceed as in \citet{Dahlhaus1996a} and use summation by parts to obtain
\begin{equation*}
\sum_{s=0}^{n-1}\snorm{ \tilde{G}_s}_2^2 = \int_{0}^{2 \pi}   \int_{0}^{2 \pi}   \sum_{s=0}^{n-1} e^{\im s (\lambda-\lambda')}\biginnerprod{ G_{\frac{s}{2T},\lambda} }{  {G_{\frac{s'}{2T},\lambda'}} } d\lambda d\lambda'
= O\Big(\frac{n\,\log(n) }{T^{\alpha}}\Big),
\end{equation*}
which follow from the properties of $ \tilde{G}_s$ and Lemma \ref{Dh1993A4}. It is straightforward to see that $ \sum_{s=0}^{n-1}  \snorm{  \tilde{G}_{-s}}_2^2$ satisfies the same bound. Hence, 
\begin{align*}
\int_{-\pi}^{\pi} \bigsnorm{\F^{(T)}_{u,\omega} - \F_{u,\omega}}^2_2\,d\omega
&= \int_{-\pi}^{\pi} \bigsnorm{G_{\frac{s}{2T},\omega}}^2_2\,d\omega+ o(1) \\& =  O\Big( n^{1-2\alpha}\Big)+ O\Big(\frac{n \log (n) }{T^{\alpha}}\Big).
\end{align*}
Choosing an appropriate value $n \ll T$ completes the proof.
\end{proof}

Intuitively, the value of $n$ such that $n\log(n)\,T^{-\alpha}\to 0$ can be seen to determine the length of the data-segment over which the observations are approximately stationary. Only those functional observations $\XT{t}$ from the triangular array with $t/T \in \big[u+\frac{n}{T},u-\frac{n}{T}\big]$  will effectively contribute to the time-varying spectral density operator at $u$. As $T$ increases, the width of this interval shrinks and sampling becomes more dense. Because the array shares dynamics through the operator-valued function $\AAu{\omega}$, which is smooth in $u$, the observations belonging to this interval will thus become close to stationary as $T\to\infty$. The theorem therefore implies that, if we have infinitely many observations with the same probabilistic structure around some time point $t$, the local second-order dynamics of the family are completely characterized by $\F_{u,\omega}$.

The above theorem provides a promising result. It is well-known from the time series setting that a Cram{\'e}r representation as given in Proposition \ref{FTVspec} is in general not unique \citep[e.g.][]{Priestley1981}. However, Theorem \ref{SDO} shows that the uniqueness property as proved by \citet{Dahlhaus1996a} generalizes to the functional setting. That is, if the family of $H$-valued processes $\{\XT{t}: t=1,\ldots,T\}_{T\in\nnum}$ has a representation with common transfer operator $\AAu{\omega}$ that operates on $H_\cnum$ and that is continuous in $u$, then the time-varying spectral density operator will be uniquely determined from the triangular array. This uniqueness of the time-varying spectral density operator is expected to be extremely valuable in the development of inference methods. For example, it would be of interest to determine whether this result will allow to develop Quasi Likelihood methods to fit parametric models in the functional setting. Such an extension is not direct and has to take into account the compactness of the operator and the properties of Toeplitz operators in the infinite dimension. This is however beyond the scope of this paper and the authors will consider this in future work.

\begin{remark}\label{SDOUNI} If assumptions (A1) and (A2) hold with $p=2$, we have by continuity of the inner product that the kernel $\aau{\omega}$ of $\AAu{\omega}$ is well-defined in $L^2_{\cnum}([0,1]^2)$ and hence is uniformly H{\"o}lder continuous of order $\alpha > 1/2$ in both $u$ and $\omega$. If we thus additionally assume that the $\{\veps_t\}_{t\in \znum}$ are mean square continuous and the operator $\AAu{\omega}$ is an element of $\mathcal{B}_2$ of which the Hilbert-Schmidt component has a kernel that is continuous in its functional arguments, the error holds in uniform norm.
\end{remark}

\section{Estimation}\label{expandcov}

The time-varying spectral density operator as defined in section \ref{SDO} allows to capture the complete second-order structure of a functional time series with possibly changing dynamics. In order to consider inferential techniques such as dynamic FPCA for nonstationary functional time series, functional Whittle likelihood methods or other general testing procedures, we require a consistent estimator for the time-varying spectral density operator. In this section, we present a nonparametric estimator of the time-varying spectral density operator. It should be noted that this requires a careful consideration of certain concepts on the function space, details of which are relegated to sections \ref{cumprops}-\ref{taperandl} of the Appendix, and that there are some discrepancies compared to existing results available in the Euclidean setting.

The section is structured as follows. First, we define a functional version of the segmented periodogram and derive its mean and covariance structure. We then consider a smoothed version of this operator and show its consistency. Finally, we provide a central limit theorem for the proposed estimator of the time-varying spectral density operator. Proofs of this section can be found in section \ref{proofsection5} of the Appendix.

\subsection{The functional segmented periodogram tensor}

The general idea underlying inference methods in the setting of locally stationary processes is that the process $\XT{t}$ can be considered to be close to some stationary process, say ${X}^{(u_0)}_{t}$, on a reasonably small data-segment around $u_0$. If this segment is described by $\{t:  | \frac{t}{T}-u_0| \le b_{\text{t}}/2\}$ for some bandwidth $b_{\text{t}}$, classical estimation methods from the stationary framework can be applied on this stretch. The estimated value is subsequently assigned to be the value of the parameter curve at the midpoint $u_0$ of the segment. The entire parameter curve of interest in time-direction can then be obtained by shifting the segment. We will also apply this technique in the functional setting.

First, let the length of the stretch considered for estimation be denoted by $N_T$, where $N_T$ is even and $N_T \ll T$. In the following, we will drop the explicit dependence of $N$ on $T$ and simple write $N=N_T$. Then the local version of the {\em functional Discrete Fourier Transform (fDFT)} is defined as 
\begin{equation}
\label{eq:lfdft}
\DT_{u,\omega}=
\lsum_{s=0}^{N-1}  \hN{s}\,\XT{\lfloor u\,T\rfloor-N/2+s+1}\,e^{-\im\omega s},
\end{equation}
where $\hN{s}$ is a data taper of length $N$. It is clear that $\DT_{u,\omega}$ is a $2 \pi$-periodic function in $\omega$ that takes values in $H_\cnum$. The data-taper is used to improve the finite-sample properties of the estimator \citep{Dahlhaus1988}: firstly, it mitigates spectral leakage, which is the transfer of frequency content from large peaks to surrounding areas and is also a problem in the stationary setting. Secondly, it reduces the bias that stems from the degree of nonstationarity of the process on the given data-segment, that is, the fact that we use the observations $\XT{t}$ for estimation rather than the unknown stationary process ${X}^{(u_0)}_{t}$. We define the data-taper by a function $h:[0,1]\to\rnum$ and setting $\hN{s}=h\big(\frac{s}{N}\big)$; the taper function $h$ should decay smoothly to zero at the endpoints of the interval while being essentially equal to $1$ in the central part of the interval. Thus the taper gives more weight to data-points closer to the midpoint.

As a basis for estimation of the time-varying spectral density operator, we consider the normalized tensor product of the local functional Discrete Fourier Transform. This leads to the concept of a segmented or {\em localized periodogram tensor} 
\begin{align*} 
\IT_{u,\omega} =(2\pi\,\HNN{2}(0))^{-1}\,\DT_{u,\omega}\otimes\DT_{u,\omega},
\tageq \label{eq:perop}
\end{align*}
where
\begin{equation} \label{eq:HkN}
\HNN{k}(\omega)=\lsum_{s=0}^{N-1}  \hN{s}^k\, e^{-\im\omega s}
\end{equation} 
is the finite Fourier transform of the $k$-th power of the data-taper. Given the moments are well-defined in $L^2_\cnum([0,1]^2)$,  the corresponding {\em localized periodogram kernel} is given by
\begin{align*} 
\IT_{u,\omega}(\tau,\sigma) =
\big(2\pi\HNN{2}(0)\big)^{-1}\,\DT_{u,\omega}(\tau)\,\DT_{u,-\omega}(\sigma).
\tageq \label{eq:perker}
\end{align*}
Similar to the stationary case, sufficient conditions for the existence of the higher order moments of the localized periodogram tensor are obtained from
\begin{align*}
\norm{\IT_{u,\omega}}^\rho_2
&= (2\pi\,\HNN{2}(0))^{-\rho}\,\bignorm{\DT_{u,\omega}}^{2\rho}_2, \tageq
\end{align*}
which implies that $\mean\norm{\IT_{u,\omega}}^\rho_2 <\infty$ if $\mean \bignorm{\DT_{u,\omega}}^{2\rho}_2<\infty$ or, in terms of moments of $X$, $\mean\bignorm{\XT{t}}^{2\rho}_2<\infty$.

To ease notation, we denote $t_{u,r} = \lfloor uT \rfloor-N/2+r+1$ to be the $r$-th element of the data-segment with midpoint $u$. For $u_j=j/T$ we also write $t_{j,r}=t_{u_j,r}$ and abbreviate $u_{j,r}=t_{j,r}/T$. The following result is used throughout the rest of the paper.

\begin{proposition} \label{cumX}
Suppose that assumption (A1) holds with $\AAtT{\omega} \in \mathcal{B}_{\infty}$, and additionally $\sup_{\omega_1,\ldots, \omega_{k-1}}\snorm{\F^{\varepsilon}_{\omega_1,\ldots, \omega_{k-1}}}_2 <\infty$. Then
\begin{align*} 
\cum  \big(\XT{t_{r_1}},\ldots,\XT{t_{r_k}}\big)
&=\int_{\Pi^{k}}e^{\im(\lam_1 r_1+\ldots+\lam_k r_k)}\,
\Big(\AAttT{t_{r_1}}{\lambda_1}\otimes\cdots\otimes \AAttT{t_{r_{k}}}{\lambda_{k}}\Big)\\
&\qquad\times\eta(\lam_1+\ldots+\lam_k)\,   \F^{\varepsilon}_{\lambda_1,\ldots,\lambda_{k-1}}d\lambda_1\cdots d\lambda_{k}, 
\tageq \label{eq:derivcumX}
\end{align*}
where the equality holds in the tensor product space $H_\cnum\otimes\cdots\otimes H_\cnum$. Moreover, for fixed $t\in\{1,\ldots,T\}$ and $T\in\nnum$, the $k$-th order cumulant spectral tensor of the linear functional process $\{\XT{t}\}$,
\begin{align*} 
\F^{(t,T)}_{\lam_1,..,\lam_{k-1}}
=\Big( \AAttT{t_{r_{1}}}{\lam_1} \otimes\cdots\otimes  \AAttT{t_{r_{k-1}}}{\alpha_{k-1}}\otimes \AAttT{t_{r_{k}}}{-\lam_+}\Big) \F^{\varepsilon}_{\lam_1,..,\lam_{k-1}},
\end{align*}
where $\lam_+=\lam_1+\ldots+\lam_{k-1}$, is well-defined in the tensor product space $\bigotimes_{i=1}^{k} H_\cnum$ with kernel $f^{(t,T)}_{\lambda_1,\ldots,\lambda_{k-1}}(\tau_1,\ldots,\tau_k)$. 
\end{proposition} 

Note that under the stronger condition $\AAtT{\omega} \in \mathcal{B}_{2}$, the tensor $\F^{(t,T)}_{\lambda_1,\ldots,\lambda_{k-1}}$ will be trace-class for all $k \ge 2$. The above proposition implies that the higher order cumulant tensor of the local fDFT can be written as
\begin{align*} 
\cum\big(\DT_{u,\omega_1},&\ldots, \DT_{u,\omega_k} \big)\\
& = \int_{\Pi^{k}}
\Big(\HN\big(\AAttT{t_{u,\bullet}}{\lambda_1},\omega_1-\lambda_1\big)
\otimes\cdots\otimes \HN\big(\AAttT{t_{u,\bullet}}{\lambda_{k}},\omega_{k}-\lambda_{k}\big)\Big) \\
&\qquad\times \eta(\lam_1+\ldots+\lam_k)\,\, \F^{\varepsilon}_{\lambda_1,\ldots,\lambda_{k-1}}\,d\lambda_1\cdots d\lambda_{k}.
\tageq \label{eq:cumDNop}
\end{align*}
Here, the function $\HN(G_{\bullet},\omega)$ and similarly $\HNN{k}(G_{\bullet},\omega)$ generalize the definitions of $\HN$ and  $\HNN{k}$ to
\begin{equation} \label{Hgs}
\HNN{k}(G_{\bullet},\omega)=\lsum_{s=0}^{N-1} \hN{s}^k\,G_{s}\,e^{-\im\omega s}
\end{equation}
with $\HN(G_{\bullet},\omega)=\HNN{1}(G_{\bullet},\omega)$, where in our setting $G_{s}\in\mathcal{B}_{\infty}$ for all $s\in\nnum_0$. For $G_{\bullet}=I_{H_\cnum}$, we get back the original definitions of $\HN$ and  $\HNN{k}$. The convolution property of $\HN$ straightforwardly generalizes to
\begin{equation}
\label{HNfunctconvolution}
\int_{\Pi}
\HNN{k}(A_{\bullet},\alpha+\gamma)\otimes\HNN{l}(B_{\bullet},\beta-\gamma)
\,d\gamma
=2\pi\,\HNN{k+l}(A_{\bullet}\otimes B_{\bullet},\alpha+\beta),
\end{equation}
where $(A_{r})_{r=0,\ldots,N-1}$ and $(B_{r})_{r=0,\ldots,N-1}$ are vectors of tensors or operators.

From the taper function $h$, we derive the smoothing kernel $\Kt$ in rescaled time $u$ by
\begin{equation}
\label{Kerneltimedef}
\Kt(x)=\SSS{\frac{1}{H_2}}\,h\Big(\SSS{x +\frac{1}{2}}\Big)^2
\end{equation}
for $x\in[-\tfrac{1}{2},\tfrac{1}{2}]$ and zero elsewhere; furthermore, we define the bandwidth $\btT=N/T$ that corresponds to segments of length $N$, and set $\KtT(x)=\tfrac{1}{\btT}\,\Kt\big(\tfrac{x}{\btT}\big)$. Finally, we define the kernel-specific constants
\[
\kappat=\int_\rnum x^2\,\Kt(x)\,dx
\qquad\text{and}\qquad
\norm{\Kt}_2^2=\int_\rnum \Kt(x)^2\,dx.
\]

The first order and second order properties of the segmented functional periodogram can now be determined. 

\begin{theorem}
\label{exp}
Suppose that assumptions (A1), (A2), (A4), and (A5) with $k\leq 4$ hold. Then the mean and covariance structure of the local functional periodogram are given by
\begin{align*}
\mean \,\innerprod{\IT_{u,\omega}g_1}{g_2}
=\innerprod{\F_{u,\omega}g_1}{g_2}+\frac{1}{2}\btT^2\,\kappat\,
\SSS{\frac{\partial^2}{\partial u^2}}\innerprod{\F_{u,\omega}g_1}{g_2}
+ o(\btT^2)+O\big(\SSS{\frac{\log(\btT\,T)}{\btT\,T}}\big),
\end{align*}
and
\begin{align*}
\Cov&\big(\innerprod{\IT_{u,\omega_1}}{ g_1 \otimes g_2},\innerprod{\IT_{u,\omega_2}}{ g_3 \otimes g_4}\big)
\\&=\HNN{2}\big(\innerprod{\F_{\frac{t_{u,\bullet}}{T},\omega_1}g_3}{g_1},\omega_1-\omega_2\big)\,
\HNN{2}\big(\innerprod{\F_{\frac{t_{u,\bullet}}{T},-\omega_1}g_4}{g_2},\omega_2-\omega_1\big)\\
&\qquad+\HNN{2}\big(\innerprod{\F_{\frac{t_{u,\bullet}}{T},\omega_1}g_4}{g_1},\omega_1+\omega_2\big)\,
\HNN{2}\big(\innerprod{\F_{\frac{t_{u,\bullet}}{T},-\omega_1}g_3}{g_2},-\omega_1-\omega_2\big)\\
&\qquad +O\big(\SSS{\frac{\log(N)}{N}}\big)+ O\big(\SSS{\frac{1}{N}}\big),
\end{align*}
for all $g_1, g_2, g_3, g_4 \in H_{\cnum}$.
\end{theorem}

The proof exploits assumption (A2) and is based on the theory of $L$-functions \citep[][]{Dahlhaus1983}, which allows to provide upper bound conditions on the data-taper function. Details of the extension of the latter to the functional setting can be found in section \ref{taperandl} of the Appendix. 

\subsection{Consistent estimation}

Theorem \ref{exp} shows that the segmented periodogram tensor is not a consistent estimator. In order to obtain a consistent estimator we proceed by smoothing the estimator over different frequencies. That is, we consider convolving the segmented periodogram tensor with a window function in frequency direction
\begin{align}\label{eq:smoothest}
\hatFT_{u,\omega} = \SSS{\frac{1}{\bfT}} \int_{\Pi} \Kf\Big(\SSS{\frac{\omega-\lambda}{\bfT}}\Big)\,\IT_{u,\lambda}\, d\lambda,
\end{align}
where $\bfT$ denotes the bandwidth in frequency direction. To ease notation, we also write
$\KfT(\omega) =\frac{1}{\bfT} \Kf\big(\frac{\omega}{{\bfT}}\big) $.
Additionally we use subsequently
\[
\kappaf{}=\int_\rnum \omega^2\,\Kf(\omega)\,d\omega\qquad\text{and}\qquad
\norm{\Kf}^2_2=\int_\rnum \Kf^2(\omega)\,d\omega
\]
as an abbreviation for kernel-specific constants.

\begin{theorem}[Properties of the estimator $\hatFT_{u,\omega}$] \label{expsmooth}
Suppose that assumptions (A1), (A2), and (A4) to (A7) with $p=\infty$ and $k\leq 4$. Then the estimator
\begin{align} 
\hatFT_{u,\omega}
= \int_{\Pi}\KfT(\omega-\lambda)\,\IT_{u,\lambda}\,d\lambda
\end{align}
has mean
\begin{equation}
\begin{split}
\E\innerprod{\hatFT_{u,\omega}g_1}{g_2}
&=\innerprod{\F_{u,\omega}g_1}{g_2}
+\SSS{\frac{1}{2}}\,\btT^2\,\kappat\,
\SSS{\frac{\partial^2}{\partial u^2}}\innerprod{\F_{u,\omega}g_1}{g_2}
+\SSS{\frac{1}{2}}\,\bfT^2\,\kappaf\,
\SSS{\frac{\partial^2}{\partial \omega^2}}\innerprod{\F_{u,\omega}g_1}{g_2}\\
&\qquad\qquad+o(\btT^2) + o(\bfT^2)+O\big(\SSS{\frac{\log(\btT\,T)}{\btT\,T}}\big),
\end{split}
\end{equation}
and covariance structure
\begin{equation}
\begin{split}
\label{eq:covsmooth}
\Cov&\big(\innerprod{\hatFT_{u,\omega_1}}{ g_1 \otimes g_2},\innerprod{\hatFT_{u,\omega_2}}{ g_3 \otimes g_4}\big)\\& 
=\frac{2\pi\,\norm{\Kt}^2_2}{\btT\,T}
\int_{\Pi}\KfT(\omega_1-\lambda_1)\,\KfT(\omega_2-\lambda_1)\, \innerprod{\F_{u,\lambda_1}g_3}{g_1}\, \innerprod{\F_{u,-\lambda_1}g_4}{g_2}\,d\lambda_1\\
&\qquad + \frac{2\pi\,\norm{\Kt}^2_2}{\bfT\,T} \int_{\Pi}\KfT(\omega_1-\lambda_1)\, \KfT(\omega_2+\lambda_1)\,\innerprod{\F_{u,\lambda_1}g_4}{g_1}\,\innerprod{\F_{u,-\lambda_1}g_3}{g_2}\, d\lambda_1 \\
&\qquad+O\big(\SSS{\frac{\log(\btT\,T)}{\btT\,T}}\big)
+O\big(\SSS{\frac{\btT}{T}}\big)
+O\big((\btT\,\bfT\,T)^{-2}\big)
\end{split}
\end{equation}
for all $g_1, g_2, g_3, g_4 \in H_{\cnum}$.
\end{theorem}

The proof follows from a multivariate Taylor expansion and an application of Lemma P4.1 of \citet{Brillinger}. We note that the covariance has greatest magnitude for $\omega_1 \pm \omega_2 \equiv 0 (\text{mod} 2\pi)$, where the weight is concentrated in a band of width $O(\bfT)$ around $\omega_1$ and $\omega_2$ respectively. The above result shows that the bandwidths need to decay to zero with an appropriate rate in order to obtain consistency.

\begin{proposition}\label{sharp}
Under assumptions (A1), (A2), and (A4) to (A8) with $p=\infty$ and $k\leq 4$, we have
\begin{align*}
\lim_{T \to \infty} \btT\,\bfT\,T\,
&\Cov\big(\innerprod{\hatFT_{u,\omega_1}}{ g_1 \otimes g_2},\innerprod{\hatFT_{u,\omega_2}}{ g_3 \otimes g_4}\big)\\
&= 2\pi\,\norm{\Kt}^2_2\,\norm{\Kf}^2_2\,\eta(\omega_1-\omega_2)\, \innerprod{\F_{u,\omega_1}g_3}{g_1}\, \innerprod{\F_{u,-\omega_1}g_4}{g_2}\\
&\qquad+2\pi\,\norm{\Kt}^2_2\,\norm{\Kf}^2_2\,\eta(\omega_1+\omega_2)\, \innerprod{\F_{u,\omega_1}g_4}{g_1}\, \innerprod{\F_{u,-\omega_1}g_3}{g_2},
\tageq  
\end{align*}
for all $g_1, g_2, g_3, g_4 \in H_{\cnum}$ and for fixed $\omega_1, \omega_2$. If $\omega_1, \omega_2$ depend on $T$ then the convergence holds provided that $\liminf_{T\to\infty}|(\omega_{1,T} \pm \omega_{2,T})\mathop{\mathrm{mod}} 2\pi|>\veps$ for some $\veps>0$. 
\end{proposition}
The proof follows straightforwardly from a change of variables and a functional generalization of approximate identities \citep[e.g.,][]{Edwards1967}. 

\begin{corollary} \label{sharp2}
Under the same assumptions
\begin{align*}
\Bigsnorm{\Cov \big(\hatFT_{u, \omega_1},\hatFT_{u, \omega_2}\big)}_2
=O\Big(\SSS{\frac{1}{\bfT\,\btT\,T}}\Big)
\end{align*}
uniformly in $\omega_1,\omega_2\in[-\pi,\pi]$ and $u\in[0,1]$.
\end{corollary}
\begin{proof}
Note that $\norm{\KfT}_{\infty}=O\big(\frac{1}{\bfT}\big)$
and $\norm{\KfT}_{1}=1$, from which it is easy to see that
\[
\sup_{\omega}\Bignorm{\int_{\Pi}\KfT(\omega+\lambda)\,\KfT(\lambda)\,d\lambda}_2 =O\Big(\SSS{\frac{1}{\bfT}}\Big).
\]
Since $\snorm{\F_{u,\omega}}_{2}$ is moreover uniformly bounded in $u$ and $\omega$, the statement follows directly from equation \eqref{covsmoothia}.
\end{proof}

\begin{remark}  \label{kernelremark}
Theorem \ref{exp}, Theorem \ref{expsmooth}, Proposition \ref{sharp}, and Corollary \ref{sharp2} can be shown to hold in a stronger sense under additional assumptions. Namely, if the respective set of assumptions hold with $p=2$ then the kernel function of the transfer operator $\AAtT{\omega}$ is well-defined. If this function is continuous and the white noise process $\{\varepsilon_t\}$ is moreover mean square continuous with $\sup_{\omega_1,\ldots,\omega_{k-1}}\|f^{\varepsilon}_{\omega_1,\cdots,\omega_{k-1}}\|_{\infty} < \infty$ for $k \le 4$, the aforementioned statements hold in uniform norm.
\end{remark}

\begin{theorem}[Convergence in integrated mean square]
\label{IMSE}
Suppose that assumptions (A1),(A2), and (A4) to (A8) with $k=2,4$ hold. Then the spectral density operator is consistent in integrated mean square. More precisely, we have
\begin{align*}
\mathrm{IMSE}(\hatFT_{u,\omega})
&=\int_{\Pi} \mean\bigsnorm{ \hatFT_{u,\omega}-\F_{u, \omega}}^2_2 \,d\omega\\
&=O\big((\btT\,\bfT\,T)^{-1}\big)+ o\big(\btT^2+\bfT^2+(\btT\,T)^{-1}\log(\btT\,T)\big).
\end{align*}
Since it is uniform in $\omega \in \Pi$, we have pointwise mean square convergence where the error also satisfies
$\mean\snorm{\hatFT_{u,\omega}-\F_{u,\omega}}^2_2
=O(\normalfont{\frac{1}{b_{\text{t}} b_{\text{f}} T}})+ o\big(b_{\text{t}}^2+b_{\text{f}}^2 + \frac{\log b_{\text{t}} T}{b_{\text{t}} T} \big)$.
\end{theorem}
The proof follows almost straightforwardly from decomposing the above in terms of its variance and its squared bias.
\textcolor{black}{\begin{remark}[Discrete observations]
In practice, functional data commonly are observed not continuously but only on a discrete grid. In this case, the above consistency results continue to hold only under additional regularity assumptions. To illustrate the effect of discrete observations, suppose that, for given $T$, we observe the random functions $X_t(\tau)$, $t=1,\ldots,T$, on the grid points
$0 \leq \tau_{1} < \ldots < \tau_{M} \leq 1$. The corresponding discrete data $x_{i,t}=X_{t}(\tau_{i})$ are interpolated to yield functions $X^M_t(\tau)$ in some $M$--dimensional subspace $H_M$ of $H$. Furthermore, let $P_M\,X_t$ be the orthogonal projection of $X_t$ onto $H_M$. Then the mean square error of the estimator $\hat{\F}_{u,\omega}(X^M)$ based on the discrete observation can be bounded by application of Minkowski's inequality by
\begin{equation}
\label{errordecomp}
\begin{split}
\E&\snorm{\hat{\F}_{u,\omega}(X^M)-\F_{u,\omega}}_2^2\\
&\leq 2\,\E\snorm{\hat{\F}_{u,\omega}(X^M)-\hat{\F}_{u,\omega}(P_M\,X)}_2^2
+2\,\E\snorm{\hat{\F}_{u,\omega}(P_M\,X)-\F_{u,\omega}}_2^2.
\end{split}
\end{equation}
Here, the first term can be interpreted as the error due to discretization. In the second term, the mean square error of $\hat{\F}_{u,\omega}(P_M\,X)$ can be rewritten as
\[
\begin{split}
\E&\snorm{\hat{\F}_{u,\omega}(P_M\,X)-\F_{u,\omega}}_2^2\\
&=\E\snorm{\hat{\F}_{u,\omega}(P_M\,X)-\F_{u,\omega}(P_M\,X)}_2^2
+\snorm{\F_{u,\omega}(P_M\,X)-\F_{u,\omega}}_2^2,
\end{split}
\]
where $\F_{u,\omega}(P_M\,X)$ is the time-varying spectral density operator of the process $P_M\,X_t$. The first term describes the estimation error and, by Parseval's equality, is bounded by
\[
\E\snorm{\hat{\F}_{u,\omega}(P_M\,X)-\F_{u,\omega}(P_M\,X)}_2^2
\leq\E\snorm{\hat{\F}_{u,\omega}(X)-\F_{u,\omega}}_2^2
\]
and hence converges to zero as $T\to\infty$ in a local stationary framework. Finally, the second term is the approximation error due to replacing the functions $X_t$ by their projections $P_M\,X_t$. Under suitable regularity conditions on the functions $X_t$ and for an increasingly dense grid, the approximation error can be made arbitrarily small. Furthermore, the discretization error $\norm{X^M_t-P_M\,X_t}_2$ converges to zero such that for an appropriate rate of $M\to\infty$, the error due to discretization in \ref{errordecomp} also tends down to zero. The detailed derivations are similar to those in section 5 of \citet{PanarTav2013a} and are therefore omitted.
\end{remark}}

\subsection{Weak convergence of the empirical process}

The results of the previous section give rise to investigating the limiting distribution of $\hatFT_{u,\omega}$, the local estimator of the spectral density operator. We will proceed by showing that joint convergence of its kernel $\hatfT_{u,\omega}$ to complex Gaussian elements in $L^2_\cnum([0,1]^2)$ can be established.

Consider the sequence of random elements $\big(\ET_{u,\omega}\big)_{T\in\nnum}$ in $L^2_\cnum([0,1]^2)$,
where
\[
\ET_{u,\omega}
=\sqrt{\btT\,\bfT\,T}\,\big(\hatfT_{u,\omega}
-\mean\big[\hatfT_{u,\omega}\big]\big)
\]
for fixed $\omega\in[-\pi,\pi]$ and $u\in[0,1]$. In order to establish convergence in $L^2_\cnum([0,1]^2)$, it is more appropriate to consider the representation of $\ET_{u,\omega}$ with respect to some orthonormal basis. For this, let $\{\psi_m\}_{m\in\nnum}$ be an orthonormal basis of $H_\cnum$. Then $\{\psi_{mn}\}_{m,n\in\nnum}$ with $\psi_{mn}=\psi_m\otimes\psi_n$ forms an orthonormal basis of $L^2_\cnum([0,1]^2)$, and $\ET_{u,\omega}$ equals
\[
\ET_{u,\omega}=\lsum_{m,n\in\nnum}
\innerprod{\ET_{u,\omega}}{\psi_{mn}}\,\psi_{mn}.
\]
Hence, the distribution of $\ET_{u,\omega}$ is fully characterized by the finite-dimensional distribution of the coefficients of its basis representation. Furthermore, weak convergence of $\ET_{u,\omega}$ will follow from the weak convergence of $\big(\innerprod{\ET_{u,\omega}}{\psi_{mn}}\big)_{m,n\in\nnum}$ in the sequence space $\ell^2_\cnum$. Subsequently, we identify $\ET_{u,\omega}$ with its dual $(\ET_{u,\omega})^*\in L^2_\cnum([0,1]^2)^*$ and write
\[
\ET_{u,\omega}(\phi)=\innerprod{\ET_{u,\omega}}{\phi}
\]
for all $\phi\in L^2_\cnum([0,1]^2)$.

To show convergence to a Gaussian functional process, we make use of the following result by \citet{CremerKadelka}, which weakens the tightness condition usually employed to prove weak convergence and generalizes earlier results by \citet{Grinblat}.

\begin{lemma}
\label{convergencecriterion}
Let $(T,\mathcal{B},\mu)$ be a measure space, let $(E,|\cdot|)$ be a Banach space, and let $(X_n)_{n\in\nnum}$ be a sequence of random elements in $L^p_E(T,\mu)$ such that
\begin{romanlist}
\item
the finite-dimensional distributions of $X_n$ converge weakly to those of a random element $X_0$ in $L^p_E(T,\mu)$ and
\item
$\DS\limsup_{n\to\infty}\mean\norm{X_n}_p^p\leq\mean\norm{X_0}_p^p$.
\end{romanlist}
Then $X_n$ converges weakly to $X_0$ in $L^p_E(T,\mu)$.
\end{lemma} 

In our setting, the weak convergence of the process $\ET_{u,\omega}$ in $L^2_\cnum([0,1]^2)$ will follow from the joint convergence of
$\ET_{u,\omega}(\psi_{m_1,n_1}),\ldots,\ET_{u,\omega}(\psi_{m_k,n_k})$ for all $k\in\nnum$ and the condition
\begin{equation}
\label{cremerkadelka-condition}
\mean\bignorm{\ET_{u,\omega}}^2_2
=\lsum_{m,n\in\nnum}\mean\big|\ET_{u,\omega}(\psi_{mn})\big|^2\to
\lsum_{m,n\in\nnum}\mean\big|E_{u,\omega}(\psi_{mn})\big|^2
=\mean\bignorm{E_{u,\omega}}^2_2
\end{equation}
as $T\to\infty$. In contrast, \citet{PanarTav2013a} employ the slightly stronger condition
\[
\big|\ET_{u,\omega}(\psi_{mn})\big|^2\leq \phi_{mn}
\]
for all $T\in\nnum$ and $m,n\in\nnum$ and some sequence $(\phi_{mn})\in\ell^1$. In fact, the condition corresponds in our setting to the one given in \citet{Grinblat}. Finally, we note that condition \eqref{cremerkadelka-condition} is sufficient for our purposes,  but recently it has been shown \citep{Bogachev} that it can be further weakened to
\[
\sup_{T\in\nnum}\mean\bignorm{\ET_{u,\omega}}^2_2<\infty.
\]

For the convergence of the finite-dimensional distributions, we show convergence of the cumulants of all orders to that of the limiting process.
For the first and second order cumulants of $\ET_{u,\omega}(\psi_{mn})$, this follows from Theorem \ref{expsmooth}. It therefore remains to show that all cumulants of higher order vanish asymptotically.

\begin{proposition}
\label{convCLT}
Suppose that assumptions (A1), (A2), and (A4) to (A8) for some $k\geq 3$
hold. Then, for all $u\in[0,1]$ and for all $\omega_{i}\in[-\pi,\pi]$ and $m_i,n_i\in\nnum$ for $i=1,\ldots,k$, we have
\begin{align}
\cum\big(\ET_{u,\omega_{1}}(\psi_{m_1n_1}),\ldots,
\ET_{u,\omega_{k}}(\psi_{m_kn_k})\big)=o(1)
\end{align}
as $T\to\infty$.
\end{proposition} 

The distributional properties of the functional process can now be summarized in the following theorem.
\begin{theorem}[Weak convergence]
\label{CLT}
Suppose that assumptions (A1), (A2), and (A4) to (A8) for all $k\in\nnum$
hold. Then
\begin{equation} 
\big(\ET_{u,\omega_j}\big)_{j=1,\ldots,J}\dconv \big(E_{u,\omega_j}\big)_{j=1,\ldots,J},
\end{equation}
where $E_{u, \omega_j}$, $j=1,\ldots,J$, are jointly Gaussian elements in $L^2_\cnum([0,1]^2)$ with means $\mean\big(E_{u, \omega_i}(\psi_{mn})\big)=0$ and covariances
\begin{equation}
\begin{split}
\Cov\Big(E_{u, \omega_i}(\psi_{mn}),E_{u,\omega_j}(\psi_{m'n'})\Big)&\\
=2\pi\,\norm{\Kt}^2_2\,\norm{\Kf}^2_2\,
\Big[\eta(\omega_i-\omega_j)&\,
\biginnerprod{\F_{u, \omega_i}\,\psi_{m'}}{\psi_{m}}\,
\biginnerprod{\F_{u,-\omega_i}\,\psi_{n'}}{\psi_{n}}\\
+\eta(\omega_i+\omega_j)&\,
\biginnerprod{\F_{u, \omega_i}\,\psi_{n'}}{\psi_{m}}\,
\biginnerprod{\F_{u,-\omega_i}\,\psi_{m'}}{\psi_{n}}\Big]
\end{split}
\label{covempproc}
\end{equation}
for all $i,j\in 1,\ldots,J$ and $m,m',n,n'\in\nnum$.
\end{theorem}

\begin{proof}[Proof of Theorem \ref{CLT}]
For condition \ref{cremerkadelka-condition}, we note that
\[
\mean\bignorm{\ET_{u,\omega}}^2_2
=\int_{[0,1]^2}
\var\big(\ET_{u,\omega}(\tau,\sigma)\big)\,d\tau\,d\sigma
=\btT\,\bfT\,T\,\bignorm{\var(\FT_{u,\omega})}^2_2.
\]
and it therefore is satisfied by Theorem \ref{expsmooth}. Together with the convergence of the finite-dimensional distributions this proves the asserted weak convergence.
\end{proof}

\section{Numerical simulations}
\label{simulation-sect}

To illustrate the performance of the estimator in finite samples, we consider a time-varying functional time series with representation
\begin{equation}\label{FAR1proc}
\XT{t} = B_{\frac{t}{T},1}(\XT{t-1})+\varepsilon_t,
\end{equation}
where $B_{u,1} \in \mathcal{B}_\infty$ is continuous in $u \in [0,1]$ and where $\{\varepsilon_t\}$ is a collection of independent innovation functions. In order to generate the process, let $\{\psi_i\}_{i \in \mathbb{N}}$ be an orthonormal basis  of $H$ and denote the vector of the first $k$ Fourier coefficients of $\XT{t}$ by ${\boldsymbol{X}}^{(T)}_{t} = (\langle \XT{t}, \psi_1 \rangle, \hdots, \langle \XT{t}, \psi_k \rangle)^{'}$. Similar to \cite{Hormann2015}, we exploit that the linearity of the autoregressive operator implies the first $k$ Fourier coefficients, for $k$ large, approximately satisfy a VAR(1) equation. That is, 
\begin{equation}\label{eq:XFC}
{\boldsymbol{X}}^{(T)}_{t} \approx \mathfrak{B}_{\frac{t}{T},1}{\boldsymbol{X}}^{(T)}_{t-1}+\boldsymbol{\varepsilon}_t \quad \forall \hspace{1pt} t,T,
\end{equation}
where $\boldsymbol{\varepsilon}_{t} = (\langle \varepsilon_{t}, \psi_1 \rangle,  \hdots, \langle \varepsilon_{t}, \psi_k \rangle)^{'}$ and $\mathfrak{B}_{\frac{t}{T},1}= (\langle B_{\frac{t}{T},1}(\psi_i),\psi_j\rangle, 1 \le i,j \le k)$. Correspondingly, the local spectral density kernel will satisfy
\begin{equation*}
f^{(T)}_{u,\omega}(\tau, \sigma) \approx \lim_{i,j \to \infty}\sum_{i,j = 1}^{k} \mathfrak{f}^{(T)}_{u,\omega,i,j}  \psi_i(\tau)  \psi_j(\sigma), 
\end{equation*}
where $ \mathfrak{f}^{(T)}_{u,\omega}$ is the spectral density matrix of the Fourier coefficients in \eqref{eq:XFC}.  Implementation was done in R together with the {\tt fda} package. For the simulations, we choose the Fourier basis functions on $[0,1]$.  The construction of the estimator in \eqref{eq:smoothest} requires specification of smoothing kernels and corresponding bandwidths in time- as well as frequency direction. Although the choice of the smoothing kernels usually does not affect the performance significantly, bandwidth selection is a well-known problem in nonparametric statistics. As seen from Theorem \ref{expsmooth}, both bandwidths influence the bias-variance relation. Depending on the persistence of the autoregressive process a smaller bandwidth in frequency direction is desirable around the peak (at $\lambda=0$ for the above process), while slow changes in time direction allow for tapering (i.e., smoothing in time direction) over more functional observations. It would therefore be of interest to develop an adaptive procedure as proposed in \cite{Delft2015} to select the bandwidth parameters. Investigation of this is however beyond the scope of the present paper. In the examples below, the bandwidths were set fixed to $\btT=T^{-1/6}$ and $\bfT = 2T^{-1/5}-\btT$. We chose as smoothing kernels
\begin{equation*}
\Kt(x) =\Kf(x) = 6(\frac{1}{4}-x^2) \quad x \in [-\frac{1}{2},\frac{1}{2}],
\end{equation*}
which have been shown to be optimal in the time series setting \citep[][]{Dahlhaus1996b}.

In order to construct the matrix $\mathfrak{B}_{\frac{t}{T},1}$, we first generate a matrix $A_{u}$ with entries that are mutually independent Gaussian where the $(i,j)$-th entry has variance
\begin{equation*}
u i^{-2c}+(1-u)e^{-i-j}.
\end{equation*}
The entries will tend to zero as $i,j \to \infty$ , because the operator $B_{\frac{t}{T},1}$ is required to be bounded. The matrix $\mathfrak{B}_{u,1}$ is consequently obtained as $\mathfrak{B}_{u,1}= \eta A_{u}/\snorm{A_{u}}_{\infty}$. The value of $\eta$ thus determines the persistence of the process. Additionally, the collection of innovation functions $\{\varepsilon_t\}$ is specified as a linear combination of the Fourier basis functions with independent zero-mean Gaussian coefficients such that the $l$-th coefficient $\langle \varepsilon_{t}, \psi_l \rangle$ has variance $1/{[(l-1.5)\pi]}^2$. The parameters were set to $c=3$ and $\eta=0.4$. To visualize the variability of the estimator, figure \ref{table:FAR1nopeak} depicts the amplitude of the true spectral density kernel of the process for various values of $u$ and $\lambda$ with 20 replications of the corresponding estimator superposed for different sample sizes $T$.  For each row, the same level curves were used where each level curve has the same color-coding within that row. The first two rows of figure \ref{table:FAR1nopeak} give the different levels for the estimator around the peak in frequency direction, while the last row provides contour plots further away from the peak. Increasing the sample size leads to less variability, as can be seen from the better aligned contour lines. It can also be observed that the estimates become more stable as we move further away from the peak. Nevertheless, the peaks and valleys are generally reasonably well captured even for the contour plots in the area around the peak.
\begin{center}
\begin{figure}[!hbt]
 \vspace*{-0cm} \hspace*{-1.0cm}
\begin{tabular}{c|c|ccc}
 \multicolumn{5}{c}{$\underline{\lambda=0}$} \\
 & true & $ T = 2^{9} $ &  $T=2^{12}$ & $T = 2^{16}$ \\ \hline
\begin{sideways}\footnotesize{\hspace{1cm} $u=0.25$} \end{sideways} &\includegraphics[width=0.22\textwidth]{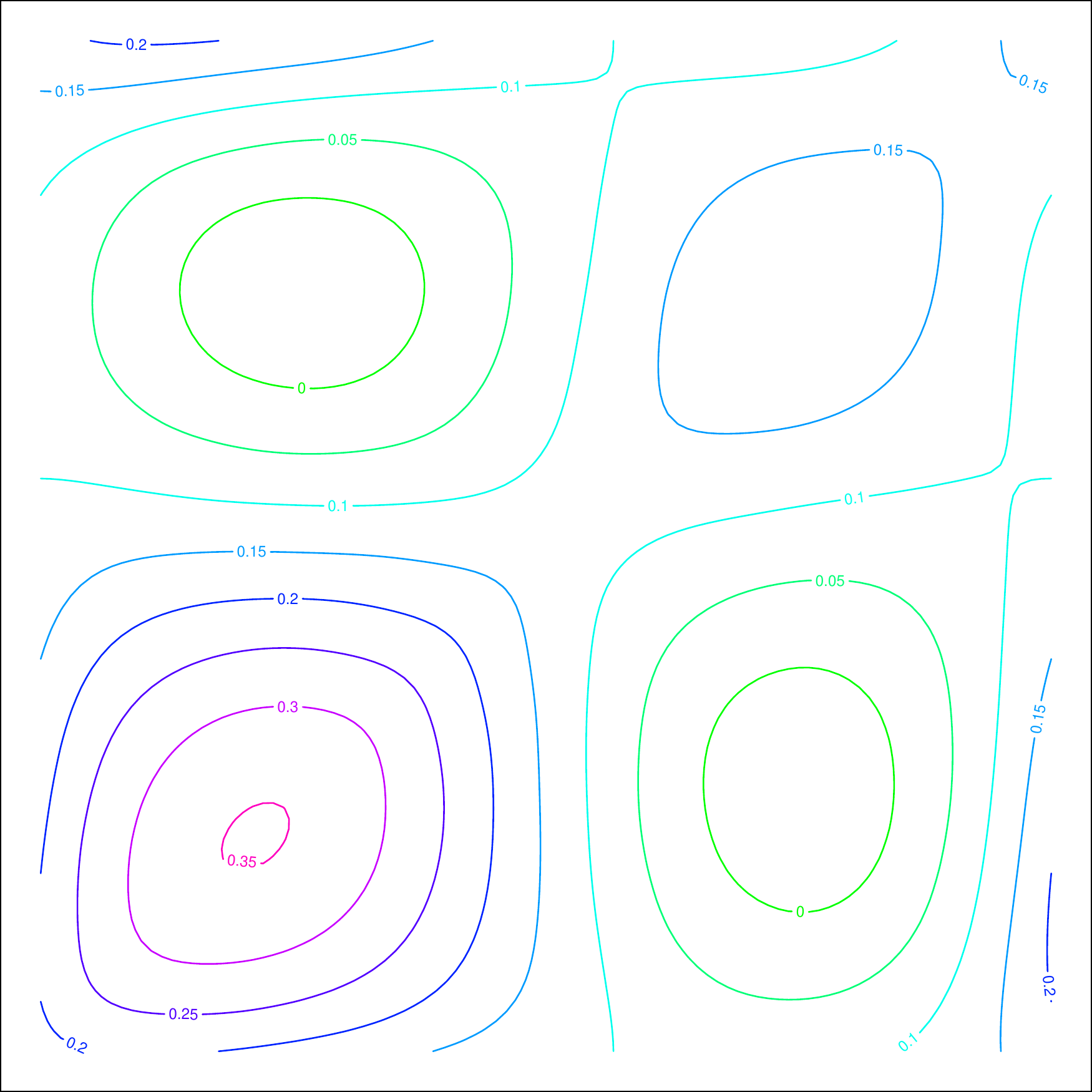}& \includegraphics[width=0.22\textwidth]{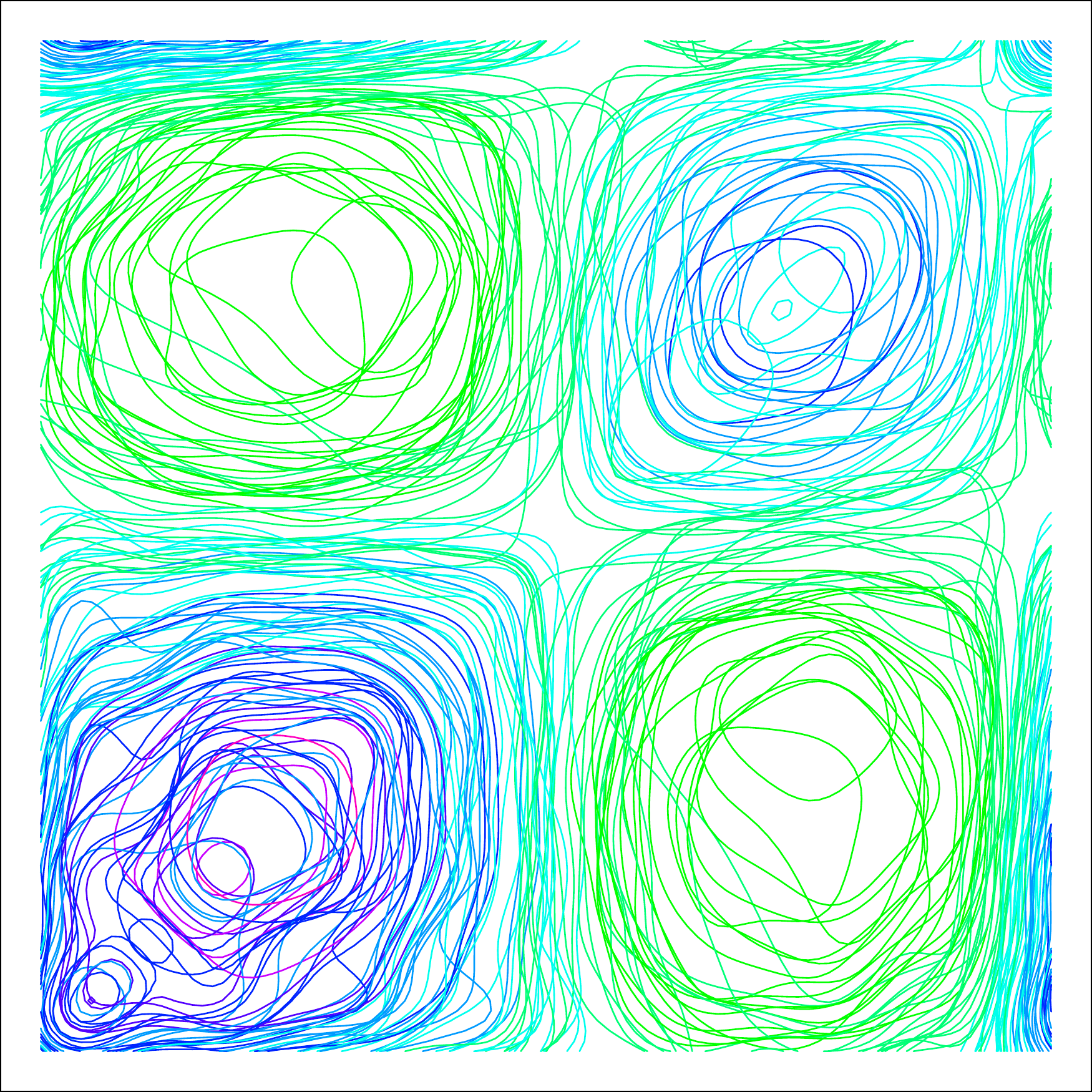}& \includegraphics[width=0.22\textwidth]{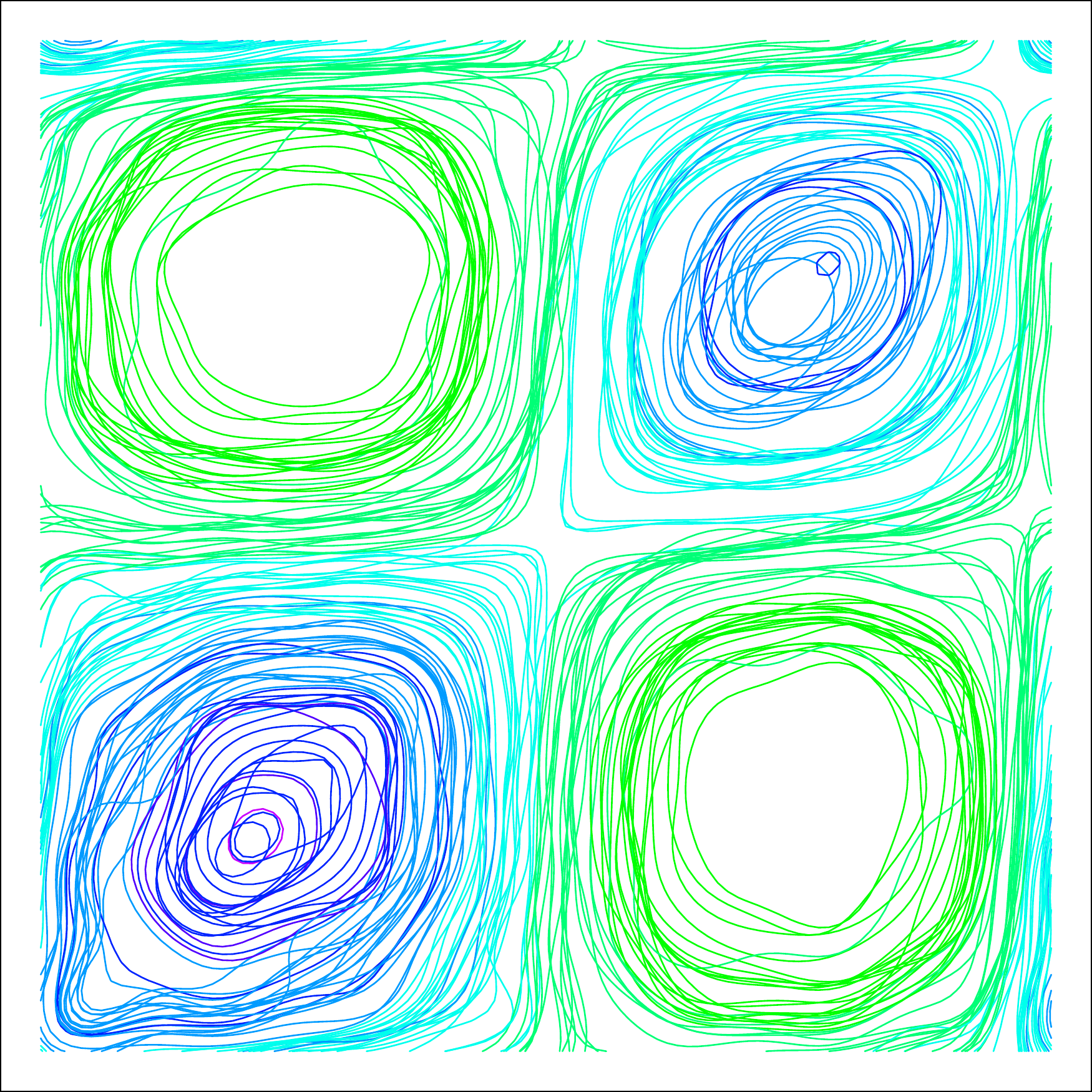}& \includegraphics[width=0.22\textwidth]{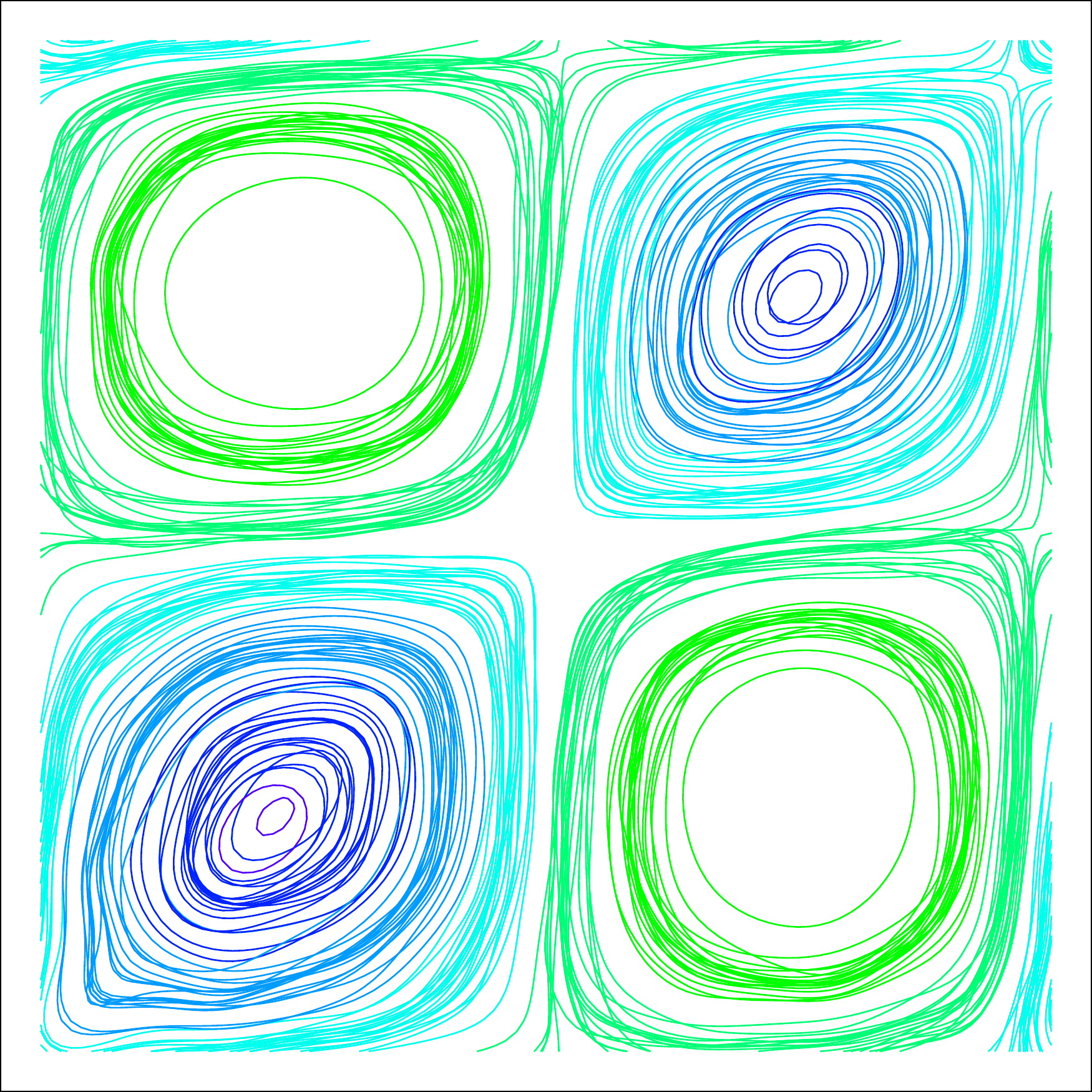} \\ 
\end{tabular}
 \vspace*{-0cm} \hspace*{-1.0cm}
\begin{tabular}{c|c|ccc}
 \multicolumn{5}{c}{$\underline{\lambda=3/10\pi}$}\\
\begin{sideways}\footnotesize{\hspace{1cm} $u=0.5$}  \end{sideways} &\includegraphics[width=0.22\textwidth]{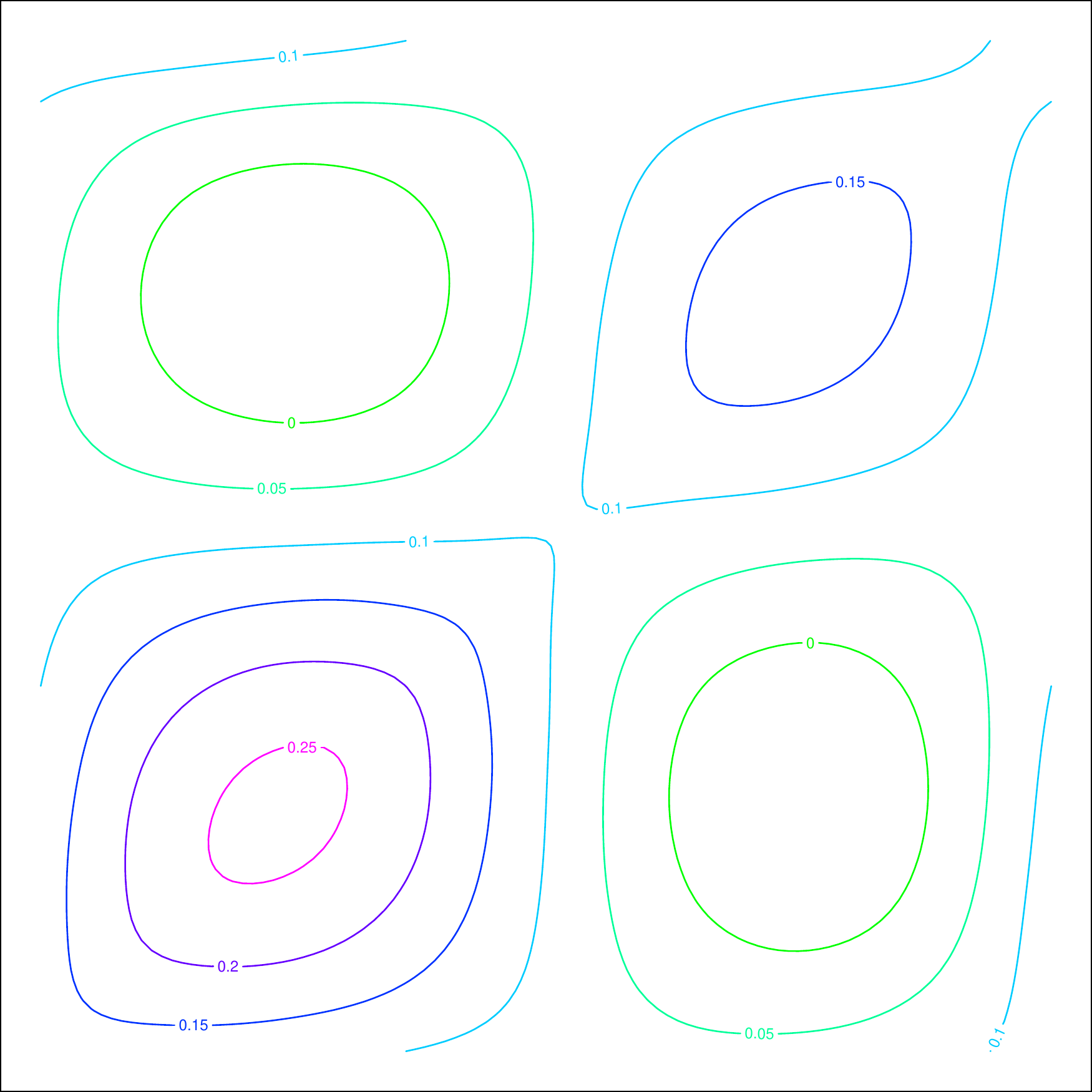}& \includegraphics[width=0.22\textwidth]{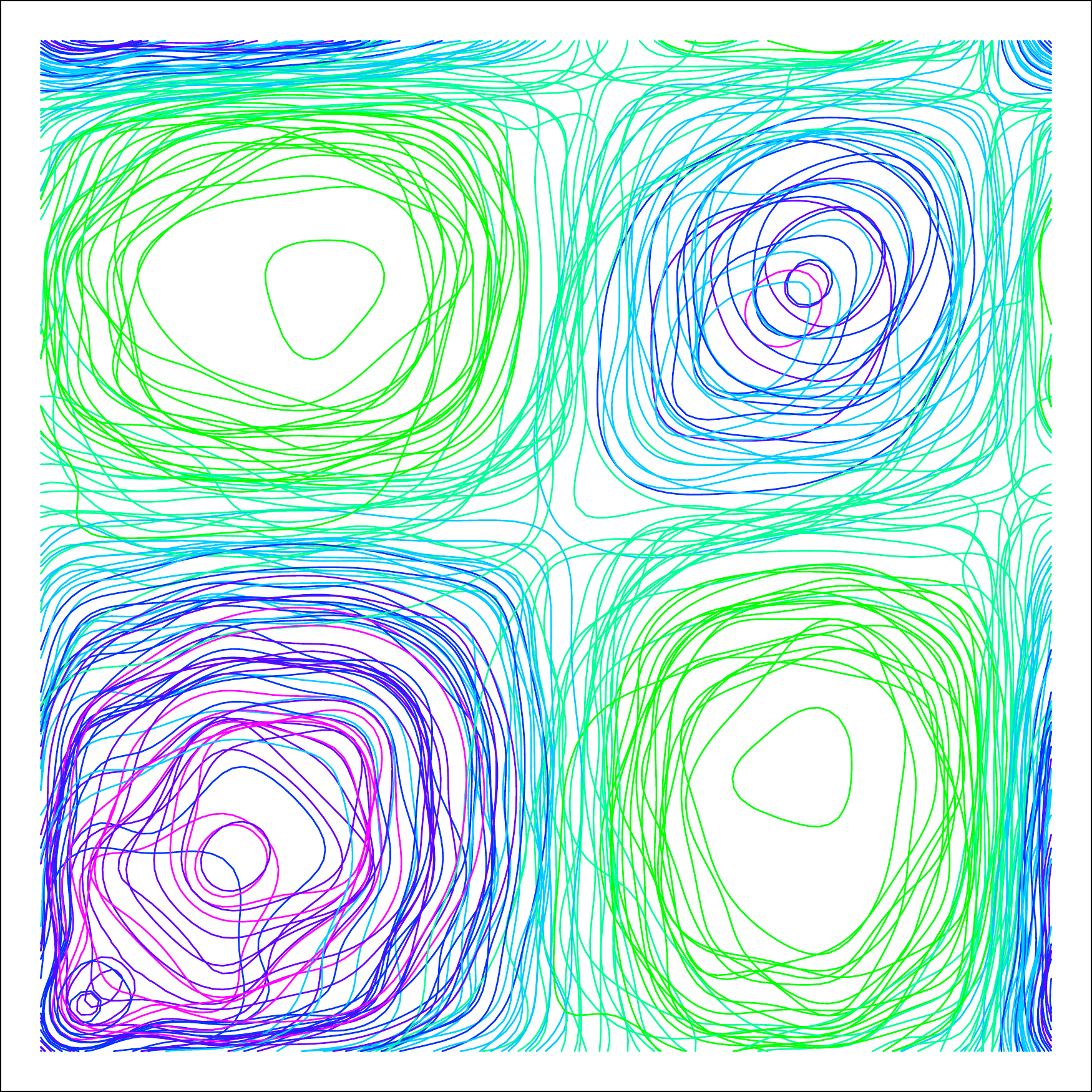}& \includegraphics[width=0.22\textwidth]{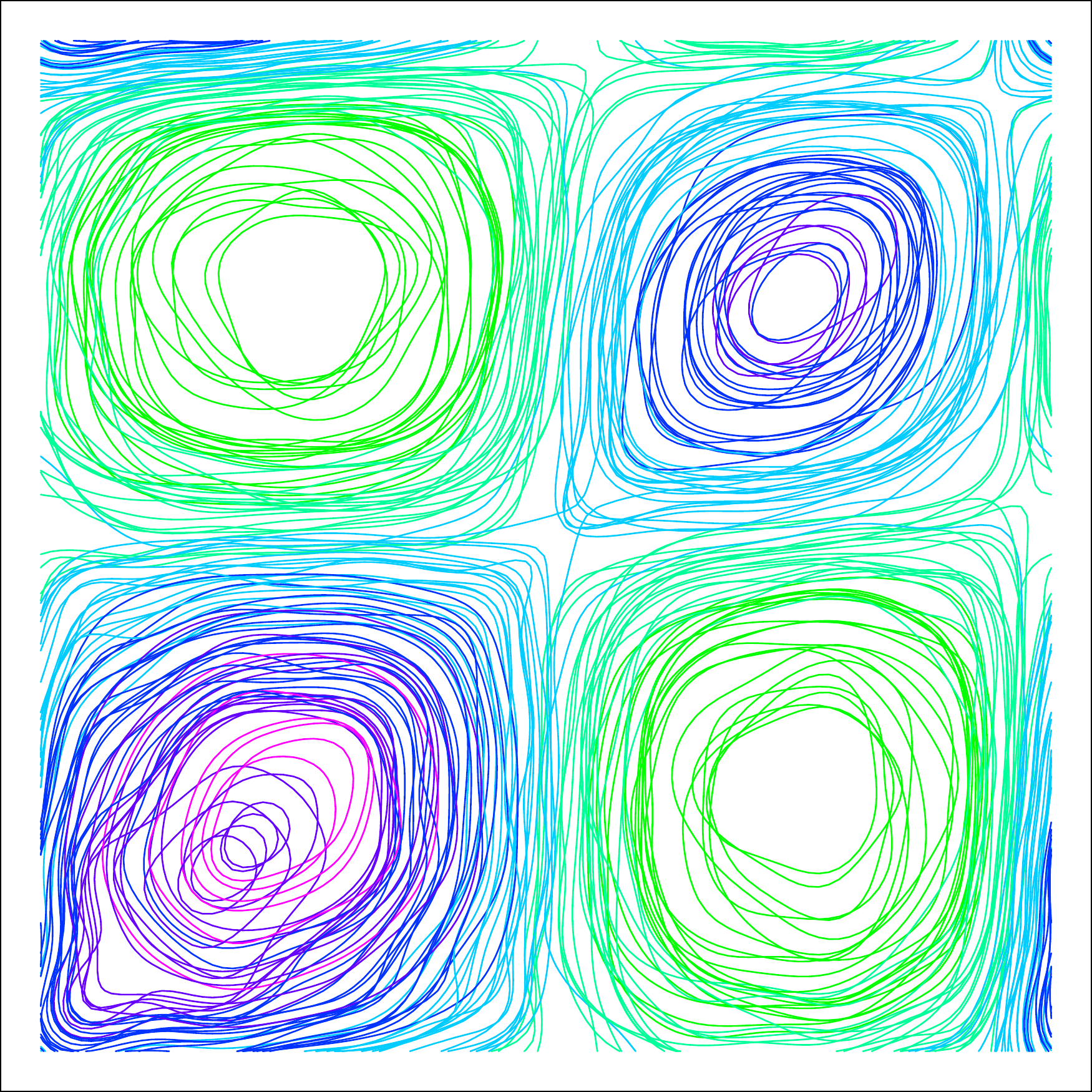}&\includegraphics[width=0.22\textwidth]{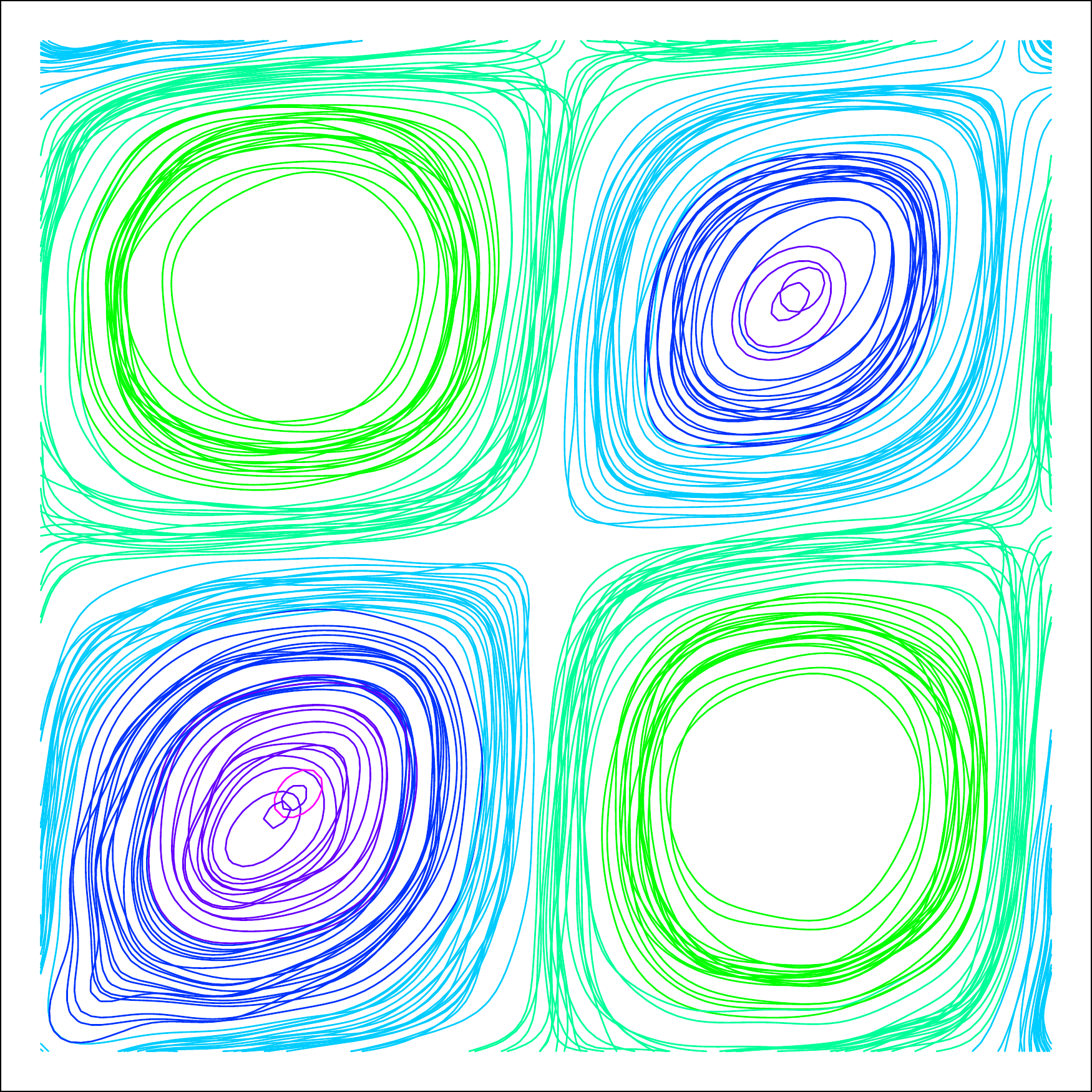} \\ 
\end{tabular}
 \vspace*{-0cm} \hspace*{-1.0cm}
\begin{tabular}{c|c|ccc}
 \multicolumn{5}{c}{$\underline{\lambda=9/10\pi}$} \\
\begin{sideways}\footnotesize{\hspace{1cm} $u=0.25$} \end{sideways} &\includegraphics[width=0.22\textwidth]{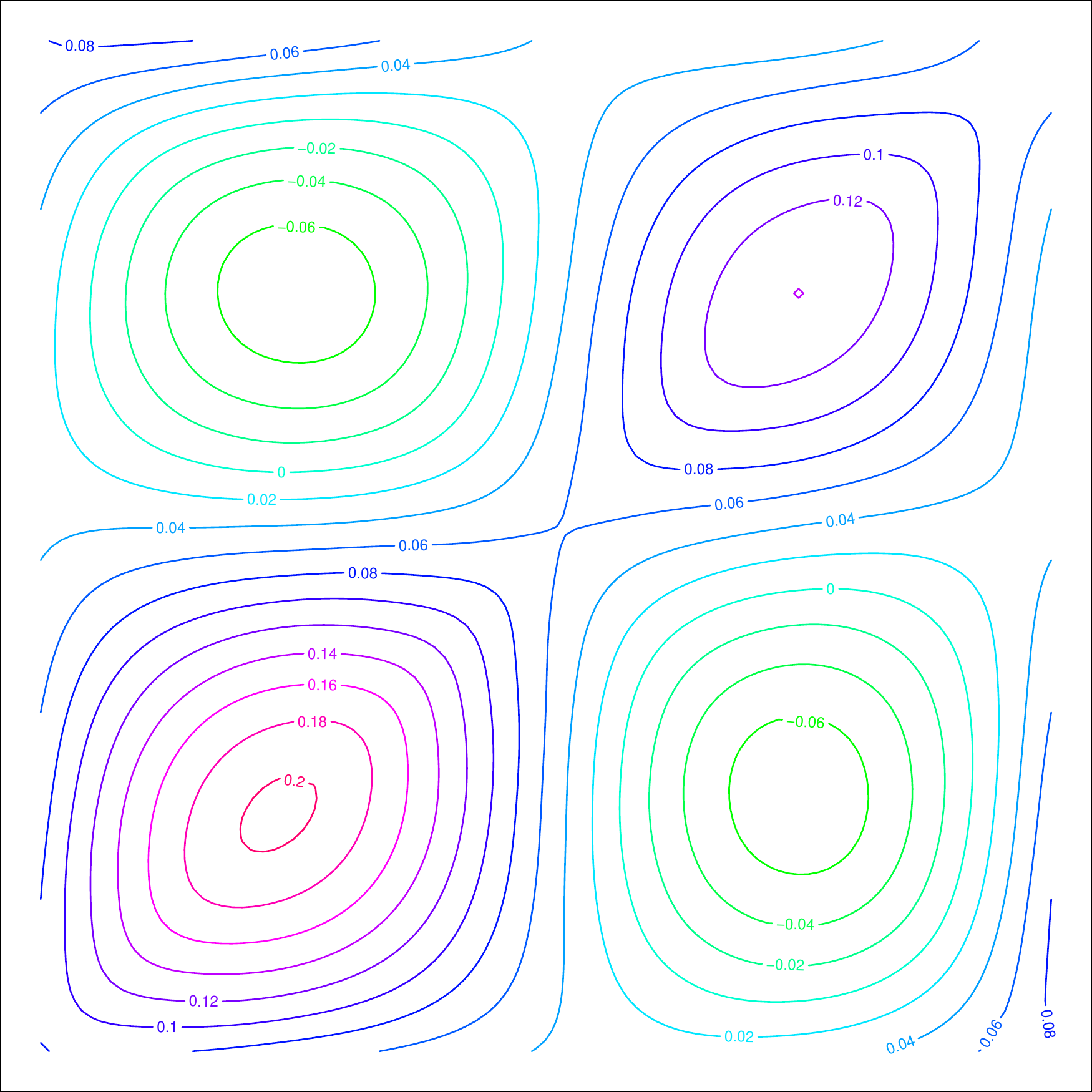}& \includegraphics[width=0.22\textwidth]{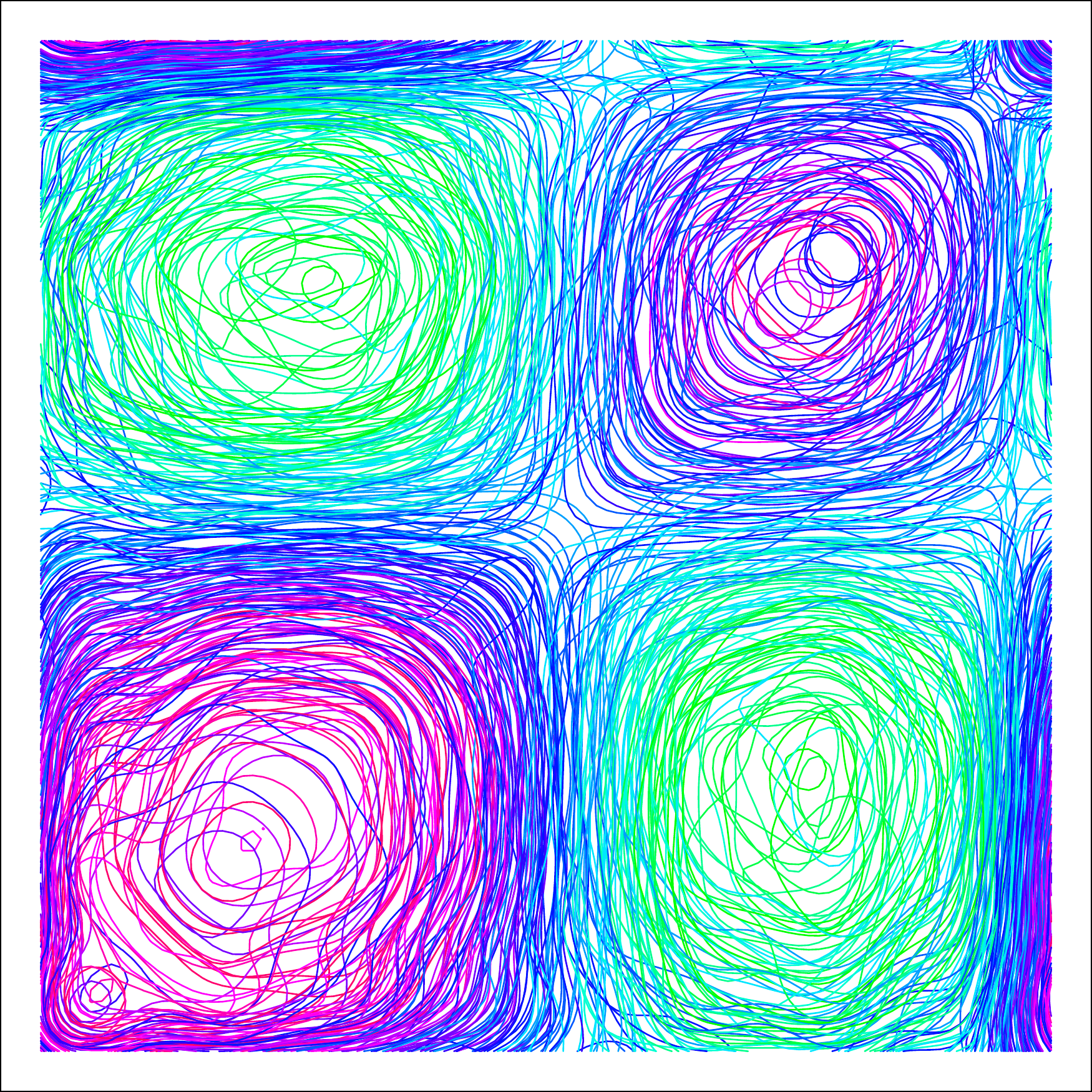}& \includegraphics[width=0.22\textwidth]{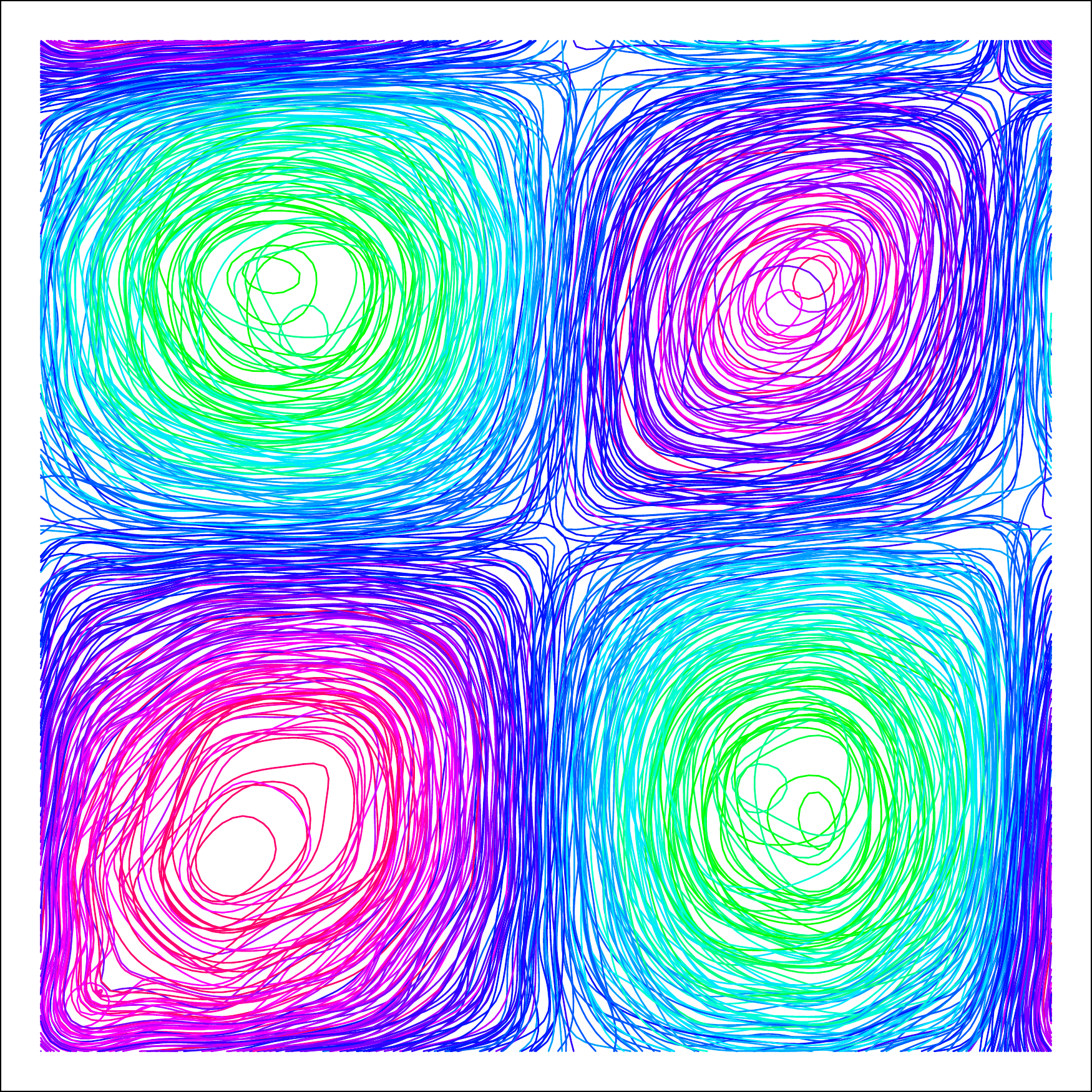}& \includegraphics[width=0.22\textwidth]{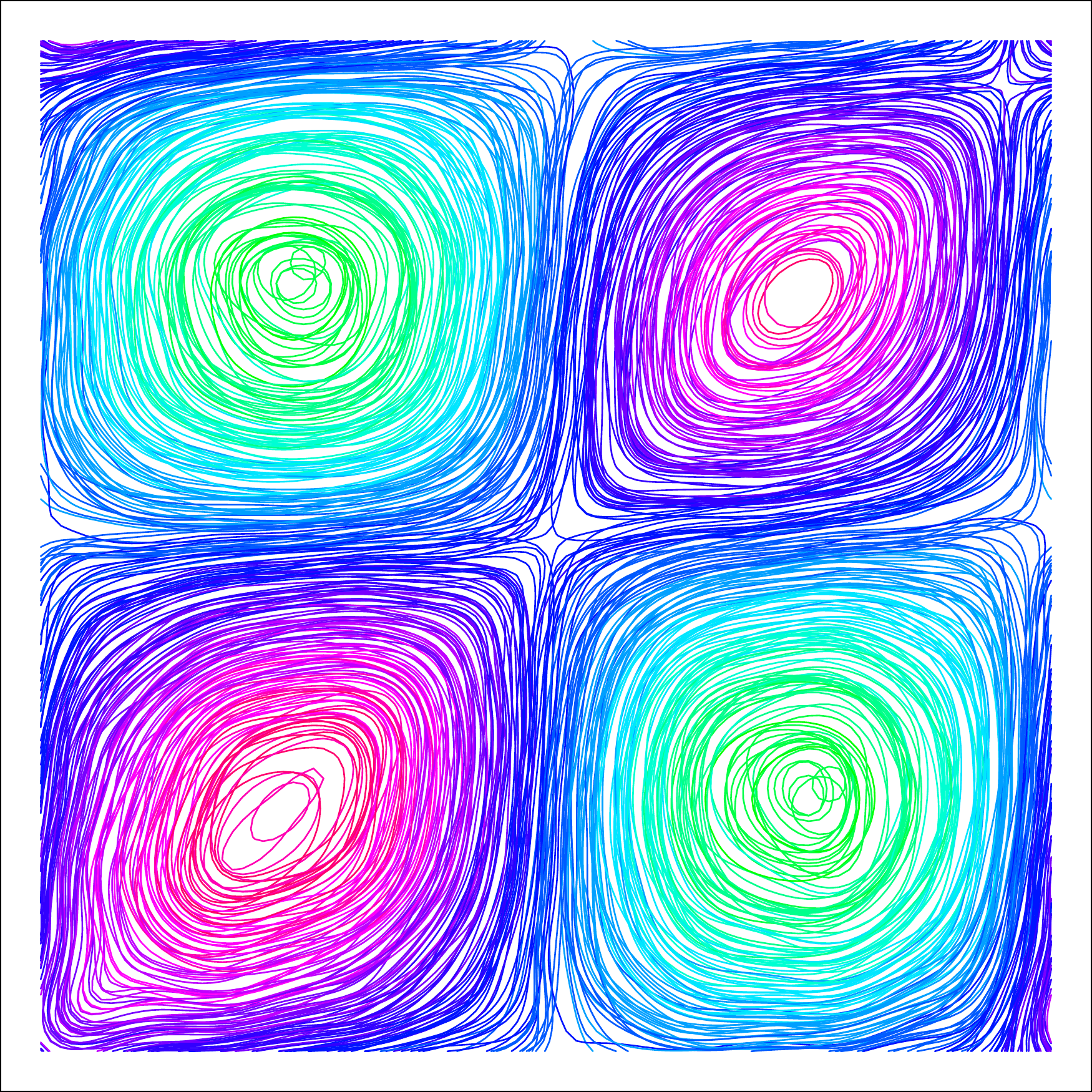} \\ 
\end{tabular}
\vspace{-5pt}
 \caption{\footnotesize{Contour plots of the true and estimated spectral density of the FAR(1) at different time points at frequencies  $\lambda=0$, $\lambda=\frac{3}{10}\pi$ and $\lambda=\frac{9}{10}\pi$.}}\label{table:FAR1nopeak}
\end{figure}
\end{center}
As a second example, we consider a FAR(2) with the location of the peak varying with time. More specifically, the Fourier coefficients are now obtained by means of a VAR(2)
\begin{equation*}
\boldsymbol{X}^{(T)}_{t} = \mathfrak{B}_{\frac{t}{T},1}\boldsymbol{X}^{(T)}_{t-1}+\mathfrak{B}_{\frac{t}{T},2}\boldsymbol{X}^{(T)}_{t-2}+\boldsymbol{\varepsilon}_t,
\end{equation*}
where $\mathfrak{B}_{u,1}=\eta_{u,1} A_{u,1}/\snorm{A_{u,1}}_{\infty}$ and $\mathfrak{B}_{u,2}=\eta_{u,2} A_{u,2}/\snorm{A_{u,2}}_{\infty}$.  The entries of the matrices $A_{u,1}$ and $A_{u,2}$ are mutually independent and are generated such that ${[A_{u,1}]}_{i,j} = \mathcal{N}(0, e^{-(i-3)-(j-3)})$ and ${[A_{u,2}]}_{i,j} = \mathcal{N}(0, {(i^{8/2}+j^{2/2})}^{-1})$, respectively. The norms are specified as 
\begin{align*} \eta_{u,1} = 0.4 \cos(1.5-\cos(\pi u)) \text{ and } \eta_{u,2} = -0.5. \end{align*} This will result in the peak to be located at $\lambda = \arccos(0.3\cos[1.5-\cos(\pi u )])$. The collection of innovation functions $\{\varepsilon_t\}$ is chosen such that the $l$-th coefficient $\langle \varepsilon_{t}, \psi_l \rangle$ has variance $1/{[(l-2.65)\pi]}^2$. Figure \ref{table:FAR2peak} provides the contour plots for different local time values where the frequency was set to $\lambda= 1.5-\cos(\pi u)$, i.e., the direction in which most change in time-direction is visible in terms of amplitude. We observe good results in terms of identifying the peaks and valleys overall where again the variability clearly reduces for $T>512$. For the value $u=0.5$,  one is really close to the location of a peak and observe wrongful detection of a small peak in the middle of the contour plot. This is an indication some over-smoothing occurs which, to some extent, is difficult to prevent for autoregressive models, even in the stationary time series case.

 \vspace*{-25pt}
\begin{center}
\begin{figure}[!h]
 \hspace*{-0cm}
\begin{tabular}{c|c|ccc}
& true  & $ T = 2^{9} $ & $T=2^{12}$  & $T = 2^{16}$ \\ \hline
\begin{sideways} \footnotesize{\hspace{0.75cm} $u=0.1$} \end{sideways} &\includegraphics[width=0.22\textwidth]{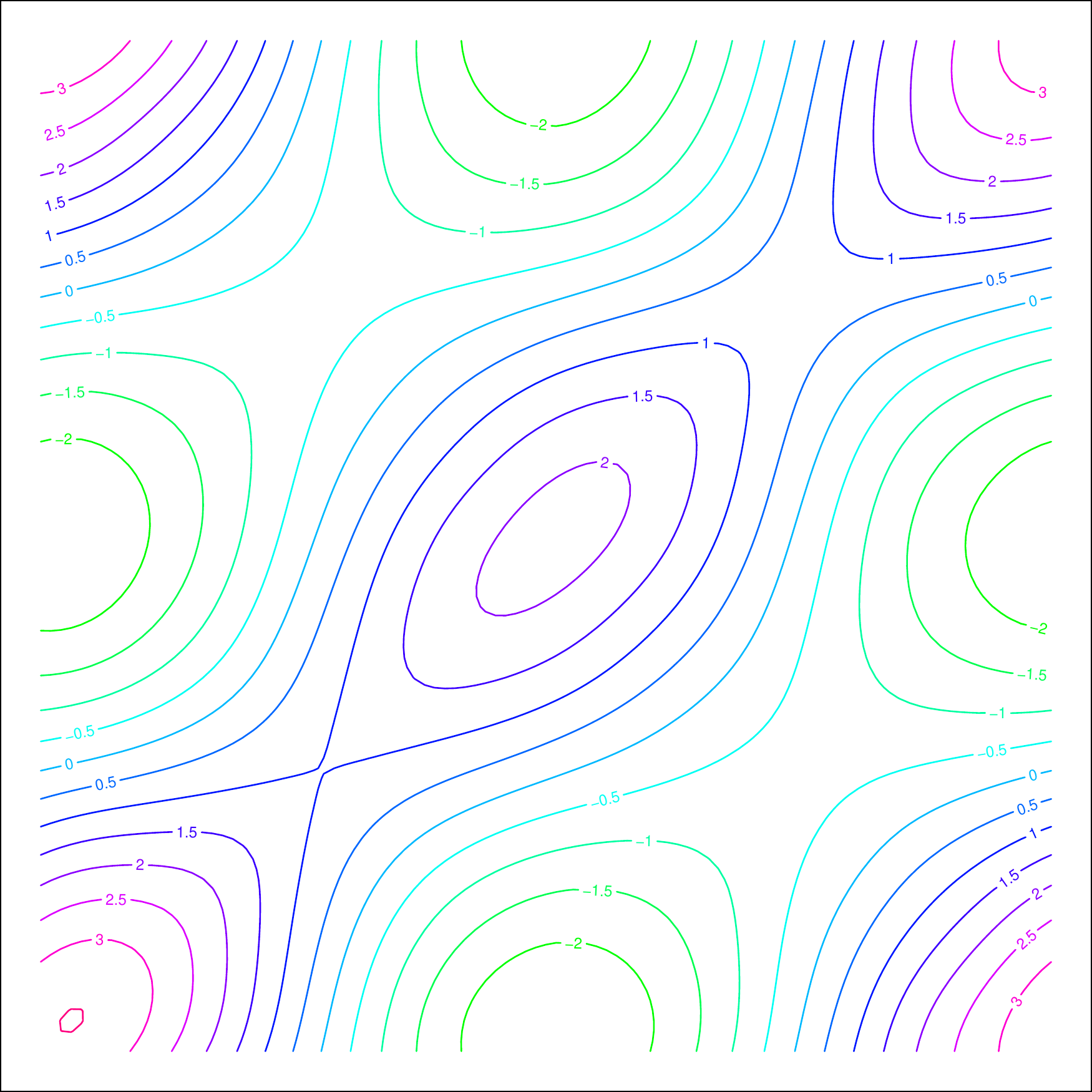}& \includegraphics[width=0.22\textwidth]{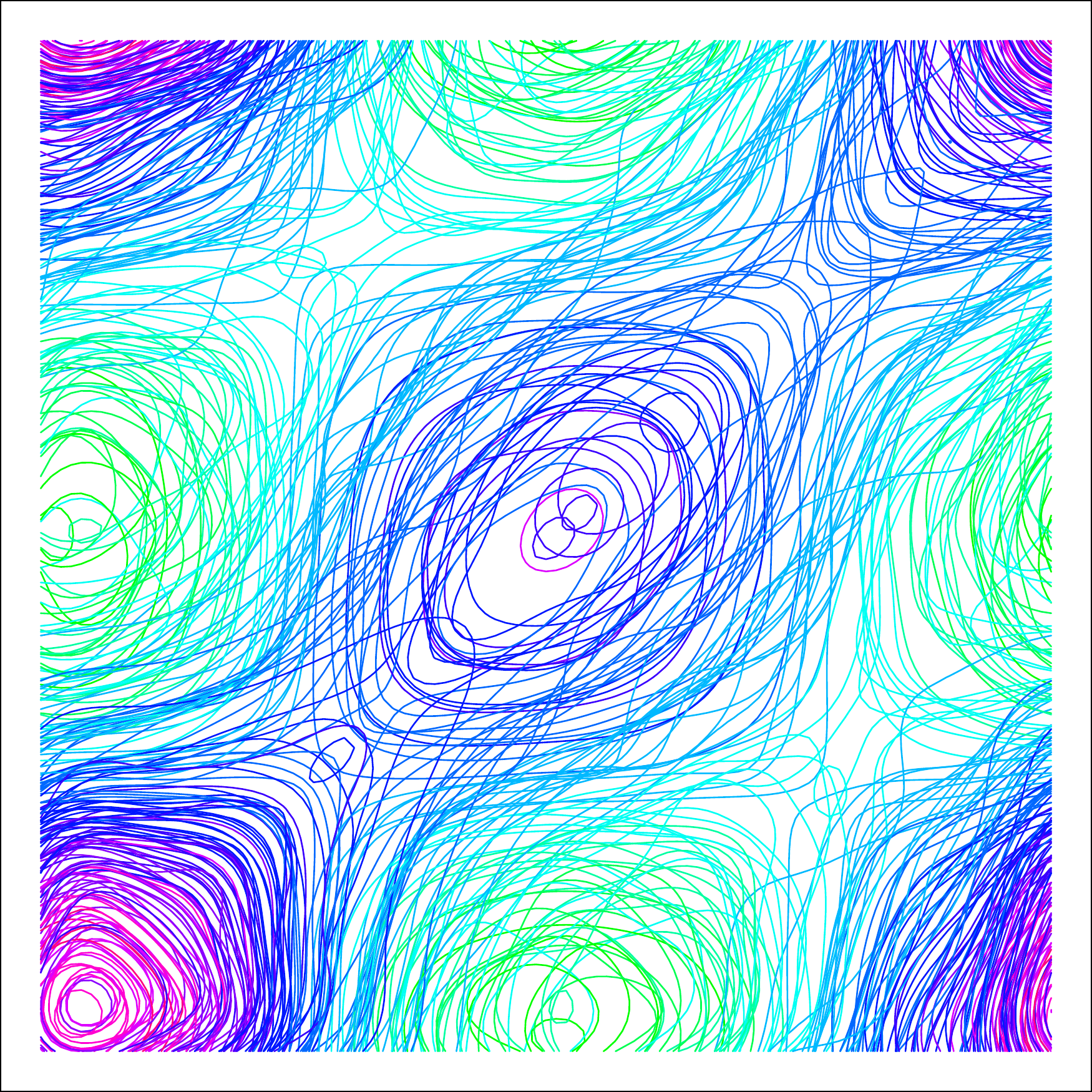}& \includegraphics[width=0.22\textwidth]{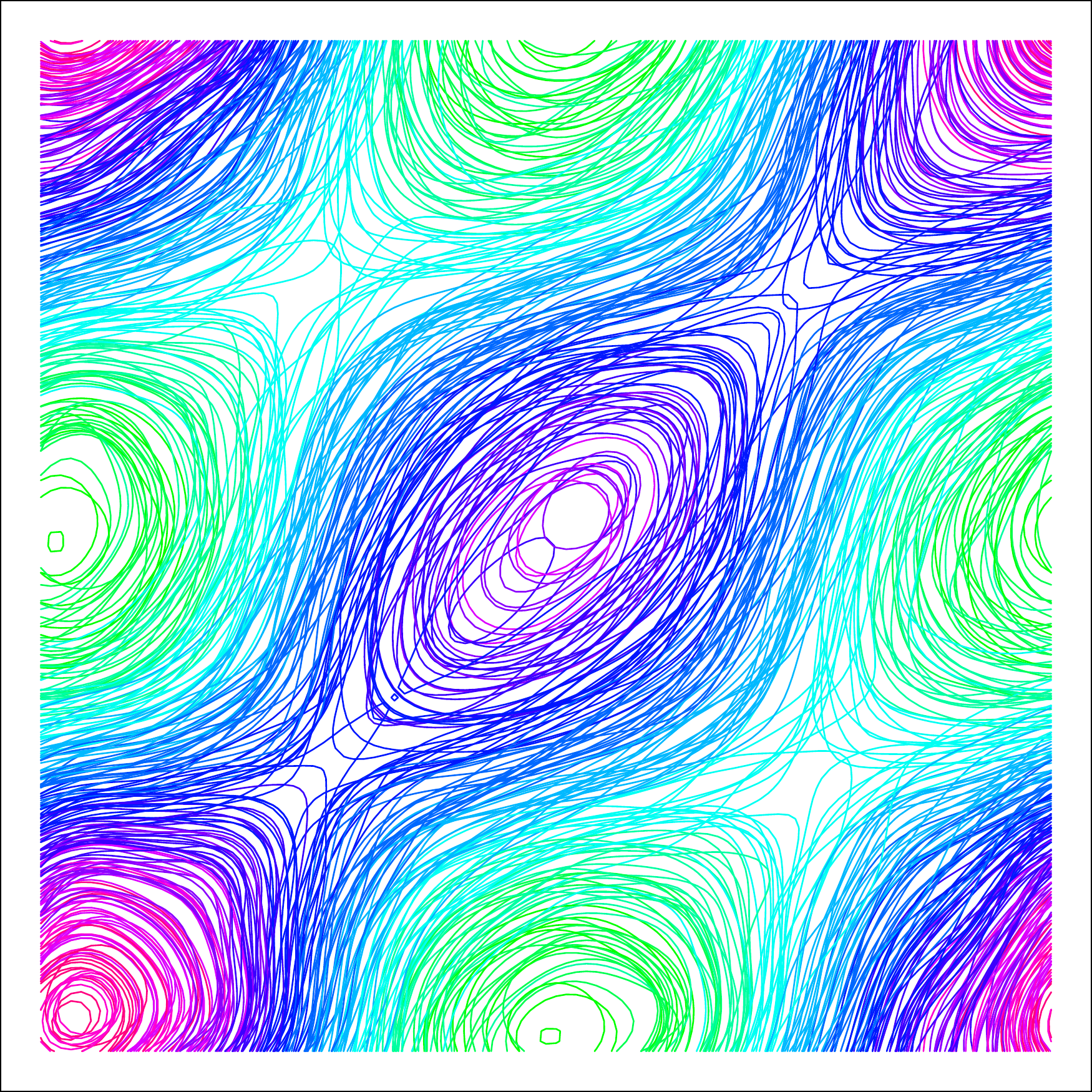}&\includegraphics[width=0.22\textwidth]{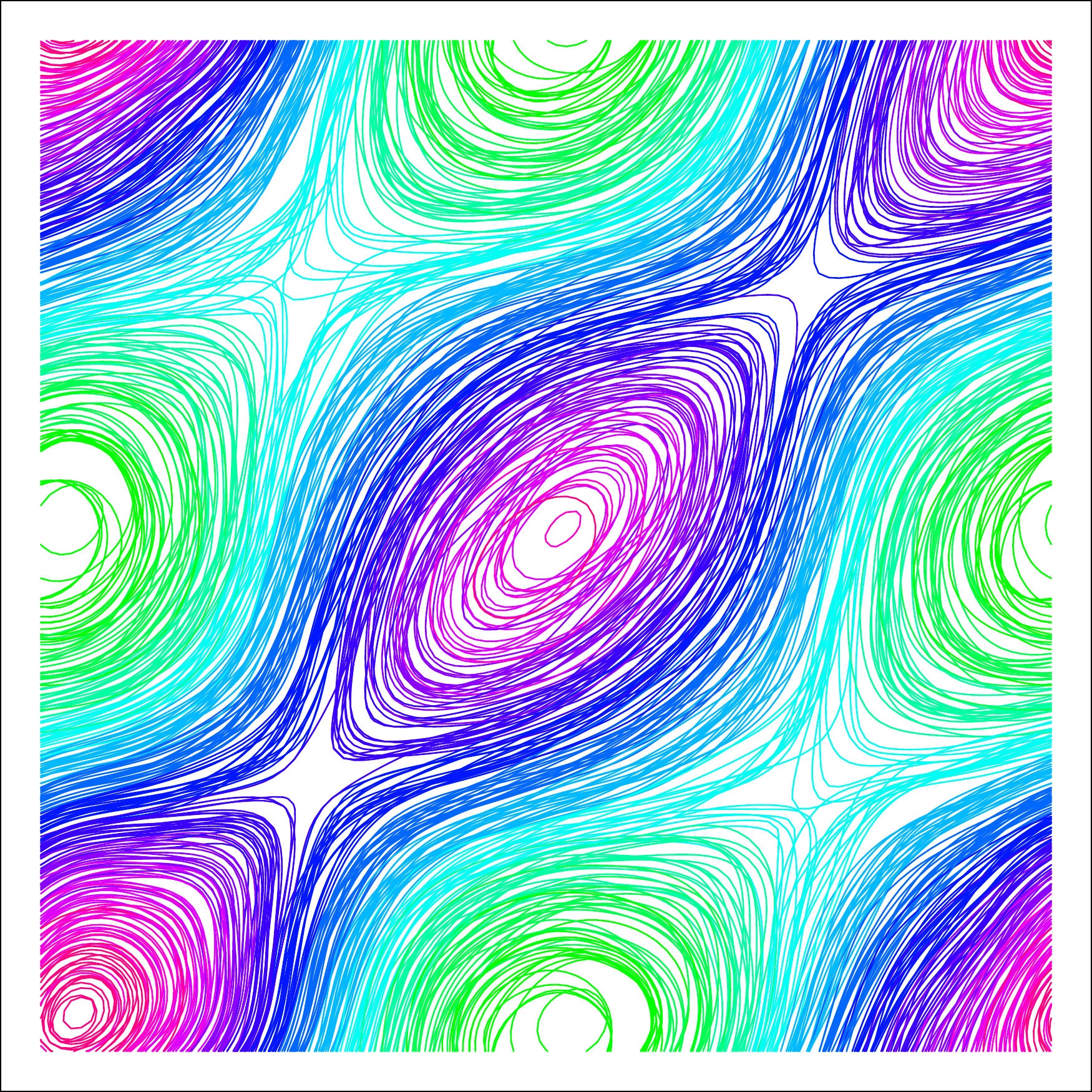} \\ 
\begin{sideways} \footnotesize{\hspace{0.75cm} $u=0.25$} \end{sideways} &\includegraphics[width=0.22\textwidth]{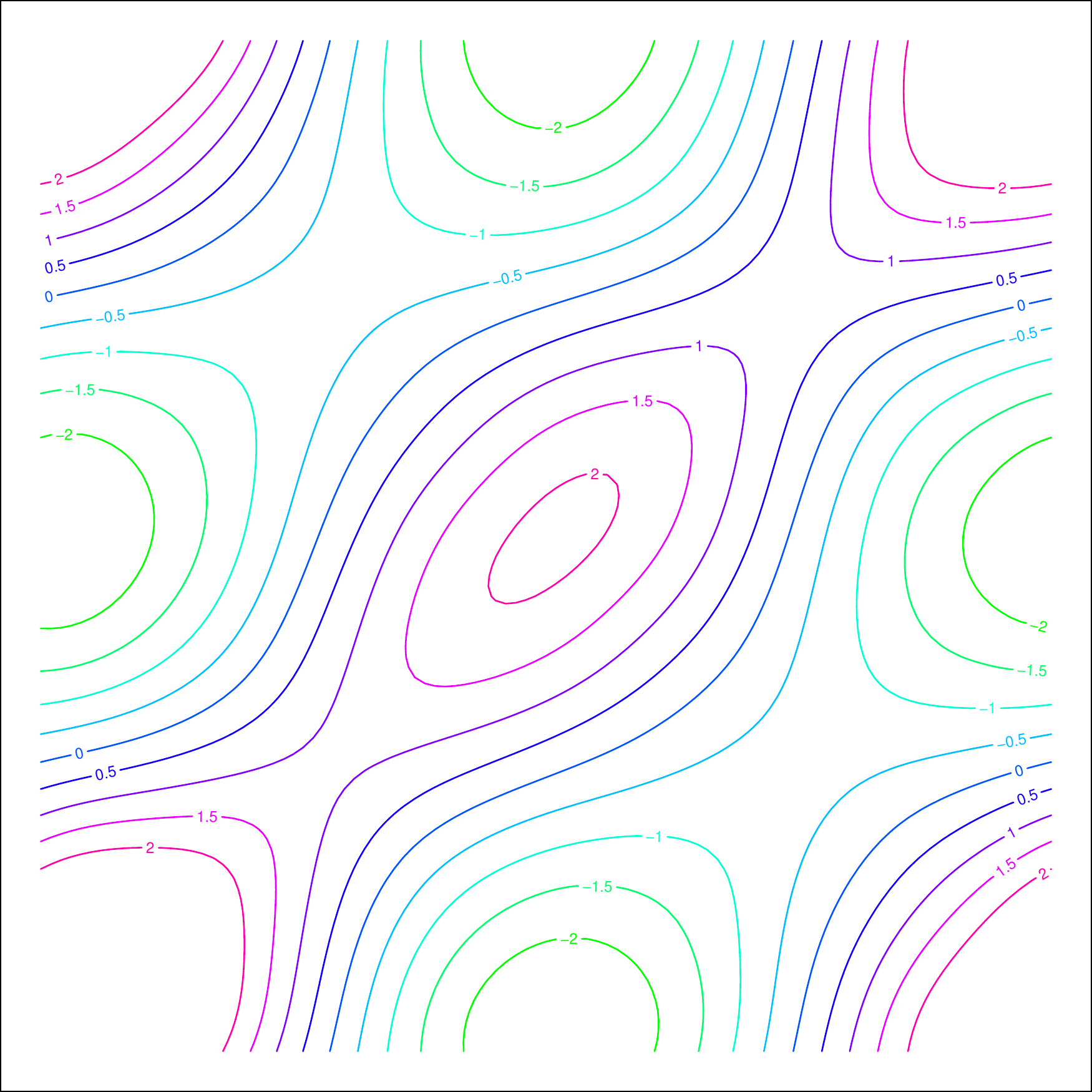}& \includegraphics[width=0.22\textwidth]{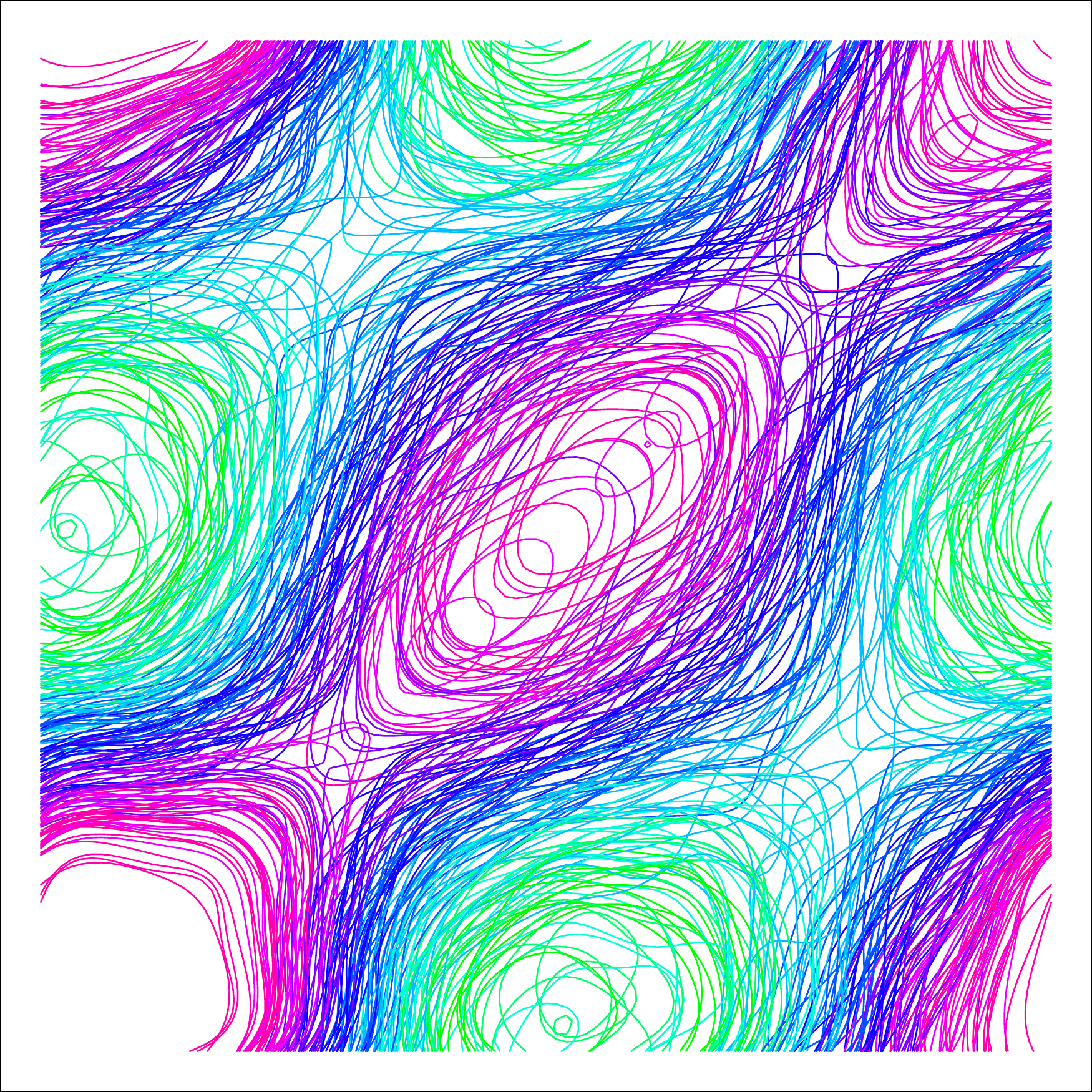}& \includegraphics[width=0.22\textwidth]{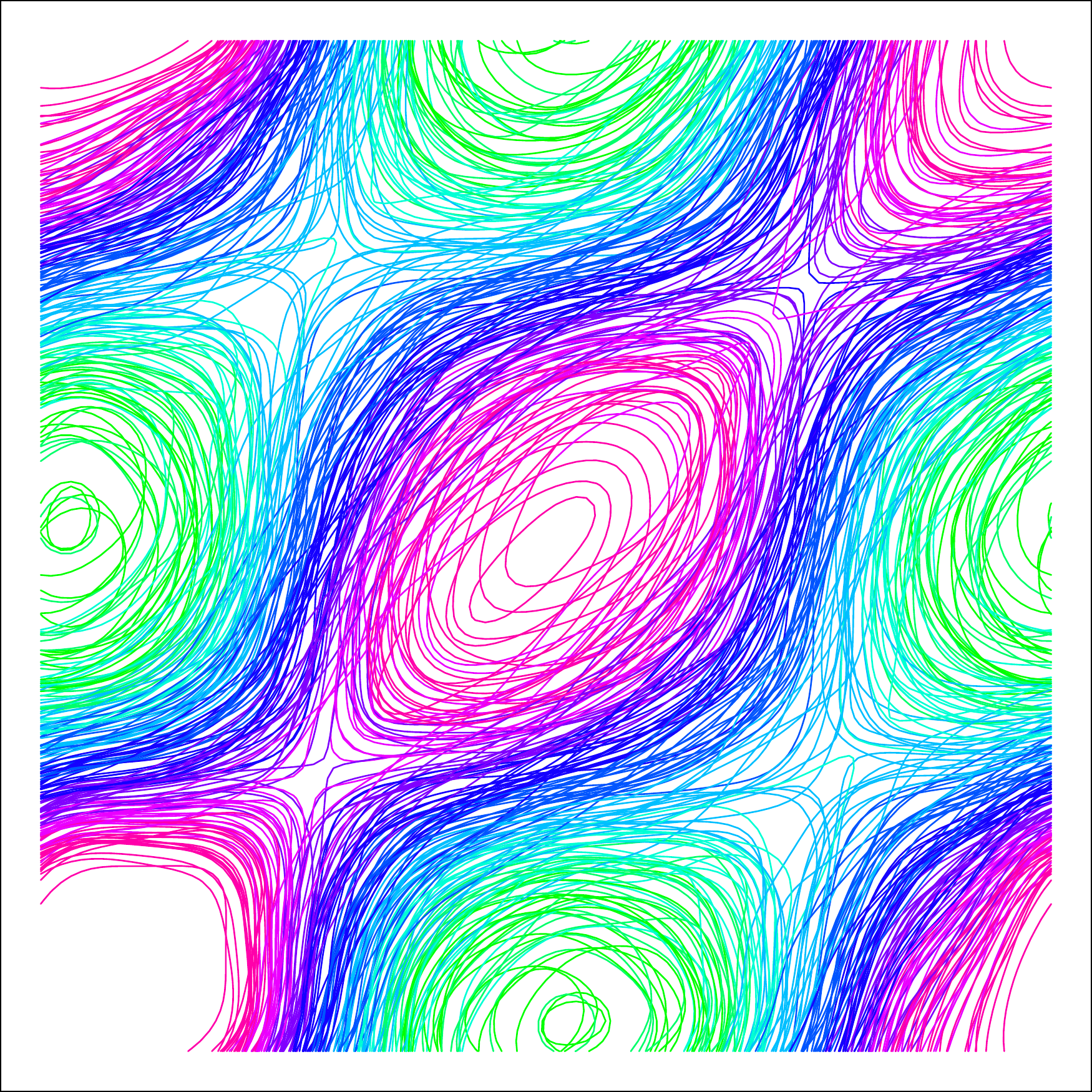}&\includegraphics[width=0.22\textwidth]{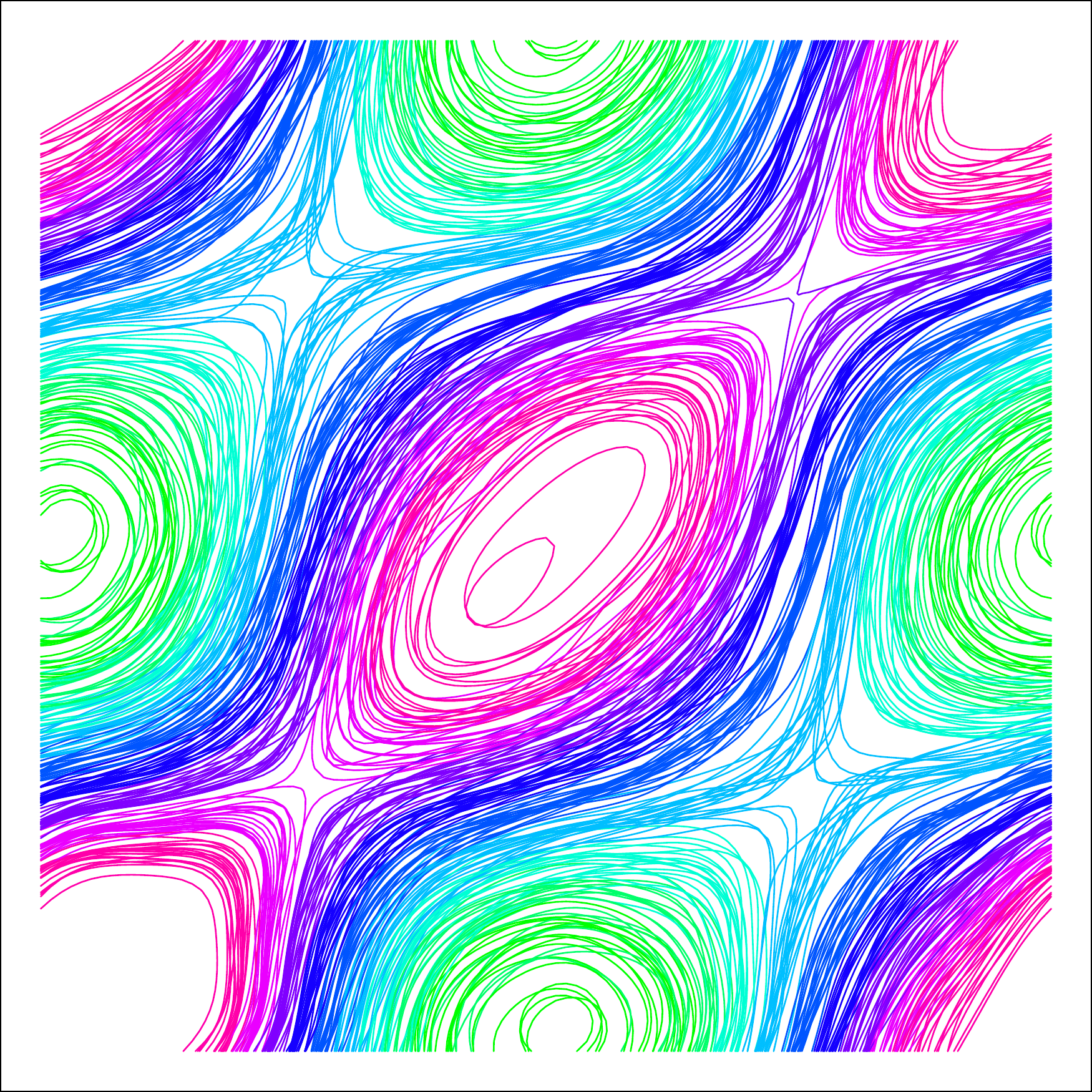} \\ 
\begin{sideways} \footnotesize{\hspace{0.75cm} $u=0.375$} \end{sideways} &\includegraphics[width=0.22\textwidth]{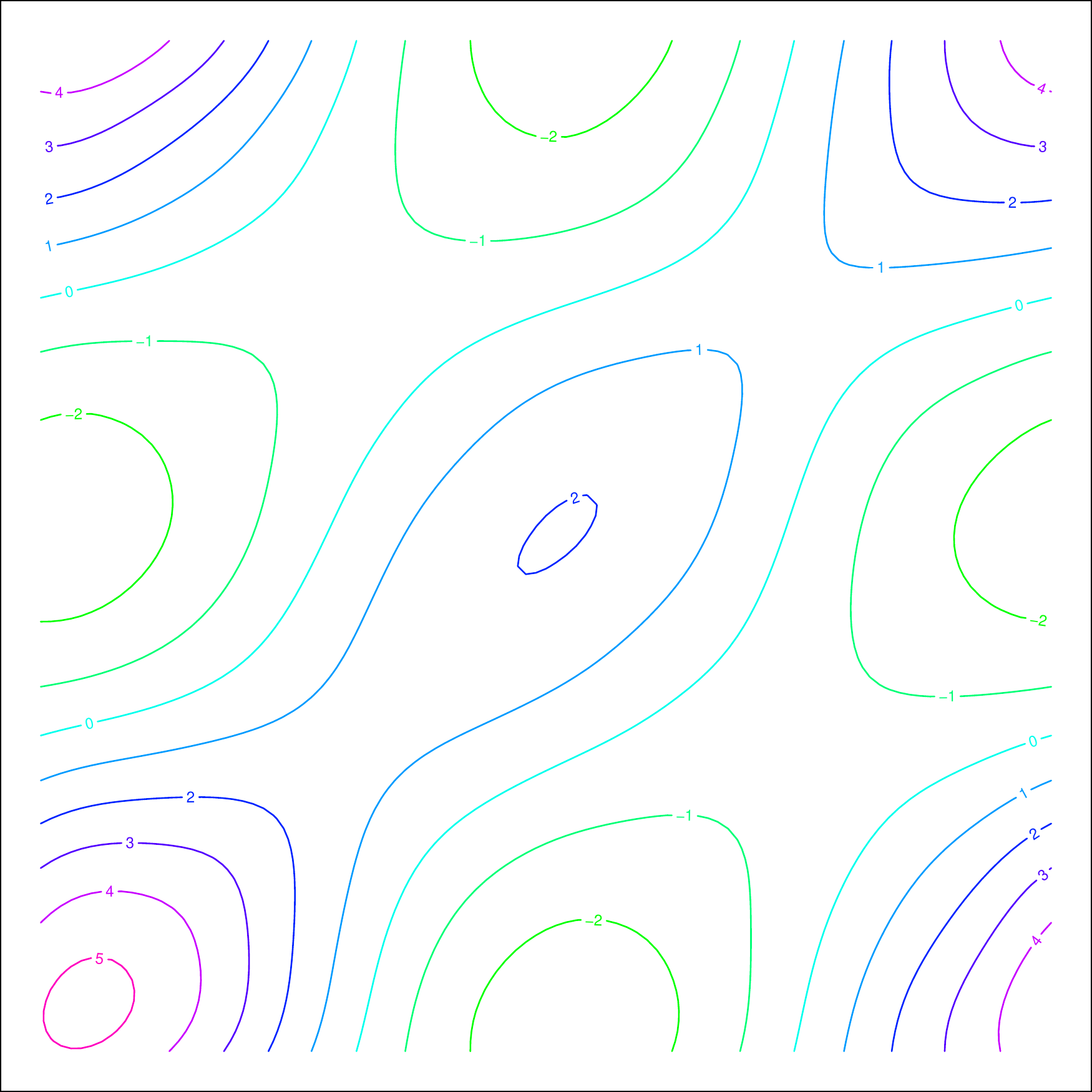}& \includegraphics[width=0.22\textwidth]{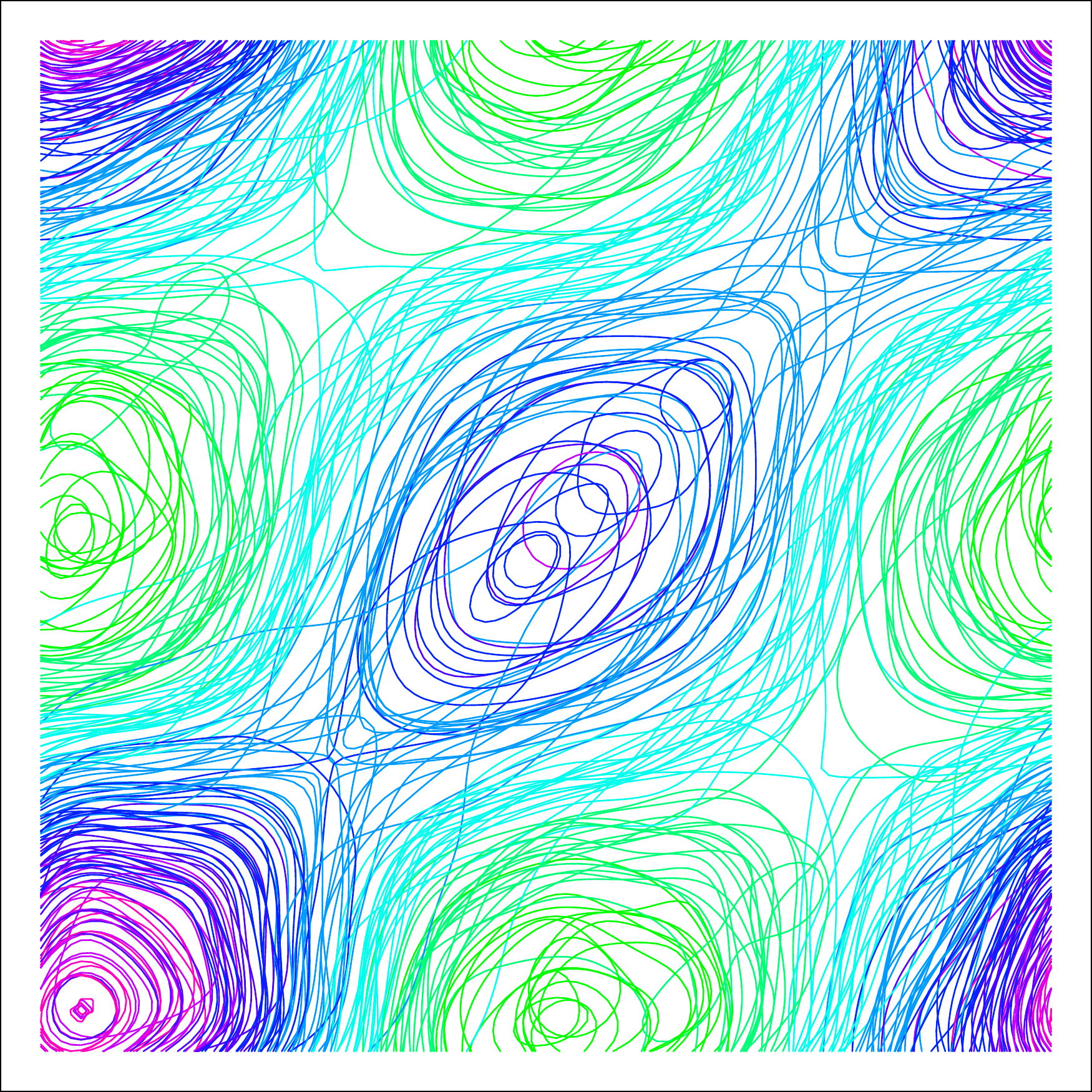}& \includegraphics[width=0.22\textwidth]{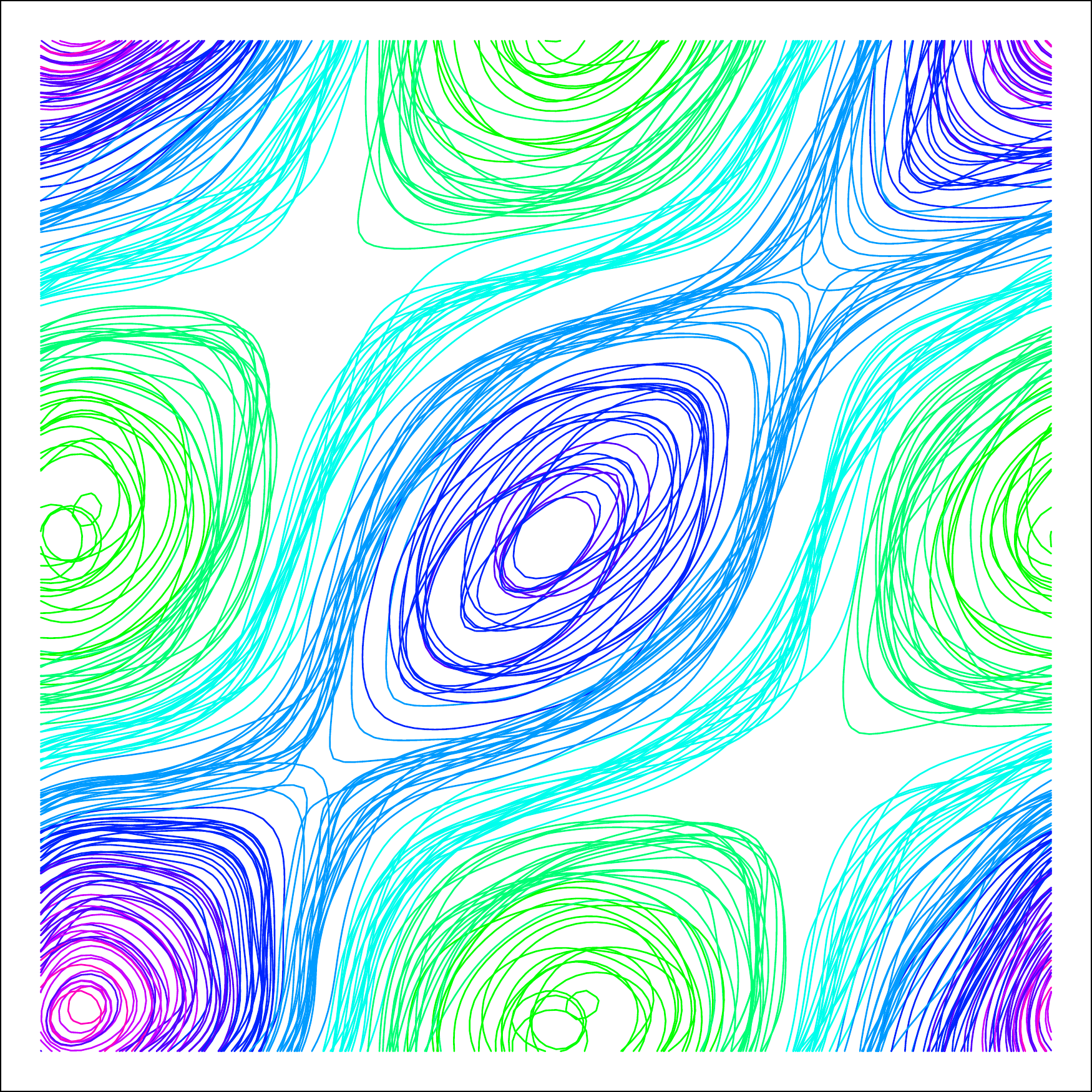}&\includegraphics[width=0.22\textwidth]{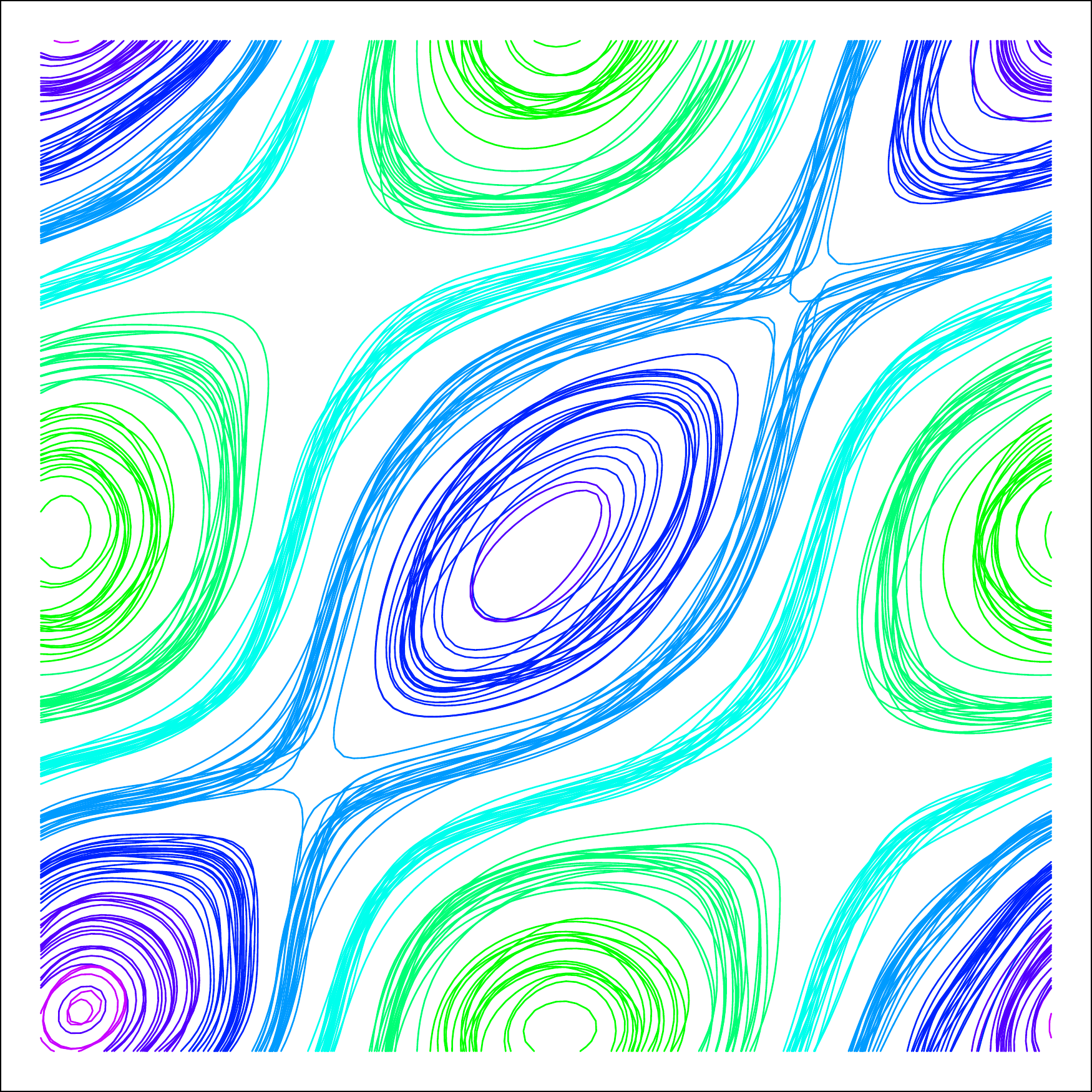} \\ 
\begin{sideways} \footnotesize{\hspace{0.75cm} $u=0.5$} \end{sideways} &\includegraphics[width=0.22\textwidth]{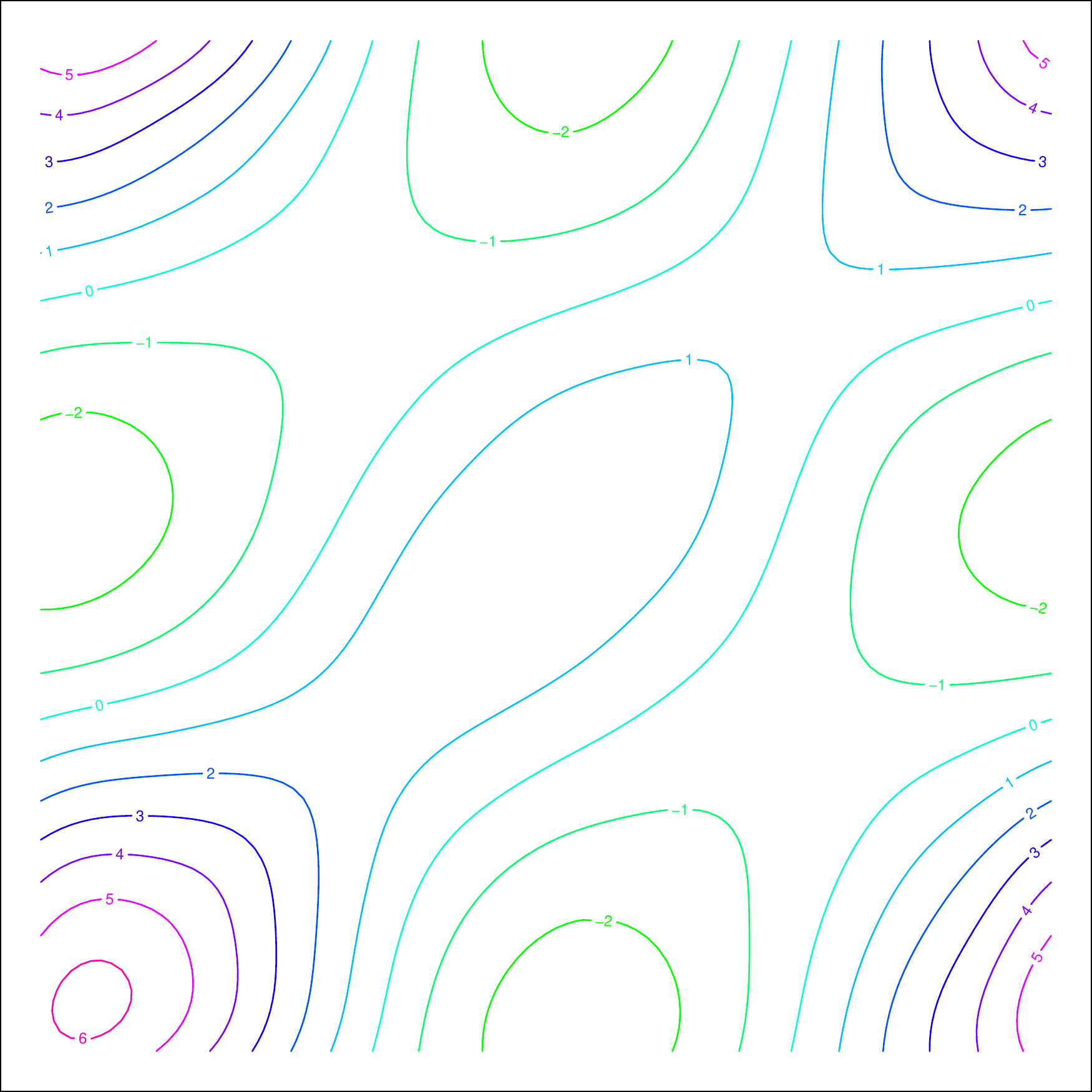}& \includegraphics[width=0.22\textwidth]{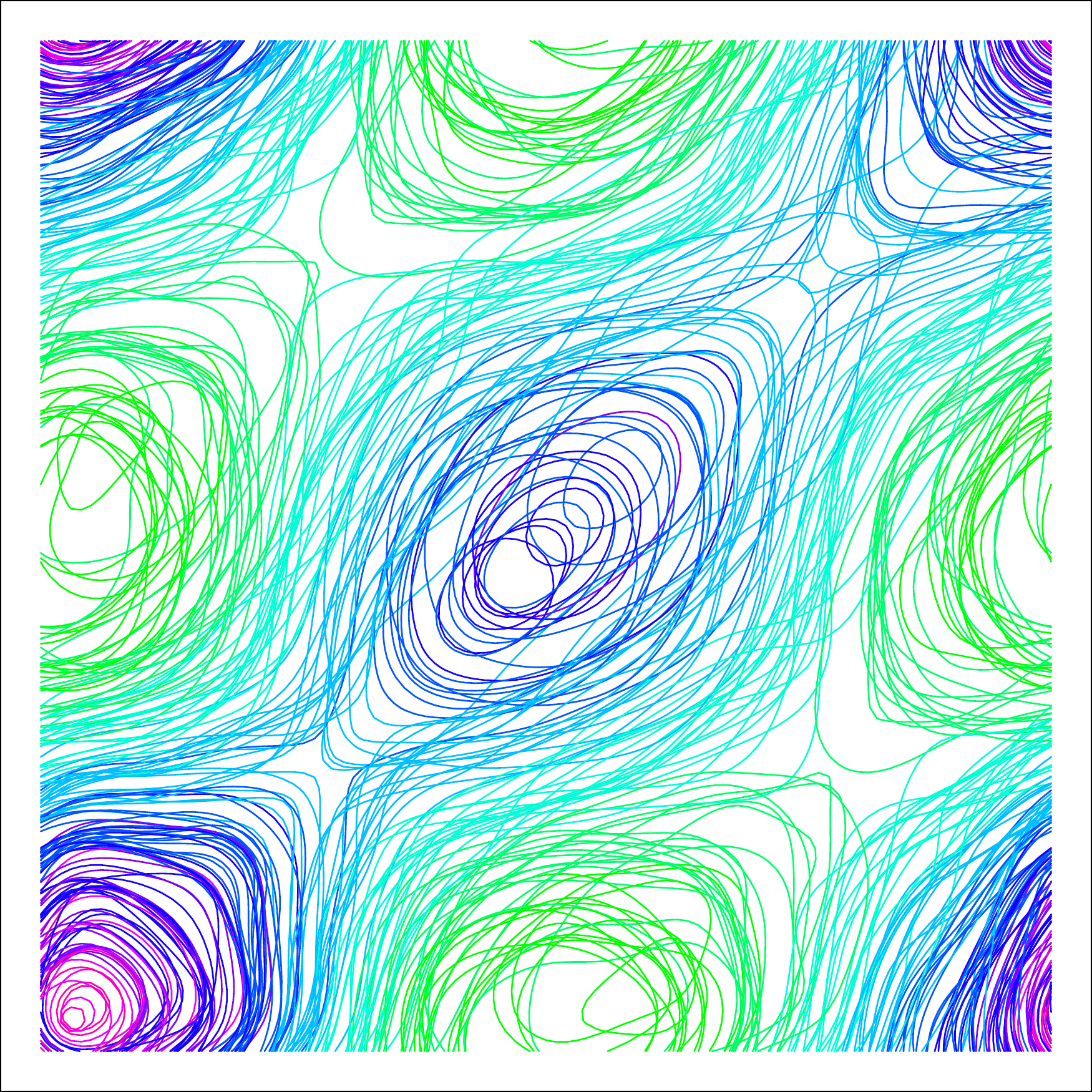}& \includegraphics[width=0.22\textwidth]{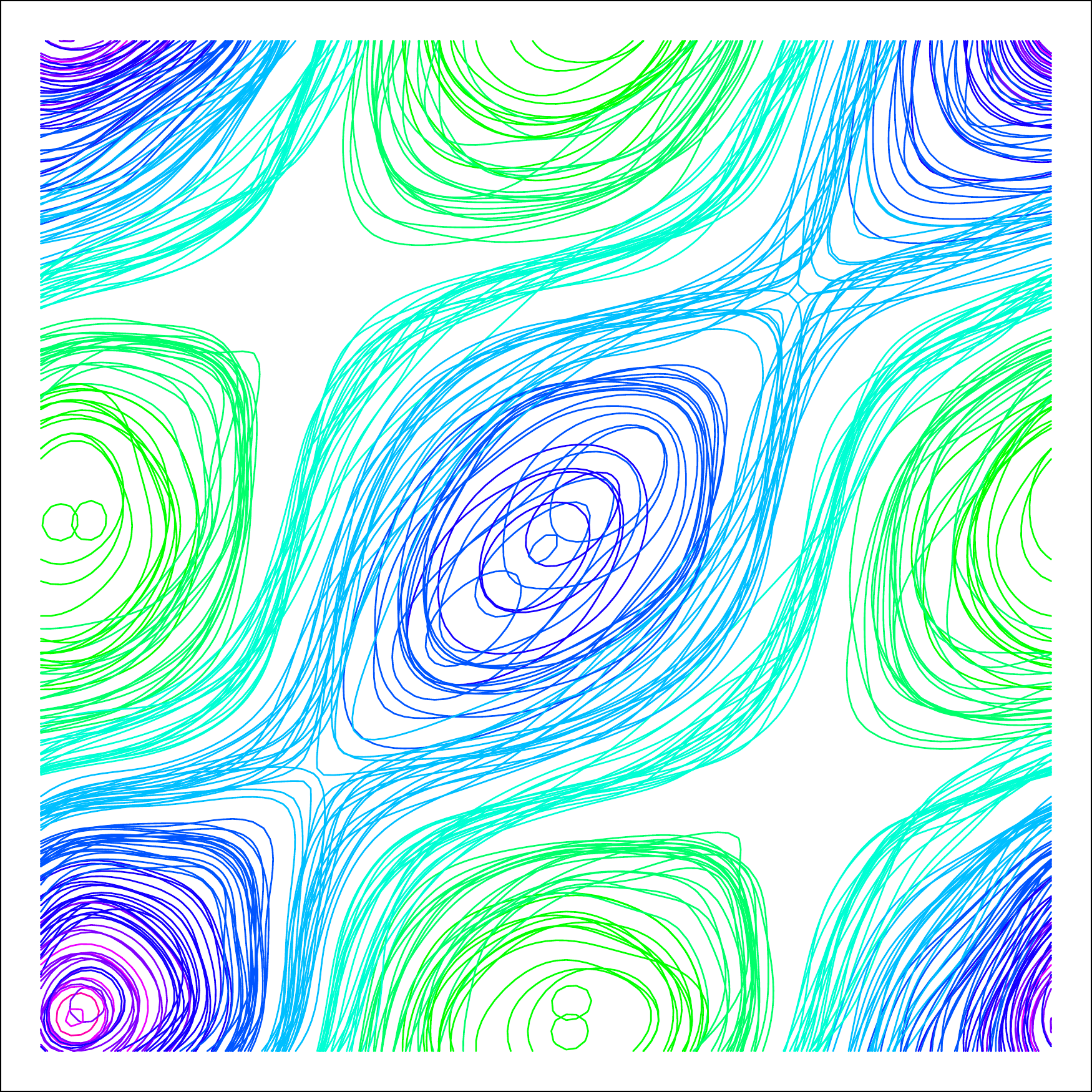}&\includegraphics[width=0.22\textwidth]{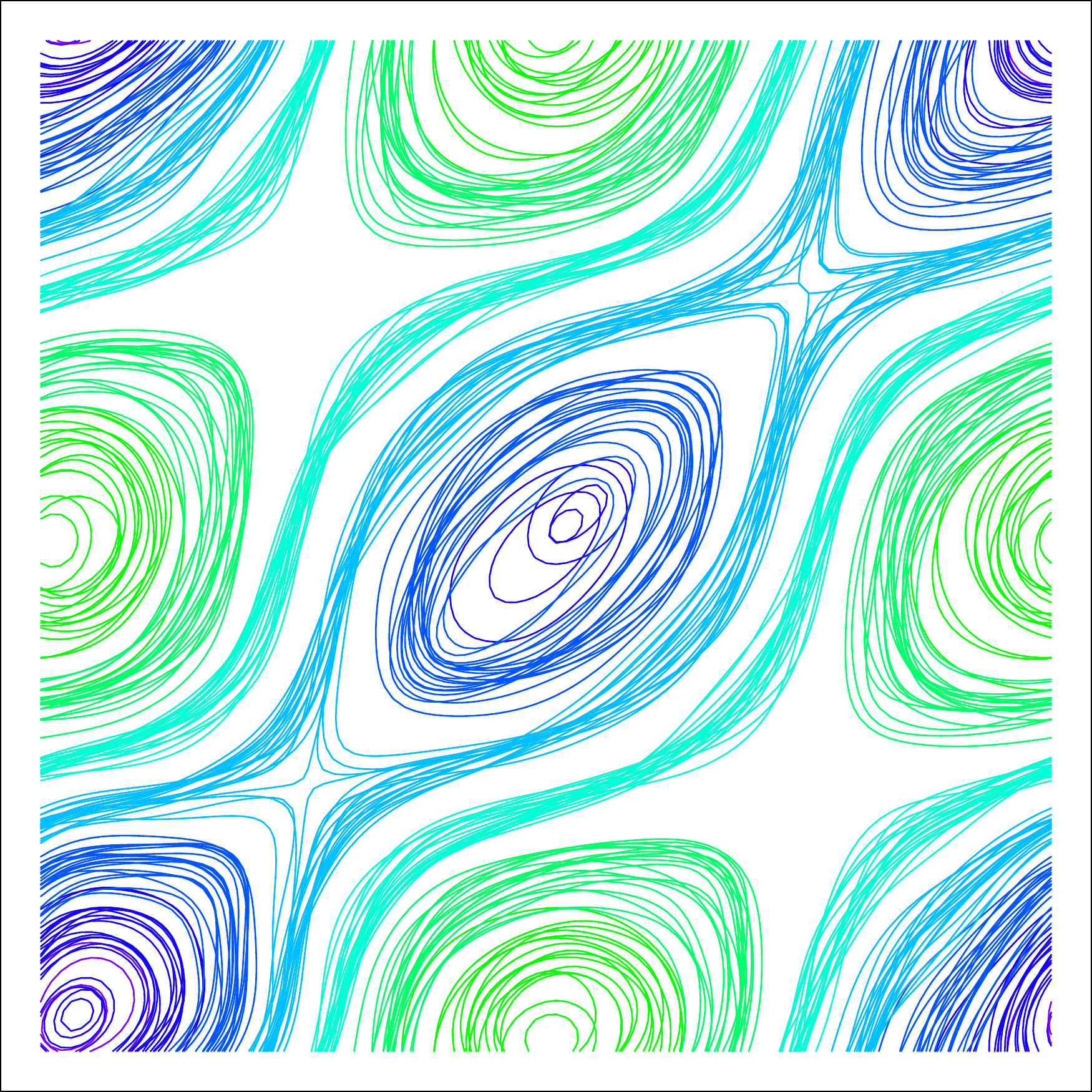} \\
\begin{sideways} \footnotesize{\hspace{0.75cm} $u=0.625$} \end{sideways}&\includegraphics[width=0.22\textwidth]{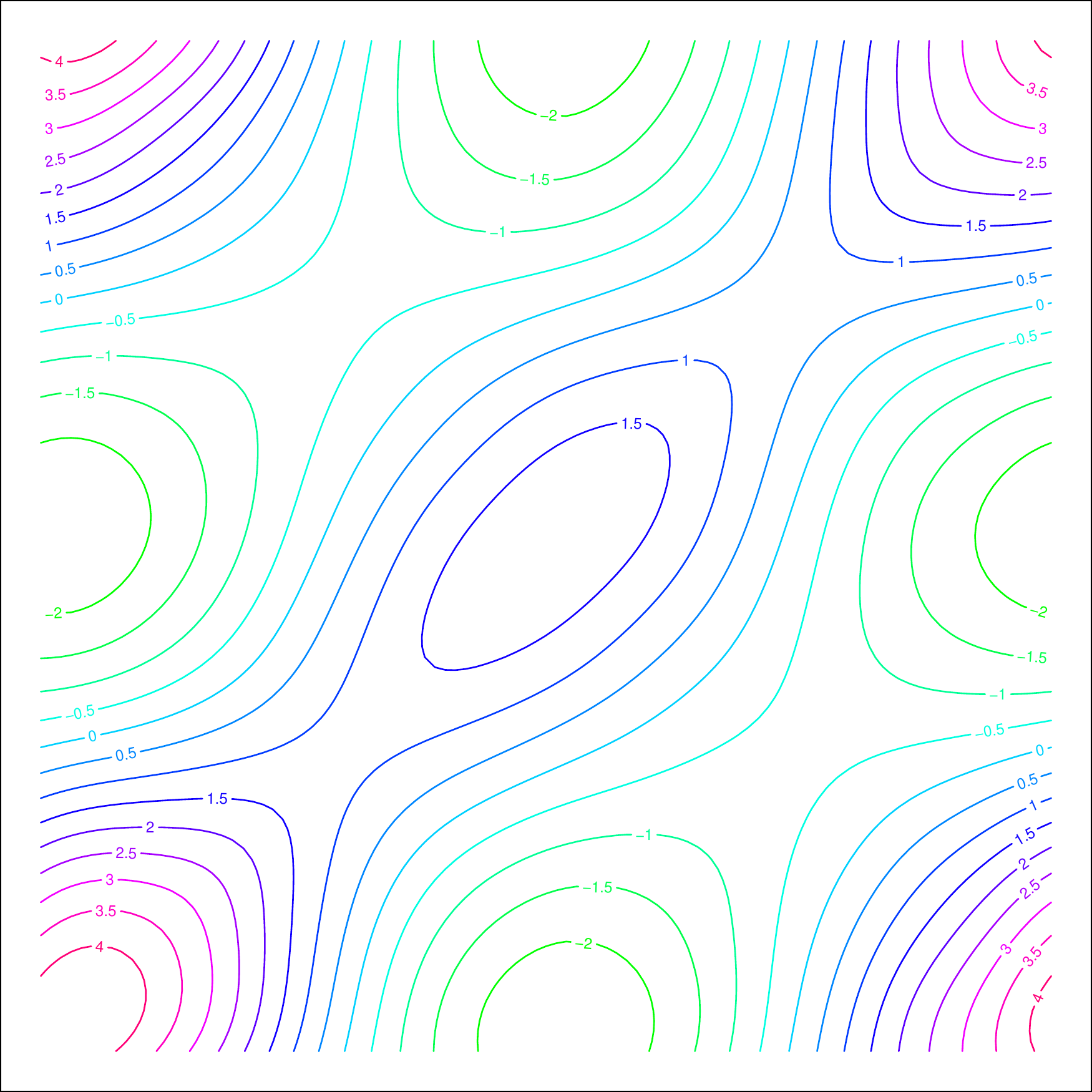}& \includegraphics[width=0.22\textwidth]{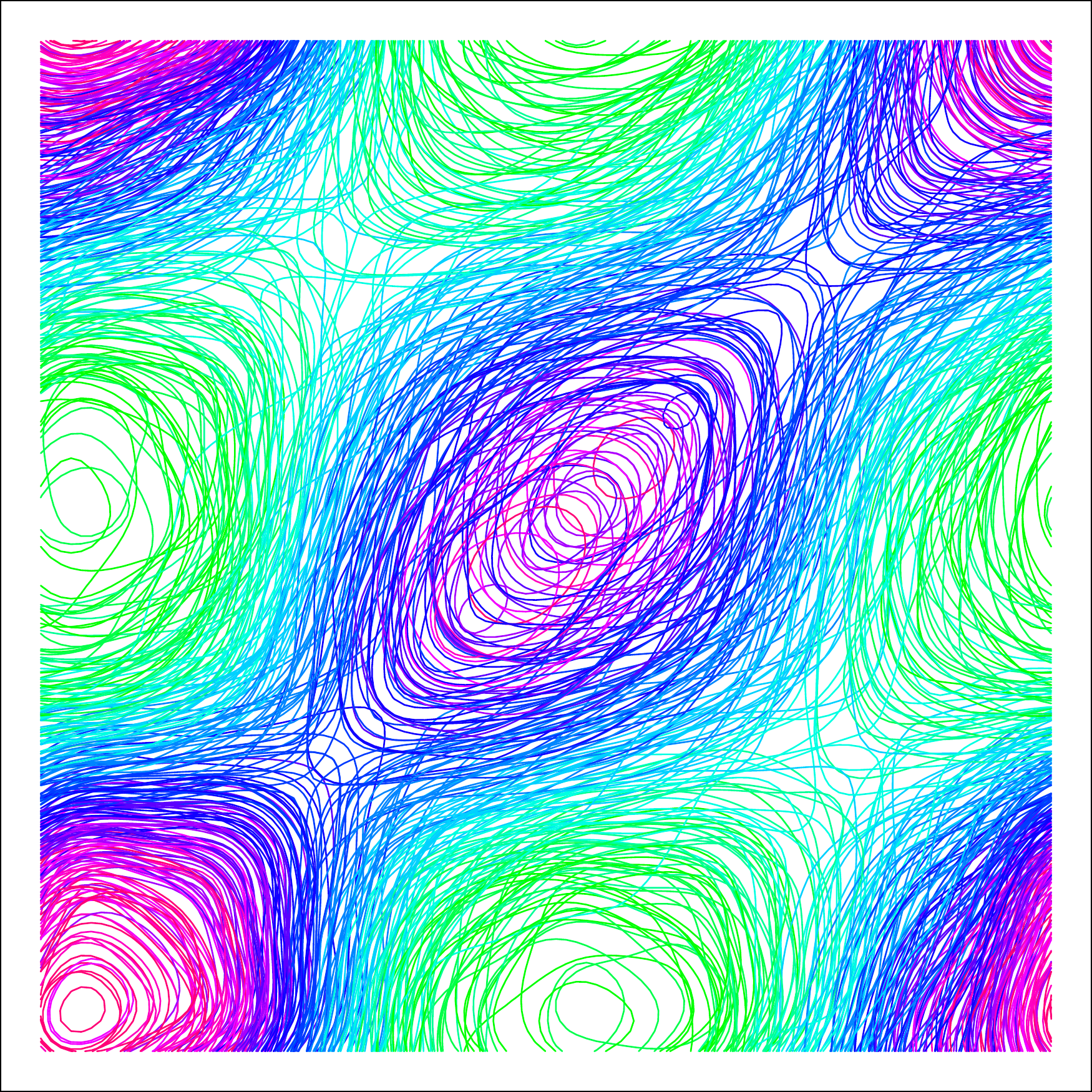}& \includegraphics[width=0.22\textwidth]{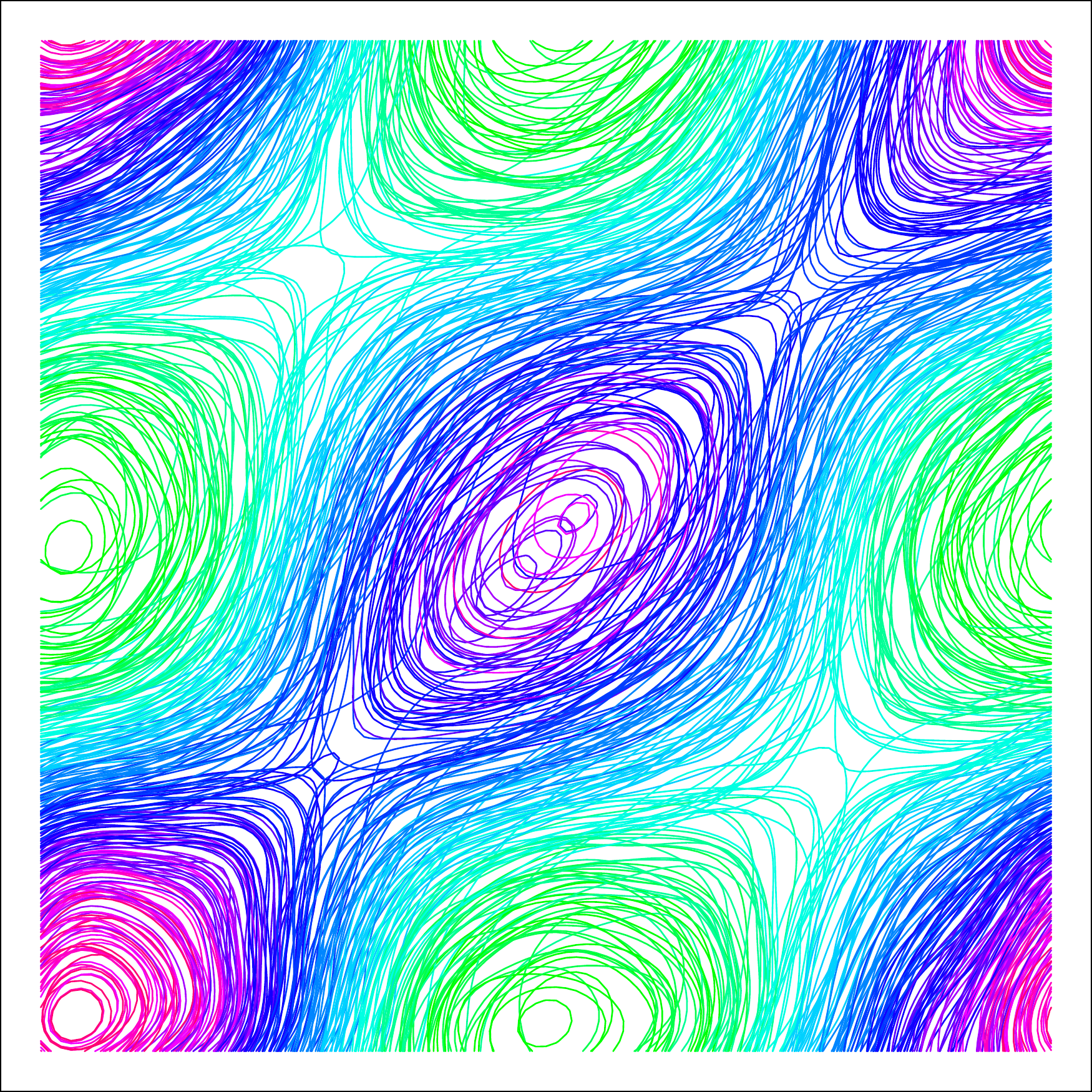}&\includegraphics[width=0.22\textwidth]{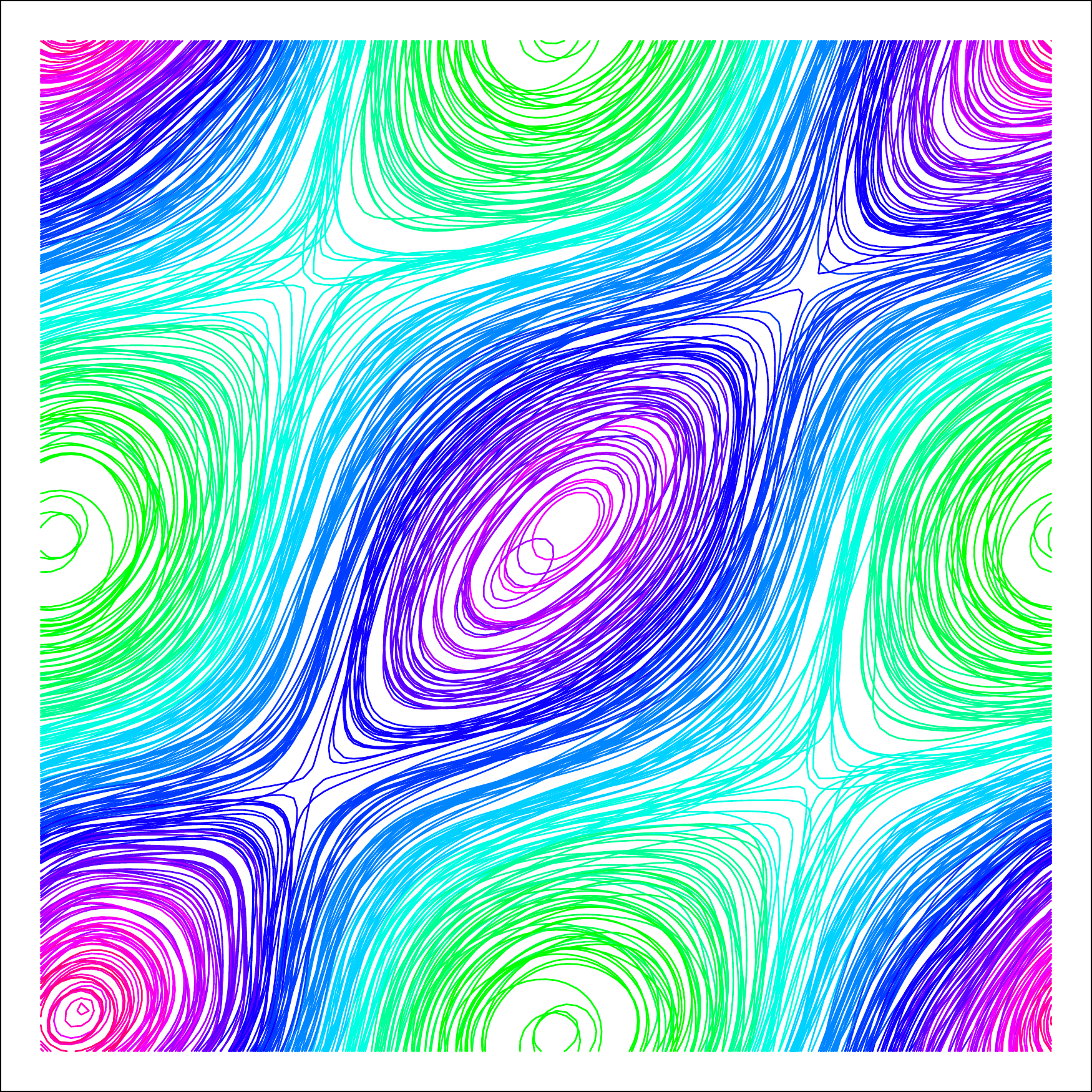} \\
\begin{sideways} \footnotesize{\hspace{0.75cm} $u=0.75$} \end{sideways}  &\includegraphics[width=0.22\textwidth]{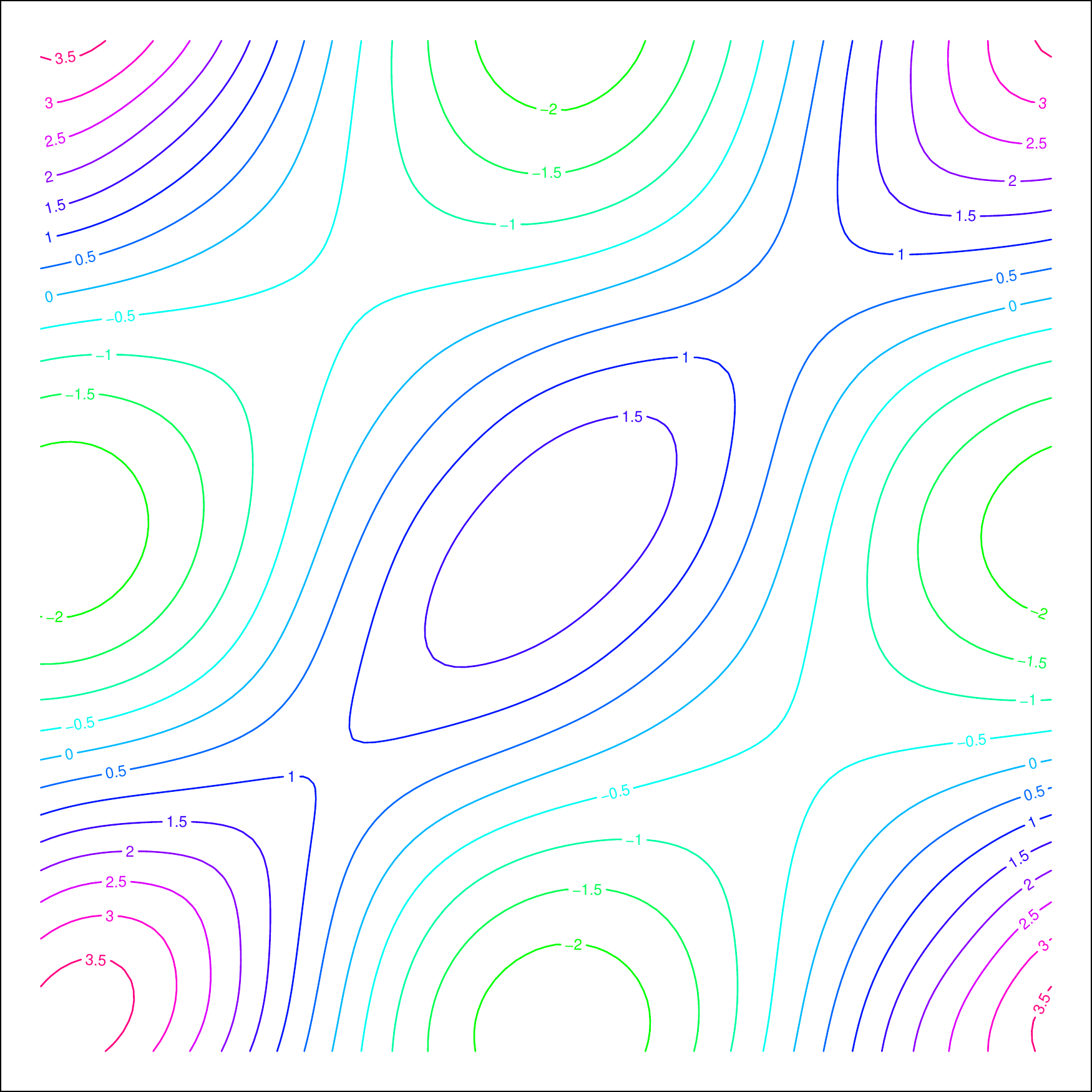}& \includegraphics[width=0.22\textwidth]{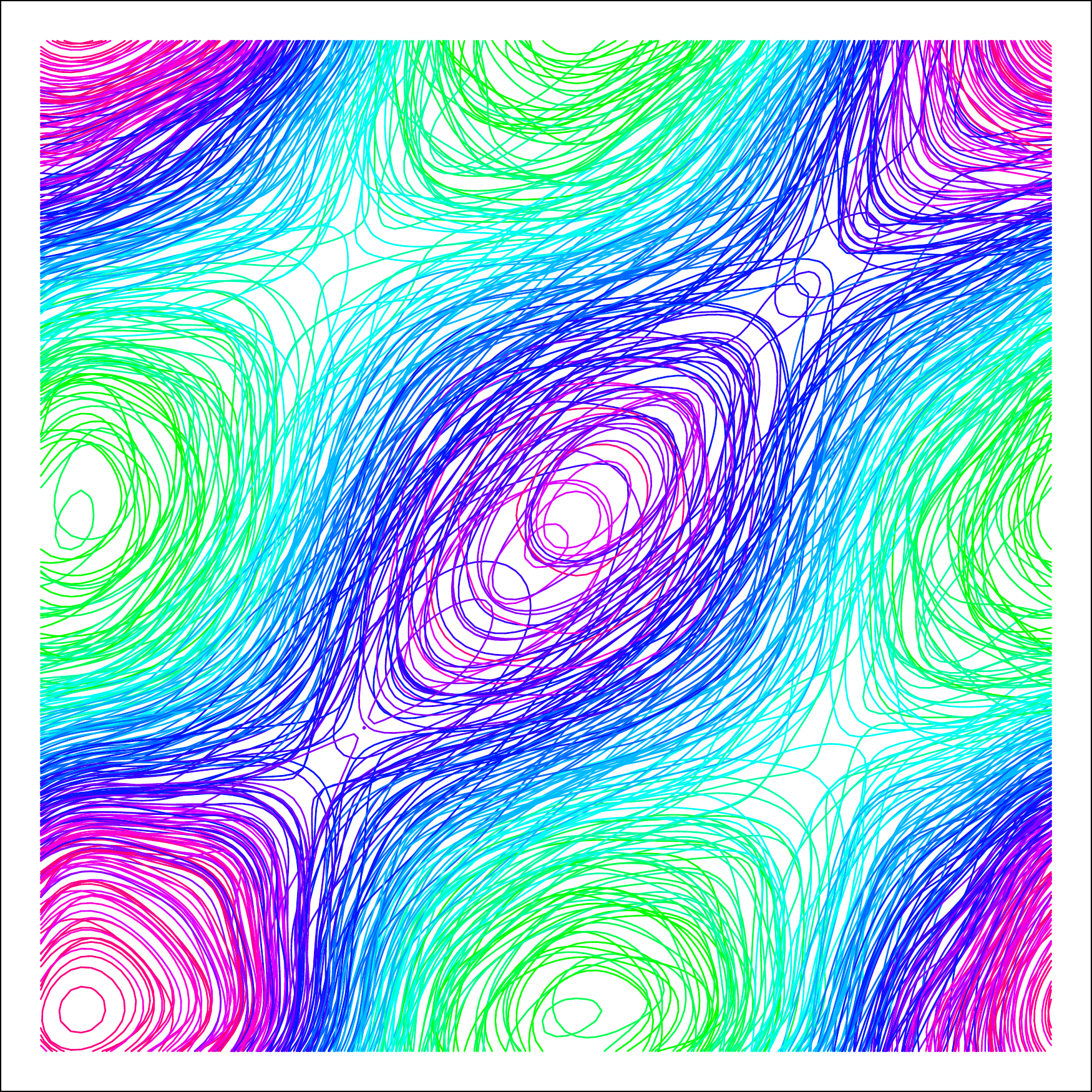}& \includegraphics[width=0.22\textwidth]{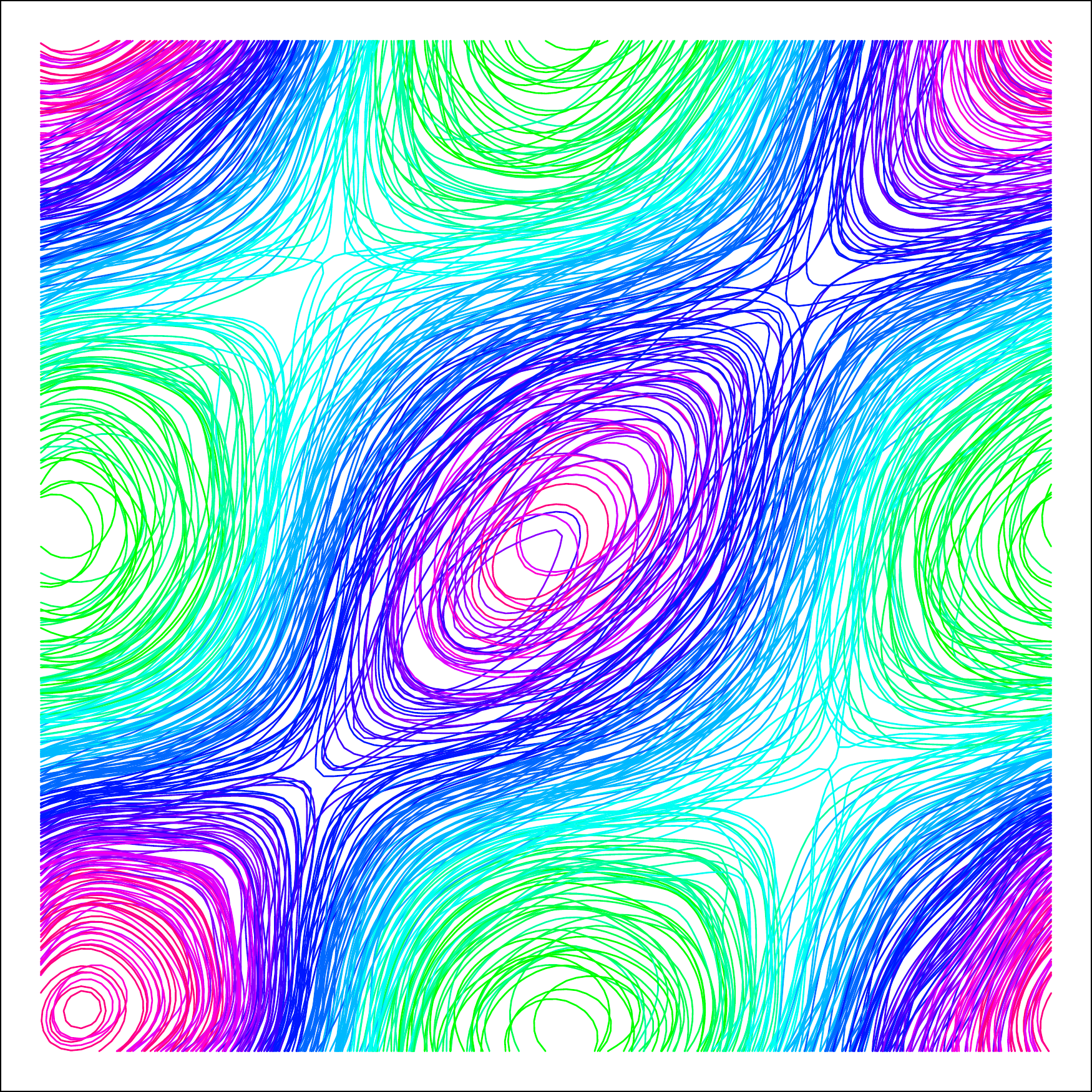}&\includegraphics[width=0.22\textwidth]{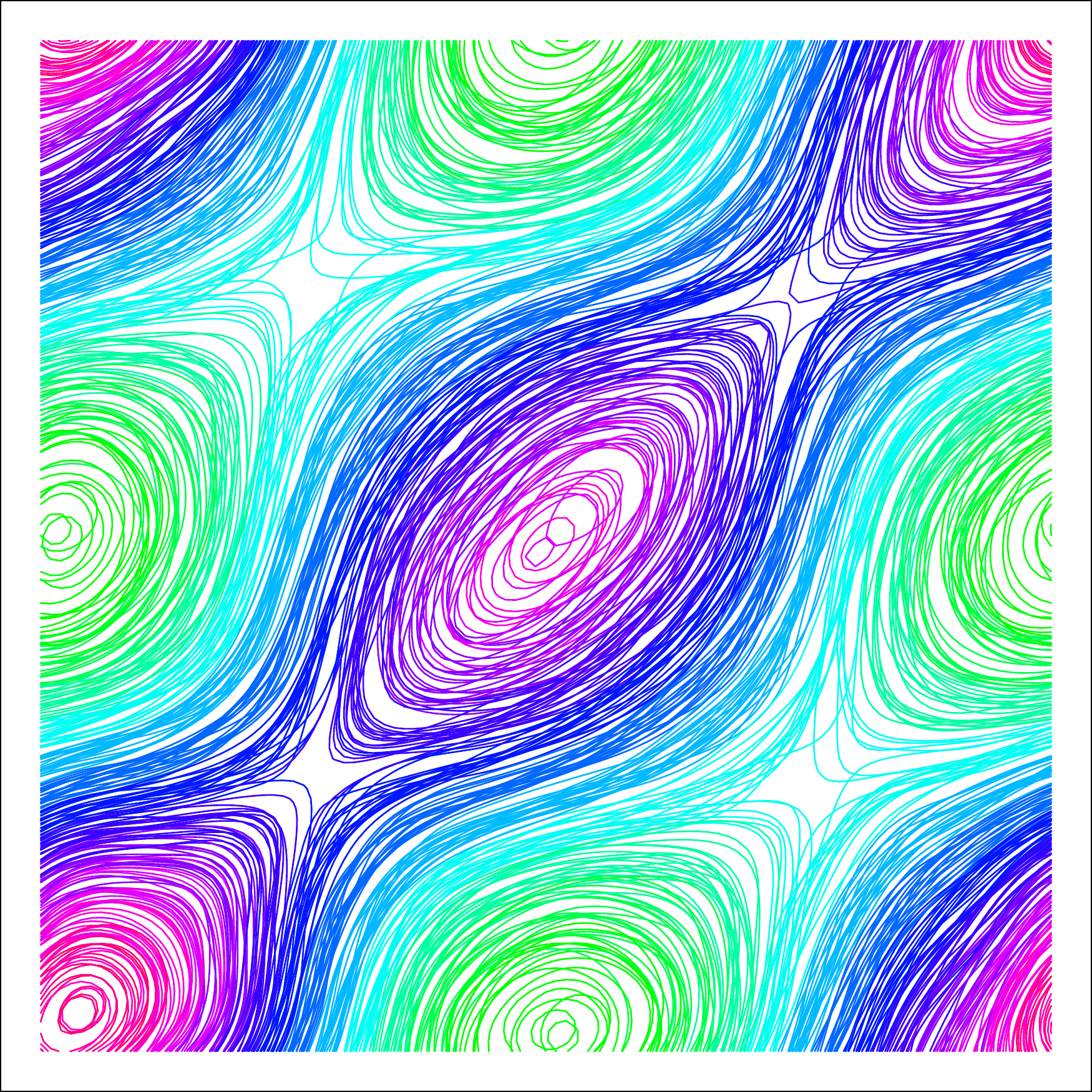} \\
\begin{sideways} \footnotesize{\hspace{0.75cm} $u=0.9$} \end{sideways}   &\includegraphics[width=0.22\textwidth]{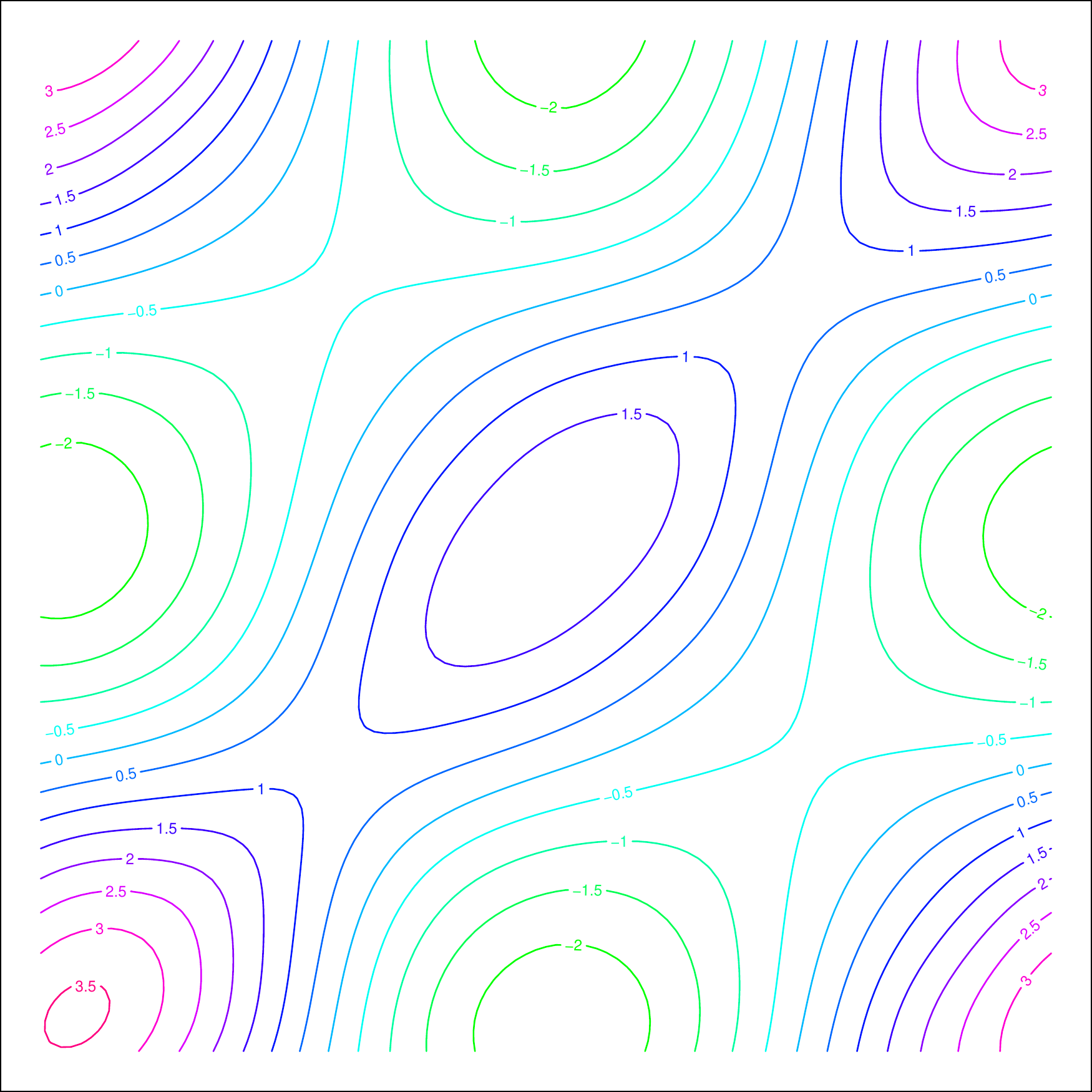}& \includegraphics[width=0.22\textwidth]{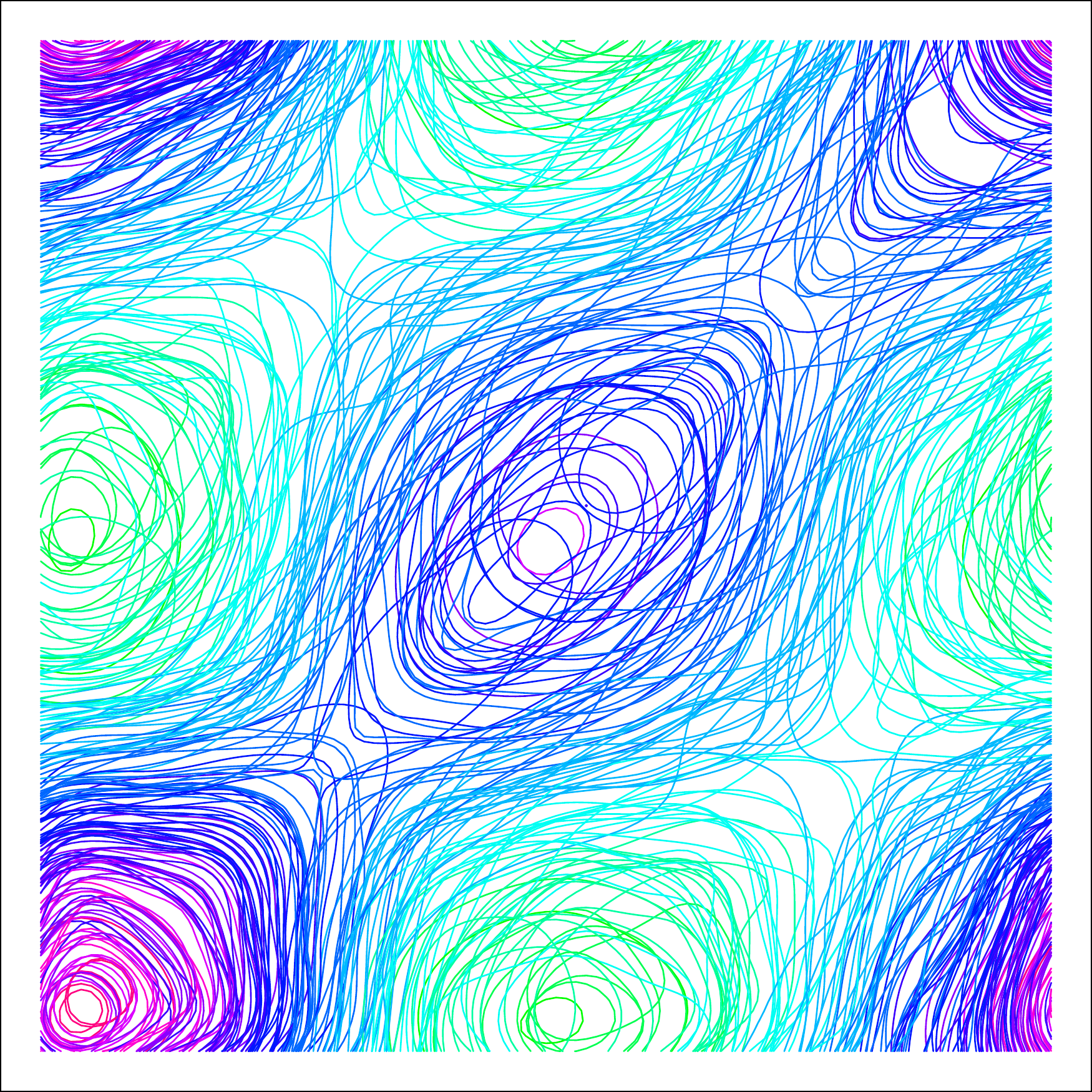}& \includegraphics[width=0.22\textwidth]{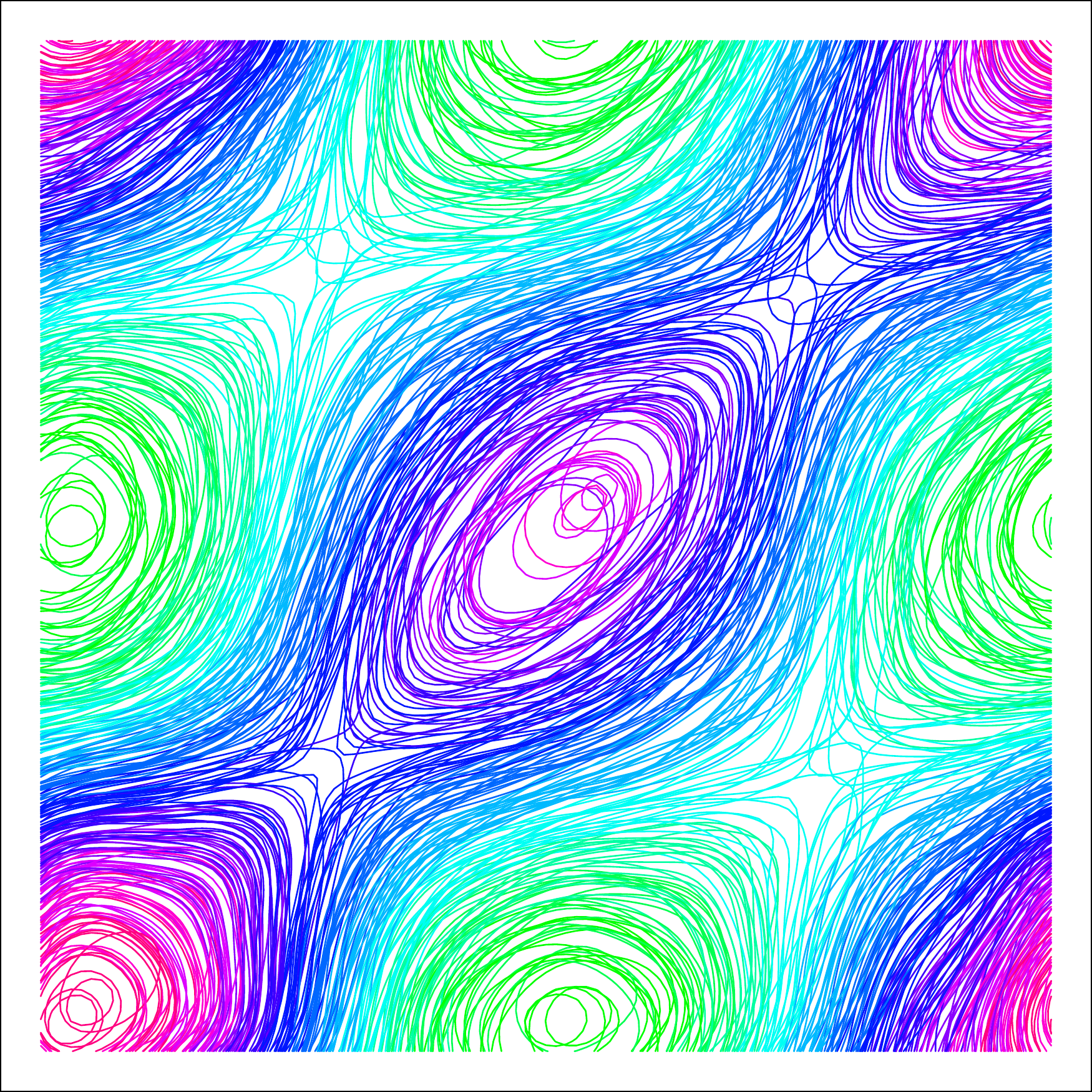}&\includegraphics[width=0.22\textwidth]{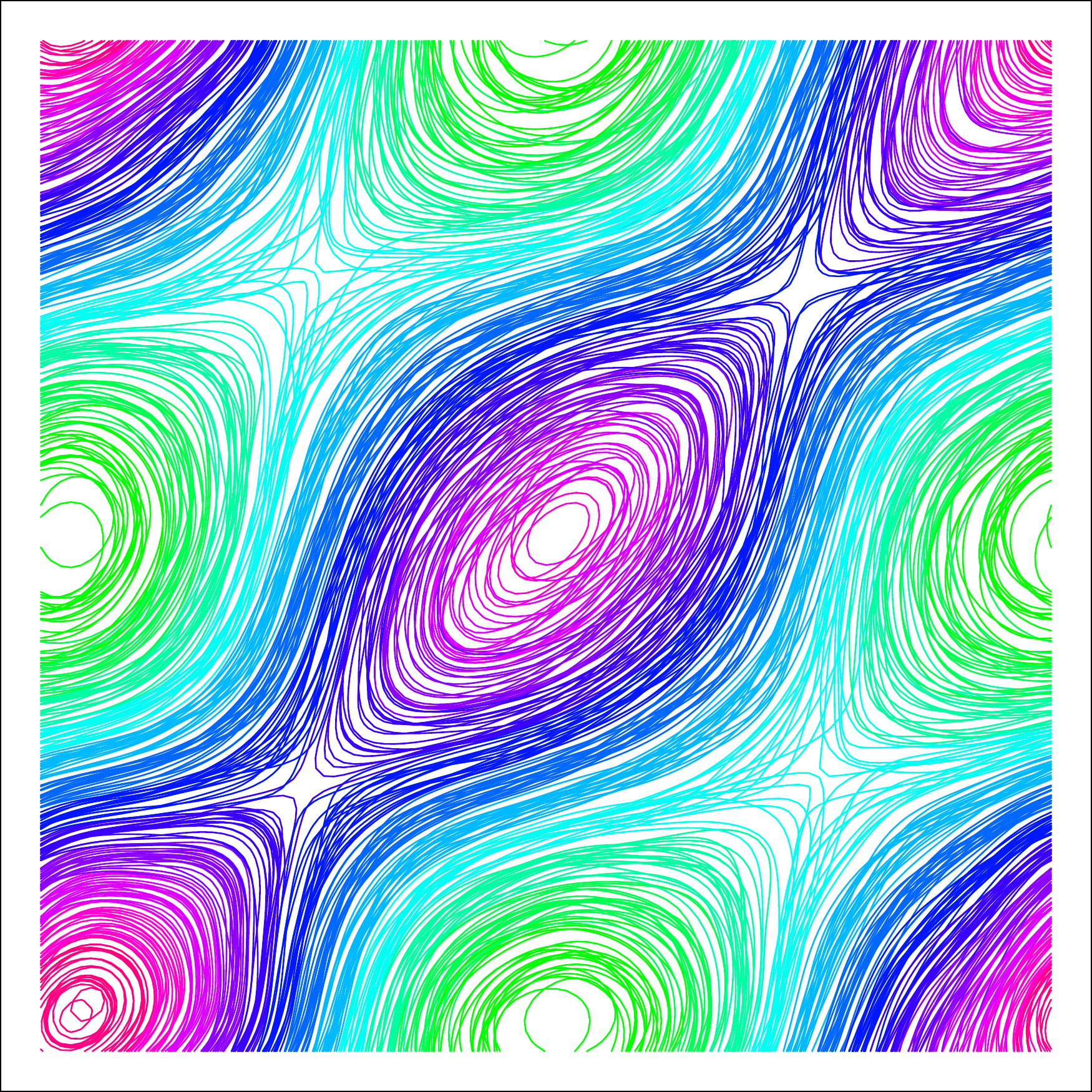} \\
\end{tabular}
\vspace{-5pt}
\caption{\footnotesize{Contour plots of the true and estimated spectral density of the FAR(2) at different time points for $\lambda = 1.5-cos(\pi u)$.}}\label{table:FAR2peak}
\end{figure}
\end{center}
\FloatBarrier

\section{Concluding remarks}

This paper forms a basis for the development of statistical techniques and methods for the analysis of nonstationary functional time series. We have provided a theoretical framework for meaningful statistical inference of functional time series with dynamics that change slowly over time. For this, the notion of local stationarity was introduced for time series on the function space. We focused on a class of functional locally stationary processes for which a time-varying functional Cram{\'e}r representation exists. The second-order characteristics of processes belonging to this class are completely captured by the time-varying spectral density operator. We moreover introduced time-varying functional ARMA processes and showed that these belong to the class of locally stationary functional processes. In the last section, we considered the nonparametric estimation of the time-varying spectral density operator. To derive the asymptotic distribution, a weaker tightness criterion is used than what is common in  the existing literature. The results derived in this paper give rise to consider Quasi-likelihood methods on the function space as well as the development of prediction and appropriate dimension reduction techniques for nonstationary time series on the function space. This is left for future work.\\

\medskip
\noindent 
{\bf Acknowledgements.}
This work has been supported in part by Maastricht University, the contract ``Projet d'Actions de Recherche Concert{\'e}es'' No. 12/17-045 of the ``Communaut{\'e} fran\c{c}aise de Belgique'' and by the Collaborative Research Center ``Statistical modeling of nonlinear dynamic processes'' (SFB 823, Project A1, C1, A7) of the German Research Foundation (DFG). The authors sincerely thank the Editor and three Referees for their constructive comments that helped to produce an improved revision of the original paper. 

\appendix
\setcounter{section}{0}
\setcounter{equation}{0}
\def\theequation{A\arabic{section}.\arabic{equation}}
\def\thesection{A\arabic{section}}
\section{Proofs}
 In this appendix, we prove the main theoretical results of the paper. 
\subsection{Proofs of section \ref{section2}} \label{proofssection2}

\begin{proof}[Proof of Proposition \ref{tvMAbochnerinf}]
For fixed $t\in\{1,\ldots,T\}$ and $T\in\nnum$, let $U_{s,\omega}=e^{\im\omega (t-s)}\,\AtT{s}$. We have
\[
\mathcal{T}( U_{s,\cdot})
= \int_{-\pi}^{\pi} e^{\im\omega (t-s)}\,\AtT{s}\,dZ_{\omega}
= \AtT{s}\,\varepsilon_{t-s},
\]
where $\mathcal{T}$ is the mapping defined by linear extension of
$\mathcal{T}(U\,\mathbf{1}_{[\alpha,\beta)}) = U(Z_{\beta}-Z_{\alpha})$ (see section \ref{proofStochInt}). By definition of the operator $ U_{s,\cdot}$, $\norm{U_{s,\cdot}}^2_{\mathcal{B}_{\infty}} \leq \snorm{\AtT{s}}^2_{\infty} \int_{-\pi}^{\pi} \snorm{\mathcal{F}^{\veps}_{\omega}}_1\,d\omega<\infty$ and thus $U_{s,\cdot} \in {\mathcal{B}_{\infty}}$. Similarly, $\sum_s \mathcal{T}( U_{s,\cdot}) \in {\mathcal{B}_{\infty}}$ from which it follows that
\[
\lim_{N \to \infty}\sum_{|s|\le N}e^{\im\omega (t-s)}\,\AtT{s}
=e^{\im\omega t} \sum_{s\in\znum}e^{-\im\omega s}\,\AtT{s}=e^{\im\omega t}\,\AAtT{\omega}  \in \mathcal{B}_{\infty}.
\]
The continuity of the mapping $\mathcal{T}$ then implies 
\[
\XT{t}
=\lsum_{s} \mathcal{T}( U_{s,\cdot})
=\mathcal{T}(\lsum_s  U_{s,\cdot})
=\int_{-\pi}^{\pi} e^{\im\omega t}\,\AAtT{\omega}\,dZ_{\omega}\qquad \text{a.e. in } \Hspace.\qedhere
\]
\end{proof}

\subsection{Proofs of section \ref{section4}}\label{proofsection3}
\begin{proof}[Proof of Lemma \ref{oplem}] 
We follow the lines of \citet[Theorem 5.2, Corollary 5.1]{Bosq}. To ease notation, we shall write $I$ and $O$ for the identity and zero operator on $H$, respectively while we denote the identity operator on $H^m$ by $I_{H^m}$. Consider the bounded linear operator $\tilde P_u(\lambda)$ on $H$
\begin{align*}
\tilde P_u(\lambda) = \lambda^m I-\lambda^{m-1} B_{u,1} - ...- \lambda B_{u,m-1} -B_{u,m}, \quad \lambda \in \mathbb{C}.
\end{align*}
It is straightforward to derive that, under the assumption $\sum_{j=1}^{m}\snorm{B_{u,j} }_{\infty} < 1$, non-invertibility of $\tilde P_u(\lambda)$ implies that $\lambda$ has modulus strictly less than $1$. Define the invertible matrices $U_u(\lambda)$ and $M_u(\lam)$ on the complex extension $H^m$ by
\begin{align*}
U_{u,ij}(\lambda) &= \begin{cases}
\lam^{j-i}\,I_H&\text{ if }j\geq i,\\
O_H&\text{ otherwise}
\end{cases}\\
\intertext{for $i,j=1,\ldots,m$ and}
M_u(\lambda) &= \begin{pmatrix}
  O_{H^{(m-1) \times 1}} & -I_{H^{m-1}}\\
 P_{u,0}(\lambda) & P_{u,1}(\lambda),\ldots,P_{u,m-1}(\lambda) 
\end{pmatrix},
\end{align*}
where $P_{u,0}(\lambda) = I$ and $P_{u,j}(\lambda) = \lambda_u P_{u,j-1}(\lambda) - B_{u,j}$ for $j = 1,\ldots,m$. Then
\[
M_u(\lambda)\,\Big(\lambda\,I_{H^m}-\boldsymbol{B}^{*}_{u}\Big)\,U_u(\lambda)
= \begin{pmatrix}
  I_{H^{(m-1)}} & O_{H^{(m-1)\times 1}}\\
  O_{H^{1 \times (m-1)}}& \tilde P_u(\lambda) 
\end{pmatrix},
\]
from which it follows that $(\lambda\,I_{H^m}-\boldsymbol{B}^{*}_{u})$ is not invertible when $\tilde P_u(\lambda)$ is not invertible. In other words, the spectrum $S_u$ of $\boldsymbol{B}^{*}_{u}$ over the complex extension of $H^m$, which is a closed set, satisfies
\begin{align*}
S_u = \{\lambda\in\cnum: \lambda\,I_{H^m}-\boldsymbol{B}^{*}_{u}\text{ not invertible} \}
& \subset \{\lambda\in\cnum: \tilde P_u(\lambda)\text{ not invertible} \} \\
& =  \{\lambda\in\cnum: |\lambda|<1 \}.
\end{align*}
Hence, the assumption that $\sum_{j=1}^{m}\snorm{B_{u,j} }_{\infty} < 1$ for all $u$, implies that the spectral radius of $\boldsymbol{B}^{*}_{u}$ satisfies \begin{equation}
\label{eq:sr}
r(\boldsymbol{B}^{*}_{u})
=\sup_{\lambda\in S_u}|\lambda|
= \lim_{k \to\infty}\bigsnorm{\boldsymbol{B}^{*k}_{u}}^{1/k}_{\infty}
< \frac{1}{1+\delta}
\end{equation}
for some $\delta >0$. The equality is a well-known result for the spectral radius of bounded linear operators\footnote{Gelfand's formula} and can for example be found in \citet{DunfordSchwartz}. From \eqref{eq:sr} it is now clear that there exists a $k_0 \in \znum, \alpha \in (0,1)$ and a constant $c_1$ such that for all $k \ge k_0$
\begin{align*} \tageq \label{eq:weaknorm}
\snorm{\boldsymbol{B}^{*k}_{u} }_{\infty} < c_1 \alpha^{k}.
\end{align*} 
Finally, it has been shown in \citet[p.74]{Bosq} that this is equivalent to the condition $\snorm{ \boldsymbol{B}^{*k_0}_{u}}_{\infty} <1$ for some integer $k_0 \ge 1$.
\end{proof} 

\begin{proof}[Proof of Theorem \ref{FAR}]  
The moving average representation \eqref{eq:MArep} and the difference equation \eqref{eq:DE} together imply that the process can be represented as
 \begin{align*}
\XT{t} 
 = \sum_{l=0}^{\infty} \AtT{l}  C^{-1}_{\frac{t-l}{T}} \sum_{j=0}^{m} B_{\frac{t-l}{T},j}(X_{t-l-j,T}). 
\end{align*}
Using the linearity of the operators and applying a change of variables $l' = l+j$, this can be written as
\begin{align*}
\displaybreak[0]
\XT{t}& = \sum_{l'=0}^{\infty}  \sum_{j=0}^{m}\AtT{l'-j} C^{-1}_{\frac{t-l'+j}{T}} B_{\frac{t-l'+j}{T},j}(\XT{t-l'}),
\displaybreak[0]
\end{align*}
where $\AtT{l'-j} = O_H$  for $l'<j$. For a purely nondeterministic solution we require
\begin{align} \label{eq:valuel}
\sum_{j=0}^{m}\AtT{l'-j} C^{-1}_{\frac{t-l'+j}{T}} B_{\frac{t-l'+j}{T},j}
= \begin{cases}
    I_H & \text{if } l' = 0, \\
    O_H & \text{if } l' \ne 0.
  \end{cases}
\end{align}
Because $\varepsilon_{t}$ is white noise in $\Hspace$, it has spectral representation 
\begin{equation} \label{eq:speceps}
\varepsilon_{t} = ({2\pi})^{-1/2}\int_{-\pi}^{\pi} e^{\im\omega t} dZ_{\omega}, \qquad t\in\znum. 
\end{equation}
Since a solution of the form $\eqref{eq:MArep}$ exists, we also have
\begin{equation*}
\XT{t} = \int_{-\pi}^{\pi} e^{\im\omega t} \AAtT{\omega} dZ_{\omega},
\end{equation*}
where $\AAtT{\omega} =\frac{1}{\sqrt{2 \pi}} \sum_{l=0}^{\infty} \AtT{l}\,
e^{-\im\omega l}$. Substituting the spectral representations of $\XT{t}$ and $\veps_t$ into \eqref{eq:DE}, we get together with the linearity of the operators $ B_{u,j}$ and $\AAtT{\omega}$
\begin{align*}
\int_{-\pi}^{\pi}  \sum_{j=0}^{m} e^{\im\omega (t-j)}\,B_{\frac{t}{T},j}\, \AAttT{t-j}{\omega}\,dZ_{\omega}
& =  ({2\pi})^{-1/2}\int_{-\pi}^{\pi} e^{\im\omega t} \,  C_{\frac{t}{T}}\, dZ_{\omega},
\end{align*}
Given the operator $\AAuu{\frac{t}{T}}{\omega} $ satisfies equation \eqref{eq:smoothf}, the previous implies we can write 
\begin{align*}
\frac{1}{\sqrt{2\pi}}C_{\frac{t}{T}}
&=\sum_{j=0}^{m} e^{-\mathrm{i}\omega j} B_{\frac{t}{T},j}\mathcal{A}_{\frac{t}{T},\omega}\\
&=\sum_{j=0}^{m} e^{-\mathrm{i}\omega j}  B_{\frac{t}{T},j}\mathcal{A}_{\frac{t-j}{T},\omega}
+  \sum_{j=0}^{m} e^{-\mathrm{i}\omega j}  B_{\frac{t}{T},j}\big(\mathcal{A}_{\frac{t}{T},\omega}-\mathcal{A}_{\frac{t-j}{T},\omega} \big). 
\end{align*}
From the last equation, it follows that
\begin{align*} 
\sum_{j=0}^{m} e^{\im\omega (t-j)}  B_{\frac{t}{T},j}\big(\AAttT{t-j}{\omega}-\AAuu{\frac{t-j}{T}}{\omega} \big)
& = \sum_{j=0}^{m} e^{\im\omega (t-j)}  B_{\frac{t}{T},j}
\big(\AAuu{\frac{t}{T}}{\omega}-\AAuu{\frac{t-j}{T}}{\omega}\big) \\
& =C_{\frac{t}{T}}\,\Omega^{(T)}_{t,\omega}, \tageq \label{eq:omegastuff}
\end{align*}
where $\Omega^{(T)}_{t,\omega}= O_H$, $t \le 0$. We will show that this operator is of order $O(\frac{1}{T})$ in $S_p(H_\cnum)$. 
By Proposition \ref{Neumann}, the smooth transfer operator satisfies $\AAu{\omega}\in S_p(H_\cnum)$. Under the conditions of Theorem \ref{FAR}, we have that for any element $\psi \in H_\cnum$ and fixed $\omega \in \Pi$, the mapping $u \mapsto \AAu{\omega}(\psi)(\tau) $ is continuous and, from the properties of the $B_{u,j}$, is differentiable and has bounded derivatives with respect to $u$. By the Mean Value Theorem, we obtain 
\begin{align*}
\sup_{t, \omega}
\bigsnorm{\AAuu{\frac{t}{T}}{\omega}-\AAuu{\frac{t-j}{T}}{\omega}}_{p}\leq
\sup_{t,\omega} \sup_{ \frac{t-j}{T}<u_0<\frac{t}{T}}\Big|\SSS{\frac{j}{T}}\Big|\,
\Bigsnorm{\SSS{\frac{\partial}{\partial u }}\Au{\omega}\vert_{u=u_o}}_{p} =O\big(\SSS{\frac{1}{T}}\big),
\end{align*}
for all $\omega \in \Pi$, uniformly in $u$. 
It then easily follows from equation \eqref{eq:omegastuff} and Proposition \ref{holderoperator} that $\snorm{C_{\frac{t}{T}} \Omega^{(T)}_{t,\omega}}_p= O(\frac{1}{T})$ uniformly in $t, \omega$. From \eqref{eq:valuel}, we additionally have
\begin{align*} 
\sum_{l=0}^{t} \AtT{l}\,\Omega^{(T)}_{t-l,\omega}& =
 \sum_{l=0}^{t}\sum_{j=0}^{m} \AtT{l-j}\,C^{-1}_{\frac{t-l+j}{T}}\,  B_{\frac{t-l+j}{T},j}\, e^{\im\omega (t-l)}
 \big[\AAttT{t-l}{\omega}-\AAuu{\frac{t-l}{T}}{\omega} \big]\\ 
& = e^{\im\omega t}
\big[\AAtT{\omega}-\AAuu{\frac{t}{T}}{\omega} \big].
\end{align*}
Since the moving average operators are either in $S_{2}(H)$ or in $S_{\infty}(H)$, the above together with another application of H{\"o}lder's inequality for operators yields
\begin{align*} 
\sup_{t, \omega}\snorm{ \AAtT{\omega}-\AAuu{\frac{t}{T}}{\omega} }_{p} &
\leq \sup_{t, \omega}\Big(\bigsnorm{\AtT{0} }_{\infty} \, \bigsnorm{\Omega^{(T)}_{t,\omega}}_{p}
+ \lsum_{l=1}^{t} \bigsnorm{\AtT{l} }_p\,  \bigsnorm{\Omega^{(T)}_{t-l,\omega} }_{p}\Big)\leq \frac{K}{T}, 
\end{align*}
for some constant $K$ independent of $T$. We remark that the state space representation of the previous section allow similarly to derive that (A1) and (A2) still hold for $p=\infty$ under the weaker assumption of Lemma \ref{oplem}.
\end{proof} 

\subsection{Proofs of section \ref{section3}}\label{proofsection4}
\begin{proof}[Proof of Proposition \ref{propl2rtv}]
For fixed $t$ and $T$, we have by Minkowski's inequality
\begin{align*} \label{sumdecom}
&   \sum_s \snorm{ \cum  (\XT{\lfloor{uT-s/2\rfloor}}, \XT{\lfloor{uT+s/2\rfloor}}) }_2 \\& 
= \sum_{s}  \snorm{ \frac{1}{2 \pi} \int_{{\Pi}}\big( \AAttT{\lfloor uT-s/2 \rfloor}{\lam_1} \otimes  \AAttT{\lfloor uT+s/2 \rfloor}{-\lam_1} \big) \mathcal{F}^{\varepsilon}_{\lambda_1} e^{\mathrm{i} \lambda_1 s}  d{\lambda_1} }_2\\&
=\sum_{\substack{ s: \{ ( 1 \le \lfloor{uT-s/2\rfloor} \le T ) \cup \\ (1 \le \lfloor{uT+s/2\rfloor} \le T) \}}} \snorm{ \CCT_{u,s} }_2 +\sum_{\substack{ s: \{ ( 1 \le \lfloor{uT-s/2\rfloor} \le T ) \cup \\ (1 \le \lfloor{uT+s/2\rfloor} \le T) \}^\complement}}  \snorm{\CCT_{u,s} }_2   
\intertext{where $\{\cdot\}^\complement$ denotes the complement event. Now since $\AAttT{t}{\omega} = \AAuu{0}{\omega}$ for $t < 1$ and $\AAttT{t}{\omega} = \AAuu{0}{\omega}$ for $ t>T$, we can write}
& = \sum_{s :B}  \snorm{ \frac{1}{2 \pi} \int_{{\Pi}}\big( \AAttT{\lfloor uT-s/2 \rfloor}{\lam_1} \otimes  \AAttT{\lfloor uT+s/2 \rfloor}{-\lam_1} \big) \mathcal{F}^{\varepsilon}_{\lambda_1} e^{\mathrm{i} \lambda_1 s}  d{\lambda_1} }_2 \\& 
+\sum_{s :B^{\complement}}  \snorm{ \frac{1}{2 \pi} \int_{{\Pi}}\big( \AAuu{0}{\lam_1} \otimes  \AAuu{1}{-\lam_1} \big) \mathcal{F}^{\varepsilon}_{\lambda_1} e^{\mathrm{i} \lambda_1 s}  d{\lambda_1} }_2,\tageq
\end{align*}
where $B =  \{ ( 1 \le \lfloor{uT-s/2\rfloor} \le T ) \cup (1 \le \lfloor{uT+s/2\rfloor} \le T) \}$. Because the first sum is finite, an application of proposition \ref{holderoperator} implies it can be bounded by
\[
K \sup_{t,T,\omega }\snorm{ \AAtT{\omega} }^2_{\infty} \snorm{\F^{\varepsilon}}_2 < \infty,
\]
for some constant $K$. For the second term, we note that 
\[
 \frac{1}{2 \pi} \int_{{\Pi}}\big( \AAuu{0}{\lam_1} \otimes  \AAuu{1}{-\lam_1} \big) \mathcal{F}^{\varepsilon}_{\lambda_1} e^{\mathrm{i} \lambda_1 s}  d{\lambda_1} = \cum(\Xuu{0}{t+s}, \Xuu{1}{t}).
\]
It thus corresponds to the cross-covariance operator of the two stationary processes $\Xuu{0}{t}$ and $\Xuu{1}{t}$ at lag $s$. By Propsosition \ref{MAcumrep},  we can alternatively express this as
\[ \cum(\Xuu{0}{s}, \Xuu{1}{0})  = \sum_{l,k} (A_{0,l} \otimes A_{1,k}) \cum(\varepsilon_{t+s-l}, \varepsilon_{t-k}).  
\]
Using then that $\varepsilon_t $ is functional white noise, we find for the second term in \eqref{sumdecom}
\begin{align*}
\sum_{s :B^{C}} \snorm{\CCT_{u,s} }_2   & \le  \snorm{ \sum_{l,k \in \znum} (A_{0,l} \otimes A_{1,k}) \cum(\varepsilon_{0}, \varepsilon_{0}) }_2 \\& \le  \sum_{l \in \znum} \snorm{  A_{0,l} }_{\infty}  \sum_{k \in \znum} \snorm{A_{1,k} }_{\infty}\snorm{\cum(\varepsilon_{0}, \varepsilon_{0}) }_2 < \infty.
\end{align*}
The result now follows.
\end{proof}

\subsection{Proofs of section \ref{expandcov}} \label{proofsection5}
\begin{proof}[Proof of Proposition \ref{cumX}]

We have by Theorem \ref{cumWN} and by Proposition \ref{momlin}, \begin{align*}
& \cum\big(\int_{\Pi} e^{\mathrm{i} \lambda_1 r_1} \AttT{t_{r_1}}{\lambda_1}dZ_{\lambda_1},\ldots, \int_{\Pi} e^{\mathrm{i} \lambda_k r_k}\AttT{t_{r_k}}{\lambda_k}dZ_{\lambda_k}\big)  \\&
=
\int_{\Pi}  \cdots \int_{\Pi}  \cum\big(e^{\mathrm{i} \lambda_1 r_1} \AttT{t_{r_1}}{\lambda_1}dZ_{\lambda_1},\ldots, e^{\mathrm{i} \lambda_k r_k} \AttT{t_{r_k}}{\lambda_k}dZ_{\lambda_k}\big)
\\&
=
\int_{\Pi}  \cdots \int_{\Pi} \big(e^{\mathrm{i} \lambda_1 r_1} \AttT{t_{r_1}}{\lambda_1} \otimes \cdots \otimes  e^{\mathrm{i} \lambda_k r_k} \AttT{t_{r_k}}{\lambda_k}\big) \cum\big(dZ_{\lambda_1},\ldots,dZ_{\lambda_k}\big)
\\& 
=\int_{\Pi^{k}}e^{\im(\lam_1 r_1+\ldots+\lam_k r_k)}\,
\big(\AAttT{t_{r_1}}{\lambda_1}\otimes\cdots\otimes \AAttT{t_{r_{k}}}{\lambda_{k}}\big) \eta(\lam_1+\ldots+\lam_k)\,   \F^{\varepsilon}_{\lambda_1,\ldots,\lambda_{k-1}}d\lambda_1\cdots d\lambda_{k},
\end{align*}
where the equality holds in the tensor product space $H_\cnum\otimes\cdots\otimes H_\cnum$. Note that the last line corresponds to the inversion formula of the cumulant tensor of order $k$. For fixed $t\in\{1,\ldots,T\}$ and $T\in\nnum$, the $k$-th order cumulant spectral tensor of the linear functional process $\{\XT{t}\}$ can thus be given by
\begin{align} \label{eq:SKCk}
\F^{(t,T)}_{\lam_1,\ldots,\lam_{k-1}}
=\Big( \AAttT{t_{r_{1}}}{\lam_1} \otimes\cdots\otimes  \AAttT{t_{r_{k-1}}}{\alpha_{k-1}}\otimes \AAttT{t_{r_{k}}}{-\lam_+}\Big) \F^{\varepsilon}_{\lam_1,\ldots,\lam_{k-1}},
\end{align}
and is well-defined in the tensor product space $\bigotimes_{i=1}^{k} H_\cnum$. In particular, Proposition \ref{holderoperator} implies the corresponding operator is Hilbert-Schmidt for $k \ge 2$ 
\begin{align*}
& \snorm{\F_{\lambda_1,\ldots,\lambda_{2k-1}}}_2  \le \snorm{\AAttT{t_{r_{1}}}{\lam_1} \otimes\cdots\otimes  \AAttT{t_{r_{k-1}}}{\alpha_{k-1}}\otimes \AAttT{t_{r_{k}}}{-\lam_+}}_{\infty} \snorm{\mathcal{F}^\varepsilon_{\lambda_1,\ldots,\lambda_{2k-1}}}_2\\& =\big(\underset{t_{r_j, \lambda_j}}{\text{sup}}\| \AAttT{t_{r_{j}}}{\lam_j} \|_{\infty}\big)^{2k} \snorm{\mathcal{F}^\varepsilon_{\lambda_1,\ldots,\lambda_{2k-1}}}_2 < \infty.
\end{align*}
We therefore have that the kernel function $f^{(t,T)}_{\lambda_1,\ldots,\lambda_{k-1}}(\tau_1,\ldots,\tau_k)$ is a properly defined element in $L^2_\cnum([0,1]^2)$. In case $k=2$, we moreover have that $\F_{\lambda_1} \in S_1(H_\cnum)$. This follows by the fact that the $\varepsilon_t$ are white noise and thus $\snorm{\mathcal{F}^{\varepsilon}_{\lambda_1}}_1 \le \sum_t \snorm{\cum(\varepsilon_t, \varepsilon_0)}_1= \snorm{ \mathcal{C}^{\varepsilon}_0}_1  = \E\|\varepsilon_{0}\|^2_2 <\infty$. 
\end{proof}

\begin{proof}[Proof of Theorem \ref{exp}]
Under assumption (A1)--(A2) we have for all $t=1,..,T$ and $T \in \mathbb{N}$ that $\XT{t}$ are locally stationary random elements in $H$. Therefore, by Proposition \ref{cumX} and \eqref{eq:cumDNop},
\begin{align*} 
\mean(\IT_{u_j,\omega})&= \SSS{\frac{1}{2\pi\,\HNN{2}(0)}}\,\cum\big(\DT_{u_j,\omega},
\DT_{u_j,-\omega}\big)\\
& = \SSS{\frac{1}{2\pi\HNN{2}(0)}}
\int_{\Pi}\Big(
\HN\big(\AAttT{t_{j,\bullet}}{\lam},\omega-\lambda\big) \otimes\HN\big(\AAttT{t_{j,\bullet}}{-\lam},\lam-\omega\big)\Big)\, \F^\veps_\lam\,d\lam.\qquad
\end{align*}
In order to replace the transfer function kernels with their continuous approximations, we write
\begin{equation} 
\begin{split}\AAttT{t_{j,r}}{\lam}&\otimes\AAttT{t_{j,s}}{-\lam}
-\AAuu{u_{j,r}}{\omega}\otimes\AAuu{u_{j,s}}{-\omega}\\ &=\big(\AAttT{t_{j,r}}{\lam}-\AAuu{u_{j,r}}{\omega}\big)
\otimes\AAttT{t_{j,s}}{-\lam}
+\AAuu{u_{j,r}}{\omega}\otimes
\big(\AAttT{t_{j,s}}{-\lam}-\AAuu{u_{j,s}}{-\omega}\big).
\end{split}
\label{eq:errorrep}
\end{equation}
We focus on finding a bound on the first term as the second term can be bounded similarly. Since $\HN(\cdot,\cdot)$ is linear in its first argument, we have by the triangle inequality
\begin{align*}
&\bigsnorm{\HN\big(\AAttT{t_{j,r}}{\lam}-\AAuu{u_{j,r}}{\omega},
\omega-\lambda\big)}_{\infty}\\
&\qquad\leq\bigsnorm{\HN\big(\AAttT{t_{j,\bullet}}{\lam}
-\AAuu{u_{j,\bullet}}{\lam},
\omega-\lambda\big)}_{\infty}
+\bigsnorm{\HN\big(\AAuu{u_{j,\bullet}}{\lam}
-\AAuu{u_{j,\bullet}}{\omega},
\omega-\lambda\big)}_{\infty}.
\end{align*}
For the first term of this expression, assumption (A2) and Lemma \ref{Dh1993A5} imply 
\begin{align*}
& \Bigsnorm{\lsum_{r=0}^{N-1} \hN{r}
\big(\AAttT{t_{j,r}}{\lam}
-\AAuu{u_{j,r}}{\lam}\big)\,
e^{-\im r(\omega-\lambda)}}_{\infty}
 \leq C\,\SSS{\frac{N}{T}} \tageq
\end{align*}
for some generic constant $C$ independent of $T$. Next, we consider the second term. Similarly as in the proof of Lemma \ref{Dh1993A5}, we have
\begin{align*}
&\HN\big(\AAuu{u_{j,\bullet}}{\lam}
-\AAuu{u_{j,\bullet}}{\omega},\omega-\lambda\big)\\
&\quad=\HN(\omega-\lambda)\,\big(\AAuu{u_j}{\lam}
-\AAuu{u_j}{\omega}\big)
+\HN(\omega-\lam)\,\big(\AAuu{u_{j,N-1}}{\lam}
-\AAuu{u_{j,N-1}}{\omega}\big)\\
&\qquad-\lsum_{r=0}^{N-1}\Big[
\Big(\AAuu{u_{j,r}}{\lam}
-\AAuu{u_{j,r-1}}{\lam}\Big)
-\Big(\AAuu{u_{j,r}}{\omega}
-\AAuu{u_{j,r-1}}{\omega}\Big)\Big]\,H_s(\omega-\lam).
\end{align*}
Since the transfer function operator is twice continuously differentiable in $u$ and $\omega$, we find by two applications of the mean value theorem
\[
\Bigsnorm{\Big(\AAuu{u_{j,r}}{\lam}
-\AAuu{u_{j,r-1}}{\lam}\Big)
-\Big(\AAuu{u_{j,r}}{\omega}
-\AAuu{u_{j,r-1}}{\omega}\Big)}_\infty
\leq\sup_{u\in[0,1],\omega\in\Pi}
\Bigsnorm{\SSS{\frac{\partial^2\AAu{\omega}}{\partial u\partial\omega}}}\,
\SSS{\frac{|\lam-\omega|}{T}}.
\]
Hence we obtain the upper bound
\[
\bigsnorm{\HN\big(\AAuu{u_{j,\bullet}}{\lam}
-\AAuu{u_{j,\bullet}}{\omega},\omega-\lambda\big)}_\infty
\leq C\,L_N(\omega-\lam)\,|\omega-\lam|
+C\,\SSS{\frac{N}{T}}\,L_N(\omega-\lam)\,|\omega-\lam|.
\]
Moreover, Lemma \ref{Dh1993A5} implies
\[
\bigsnorm{\HN\big(\AAttT{t_{j,\bullet}}{-\lam},\omega-\lam\big)}_\infty
\leq C\,L_N(\omega-\lam).
\]
With these bounds and Proposition \ref{holderoperator} and  Lemma \ref{Dh1993A4}, we now obtain
\begin{align*}
\int_{\Pi}
&\Bigsnorm{\Big(\HN\big(\AAuu{u_{j,\bullet}}{\lam}
-\AAuu{u_{j,\bullet}}{\omega},\omega-\lambda\big)\otimes
\HN\big(\AAttT{t_{j,\bullet}}{-\lam},\omega-\lam\big)\Big)\,\F^\veps_\lam}_2
\,d\lam\phantom{XXXX}\\
&\leq\int_{\Pi} \bigsnorm{\HN\big(\AAuu{u_{j,\bullet}}{\lam}
-\AAuu{u_{j,\bullet}}{\omega},\omega-\lambda\big)}_\infty\,
\bigsnorm{\HN\big(\AAttT{t_{j,\bullet}}{-\lam},\omega-\lam\big)}_\infty\,
\bigsnorm{\F^\veps_\lam}_2\,d\lam\\
&\leq
C\int_{\Pi}L_N(\omega-\lam)^2\,d\lam
\leq C\,\log(N).
\end{align*}
The second term of \eqref{eq:errorrep} is similar and thus the error from replacing $\AAttT{t_{j,r}}{\lam}$ and $\AAttT{t_{j,s}}{-\lam}$ by $\AAuu{u_{j,r}}{\omega}$ and $\AAuu{u_{j,s}}{-\omega}$, respectively, is of order $O\big(\frac{\log(N)}{N})$ in $L^2$.

The expectation of the periodogram tensor can therefore be written as
\begin{align*}
\mean(\IT_{u_j,\omega})
&=\SSS{\frac{1}{2\pi \HNN{2}(0)}}\int_{\Pi}
\Big(\HN\big(\AAttT{t_{j,\bullet}}{\lam},\omega-\lam\big)\otimes
\HN\big(\AAttT{t_{j,\bullet}}{-\lam},\lam-\omega\big)\Big)\,\F^\veps_\lam\,
d\lam\\
&=\SSS{\frac{1}{2\pi \HNN{2}(0)}}\int_{\Pi}
\Big(\HN\big(\AAuu{u_{j,\bullet}}{\omega},\omega-\lam\big)\otimes
\HN\big(\AAuu{u_{j,\bullet}}{-\omega},\lam-\omega\big)\Big)\,\F^\veps_\lam\,
d\lam+R_T\\
&=\SSS{\frac{1}{\HNN{2}(0)}}
\HNN{2}\big(\AAuu{u_{j,\bullet}}{\omega}\otimes
\AAuu{u_{j,\bullet}}{-\omega},0\big)\,\F^\veps_\omega+R_T\\
&=\SSS{\frac{1}{\HNN{2}(0)}}
\HNN{2}\big(\F_{u_{j,\bullet},\omega},0\big)+R_T
\end{align*}
where the remainder term $R_T$ is of order $O\big({\frac{\log(N)}{N}}\big)$.
Correspondingly, the local periodogram kernel is given by
\begin{align*}
\mean(\IT_{u_j,\omega}(\tau,\sigma)
&=\SSS{\frac{1}{\HNN{2}(0)}}
\lsum_{r=1}^{N} \hN{r}^2\,f_{u_{j,r},\omega}(\tau,\sigma)
+O\Big(\SSS{\frac{\log(N)}{N}}\Big).
\end{align*}
Since by the conditions of the theorem, the operator-valued function $\AAu{\omega}$ is twice continuously differentiable with respect to $u$, Theorem \ref{ProdBan} implies that the spectral density operator $\F_{u, \omega}$ is also twice continuously differentiable in $u\in(0,1)$. Hence, by a Taylor approximation of $\F_{u_{j,r},\omega}$ about $u_j$, we find for the mean of the periodogram tensor
\[
\mean(\IT_{u_j,\omega})
=\F_{u_j,\omega}+\frac{1}{2}\btT^2\,\kappat\,
\SSS{\frac{\partial^2 \F_{u,\omega}}{\partial u^2}}\Big|_{u=u_j}
+O\Big(\frac{\log(N)}{N}\Big),
\]
where we have used the definition in \eqref{Kerneltimedef} of the smoothing kernel $\Kt$ in time direction. As by the assumption on the taper function this kernel is symmetric about zero, the first order term in the Taylor approximation is zero.

This proves the first part of from Theorem \ref{exp}. For the covariance, we note that the product theorem for cumulants  (see section \ref{cumprops}) and the fact that the means are zero imply 
\begin{equation}
\label{eq:covariancep}
\begin{split}
\Cov & \big(\IT_{u_j,\omega_1},\IT_{u_j,\omega_2}\big)\\
&=\SSS{\frac{1}{4\pi^2\HNN{2}(0)^2}}\,\Big[
\cum\big(\DT_{u_j,\omega_1},\DT_{u_j,-\omega_1},
\DT_{u_j,-\omega_2},\DT_{u_j,\omega_2}\big)\\
&\qquad+S_{1423}\Big(
\cum\big(\DT_{u_j,\omega_1},\DT_{u_j,\omega_2}\big)\otimes
\cum\big(\DT_{u_j,-\omega_1},\DT_{u_j,-\omega_2}\big)\Big)\\
&\qquad+S_{1324}\Big(
\cum\big(\DT_{u_j,\omega_1},\DT_{u_j,-\omega_2}\big)\otimes
\cum\big(\DT_{u_j,-\omega_1},\DT_{u_j,\omega_2}\big)\Big)\Big],
\end{split}
\end{equation}
where $S_{ijkl}$ denotes the permutation operator on $\otimes_{i=1}^4 L^2_\cnum([0,1])$ that permutes the components of a tensor according to the permutation $(1,2,3,4)\mapsto(i,j,k,l)$, that is, $S_{ijkl}(x_1\otimes\cdots\otimes x_4)=x_i\otimes\cdots\otimes x_l$.

We first show that the first term of this expression is of lower order than the other two. By \eqref{eq:cumDNop}, the cumulant is equal to
\[
\begin{split}
\int_{{\Pi}^4}&\HN\big(\AAttT{t_{j,{\cdot}}}{\lam_1},\omega_1-\lam_1\big)
\otimes\HN\big(\AAttT{t_{j,{\cdot}}}{\lam_2},-\omega_1-\lam_2\big)
\otimes\HN\big(\AAttT{t_{j,{\cdot}}}{\lam_3},-\omega_2-\lam_3\big)\\
&\otimes\HN\big(\AAttT{t_{j,{\cdot}}}{\lam_4},\omega_2-\lam_4\big)\,
\eta(\lam_1+\ldots+\lam_4)\,\F^\veps_{\lam_1,\lam_2,\lam_3}\,d\lam_1\cdots d\lam_4
\end{split}
\]
and hence, by Lemma \ref{Dh1993A5}, is bounded in $L^2$-norm by
\begin{align*} 
C\int_{{\Pi}^3}&
L_N(\omega_1-\lam_1)\,L_N(-\omega_1-\lam_2)\,L_N(-\omega_2-\lam_3)\,
L_N(\lam_1+\lam_2+\lam_3+\omega_2)\,d\lam1\,d\lam_2\,d\lam_3\\
&\leq
C\,\log(N)^2\int_{{\Pi}^3}
L_N(\omega_2+\lam_3)^2\,d\lam_3\leq C\,N\,\log(N)^2.
\end{align*}

Next we consider the second term of \eqref{eq:covariancep}. A similar derivation as for the expectation of the periodogram tensor shows that the term equals
\begin{align*}
\int_{\Pi^2}&
\HN\big(\AAttT{t_{j,\bullet}}{\lam_1},\omega_1-\lam_1\big)
\otimes\HN\big(\AAttT{t_{j,\bullet}}{-\lam_1},\omega_2+\lam_1\big)
\otimes\HN\big(\AAttT{t_{j,\bullet}}{\lam_2},-\omega_1-\lam_2\big)\\
&\qquad\qquad\otimes
\HN\big(\AAttT{t_{j,\bullet}}{-\lam_2},\lam_2-\omega_2\big)\,
\F^\veps_{\lam_1}\otimes\F^\veps_{\lam_2}\,d\lam_1\,d\lam_2\\
&=\int_{\Pi^2}
\HN\big(\AAuu{u_{j,\bullet}}{\lam_1},\omega_1-\lam_1\big)
\otimes\HN\big(\AAuu{u_{j,\bullet}}{-\lam_1},\omega_2+\lam_1\big)
\otimes\HN\big(\AAuu{u_{j,\bullet}}{\lam_2},-\omega_1-\lam_2\big)\\
&\qquad\qquad\otimes\HN\big(\AAuu{u_{j,\bullet}}{-\lam_2},\lam_2-\omega_2\big)\,
\F^\veps_{\lam_1}\otimes\F^\veps_{\lam_2}\,d\lam_1\,d\lam_2+R_T\\
&=\HNN{2}\big(\AAuu{u_{j,\bullet}}{\omega_1}\otimes\AAuu{u_{j,\bullet}}{-\omega_1},\omega_1-\omega_2\big)
\otimes
\HNN{2}\big(\AAuu{u_{j,\bullet}}{\omega_1}\otimes\AAuu{u_{j,\bullet}}{-\omega_1},\omega_2-\omega_1\big)\\
&\qquad\qquad\times\F^\veps_{\omega_1}\otimes\F^\veps_{-\omega_1}\\
&=\HNN{2}\big(\F_{u_{j,\bullet},\omega_1},\omega_1-\omega_2\big)
\otimes
\HNN{2}\big(\F_{u_{j,\bullet},-\omega_1},\omega_2-\omega_1\big)
\end{align*}
Proceeding in an analogous matter for the third term of (\ref{eq:covariancep}), we obtain the stated result.

\end{proof}

\begin{proof}[Proof of Theorem \ref{expsmooth}]
By Theorem \ref{exp}, the expectation of the periodogram tensor can be written as
\begin{align*}
\mean(\IT_{u_j,\omega})
&=\HNN{2}\big(\F_{u_{j,\bullet},\omega},0\big)+R_T = \SSS{\frac{1}{\HNN{2}(0)}}
\lsum_{r=1}^{N} \hN{r}^2\,\F_{u_{j,r},\omega}+O\Big(\SSS{\frac{\log(N)}{N}}\Big).
\end{align*}
where the remainder term $R_T$ is of order $O\big({\frac{\log(N)}{N}}\big)$.
Because the operator-valued function $\AAu{\omega}$ is twice differentiable with respect to both $u$ and $\omega$, it follows from Theorem \ref{ProdBan} that the tensor $\F_{u,\omega}$ is twice continuously differentiable in both $u,$ and $\omega$. We can therefore apply a Taylor expansion of $\F_{u_{j,r},\omega}$ about to the point $x=(u_j,\omega_o)$ to obtain
\begin{align*} \label{eq:Taylor}
\F_{u_{j,r},\omega} & = \F_{u_j,\omega_o} +\Big(\frac{r-N/2}{T}\Big)\frac{\partial}{\partial u} \F_{u,\omega}\Big|_{(u,\omega)=x}+(\omega-\lam) \frac{\partial}{\partial \omega} \F_{u,\omega}\Big|_{(u,\omega)=x} \\&
=\frac{1}{2}\Big(\frac{r-N/2}{T}\Big)^2 \frac{\partial^2}{\partial u^2} \F_{u,\omega}\Big|_{(u,\omega)=x}+\frac{1}{2}(\omega-\lam)^2 \frac{\partial^2}{\partial \omega^2} \F_{u,\omega}\Big|_{(u,\omega)=x} 
\\ & +\Big(\frac{r-N/2}{T}\Big)(\omega-\lam) \Big( \frac{\partial^2}{\partial u \partial \omega} \F_{u,\omega}\Big|_{(u,\omega)=x}+ \frac{\partial^2}{\partial \omega\partial u } \F_{u,\omega}\Big|_{(u,\omega)=x}
\Big)+ R_{T,p}, \tageq
\end{align*}
where the remainder can generally be bounded by 
\begin{align*} \label{Rwind} R_{T,p} & =\sum_{i_1,i_2 \in \nnum:i_1+i_2 > p} \frac{\btT^{i_1} |\omega-\alpha|^{i_2}}{(i_1)!(i_2)!} \sup_{u,\omega}\Big\| \frac{ {\partial}^{i_1+i_2}}{\partial u^{i^1}\partial \omega^{i_2} }f_{u,\omega}\Big\|_2 
\\ & =\lsum_{i_1,i_2 \in \nnum:i_1+i_2=p} o(\btT^{i_1} |\omega-\alpha|^{i_2} ) \qquad p \ge 2. \tageq 
\end{align*}
In order to derive the mean of the estimator, we set $v \btT =\frac{r-N/2}{T}$ and recall that the taper function relates to a smoothing kernel $\Kt$  in time direction by
\begin{equation}
\Kt(v)=\SSS{\frac{1}{H_2}}\,h^2\Big(\SSS{\frac{v+1}{2}}\Big)
\end{equation}
for $v\in[-\tfrac{1}{2},\tfrac{1}{2}]$ with bandwidth $\btT=N/T$. It then follows from \eqref{eq:Taylor} that a Taylor expansion about to the point $x=(u_j,\omega_o)$ yields
\begin{align*}
& \mean (\hatFT_{u_j,\omega_o})=\F_{u_j,\omega_o} + \sum_{i=1}^{2} \frac{1}{i! }\btT^i \int v^i \Kt (v) dv \, \int_{\Pi}  \Kf(\alpha) d\alpha \frac{\partial^i}{\partial u^i} \F_{u,\omega}\Big|_{(u,\omega)=x}  \\&+ 
\sum_{i=1}^{2} \frac{1}{i!} \bfT^{i} \int_{\Pi}\alpha^i \Kf(\alpha) d\alpha  \int \Kt(v)dv  \frac{\partial^i}{\partial \omega^i}\F_{u\,\omega}\Big|_{(u,\omega)=x}\\&+
\frac{1}{2}\btT \bfT \int v \Kt (v) dv \, \int_{\Pi} \alpha \Kf(\alpha) d\alpha  \Big(\frac{\partial^2}{\partial u \partial{\omega}} \F_{u,\omega}\Big|_{(u,\omega)=x}+ \frac{\partial^2}{\partial {\omega} \partial{u}} \F_{u,\omega}\Big|_{(u,\omega)=x}\Big)
+ R_{T,p}.
\end{align*}
Because the smoothing kernels are symmetric around $0$, we obtain 
\begin{align*}
\mean (\hatFT_{u_j,\omega_o})&=\F_{u_j,\omega_o} +\frac{1}{2} \btT^2 \kappat{t}\frac{\partial^2}{\partial u^2} \F_{u,\omega}\Big|_{(u,\omega)=x} + 
\frac{1}{2} \bfT^{2} \kappaf{2,1}  \frac{\partial^2}{\partial \omega^2}\F_{u\,\omega}\Big|_{(u,\omega)=x}\\& +o(\btT^2) + o(\bfT^2)+O\big(\SSS{\frac{\log(\btT\,T)}{\btT\,T}}\big),\tageq
\end{align*}
where the error terms follow from \eqref{Rwind} and Theorem \ref{exp}, respectively. This establishes Result $i)$ of Theorem \ref{expsmooth}. \\

For the proof of the covariance structure, we note that
\[
\cov\big(\hatFT_{u,\omega_1},\hatFT_{u,\omega_2}\big)
=\int_{\Pi^2}\KfT(\omega_1-\lambda_1)\,\KfT(\omega_2-\lambda_2)\,\cov\big( \IT_{u,\lambda_1},\IT_{u,\lambda_2}\big)\, d\lambda_1\,d\lambda_2,
\]
where by Theorem \ref{exp}
\begin{align*}
\Cov&\big(\IT_{u,\lam_1},\IT_{u,\lam_2}\big)\\
&=\SSS{\frac{1}{4\pi^2\,\HNN{4}(0)^2}}\,\Big[S_{1423}
\Big(\HNN{2}\big(\F_{{u_{\bullet}},\lam_1},\lam_1-\lam_2\big)
\otimes\HNN{2}\big(\F_{{u_{\bullet}},-\lam_1},\lam_2-\lam_1\big)\Big)\\
&\qquad
+S_{1324}\Big(\HNN{2}\big(\F_{{u_{\bullet}},\lam_1},\lam_1+\lam_2\big)
\otimes
\HNN{2}\big(\F_{{u_{\bullet}},-\lam_1},-\lam_1-\lam_2\big)\Big)\Big]
+O\Big(\SSS{\frac{\log(N)}{N}}\Big).
\end{align*}
We treat the two terms of the covariance tensor separately. Starting with the first term, we have
\begin{align*}
&\biggnorml\int_{\Pi^2}
\KfT(\omega_1-\lambda_1)\Big[\KfT(\omega_2-\lambda_2)\,
\big[\HNN{2}\big(F_{{u_{\bullet}},\lam_1},\lam_1-\lam_2\big)\,
\otimes\HNN{2}\big(\F_{{u_{\bullet}},-\lam_1},\lam_2-\lam_1\big)\big]\\
&\qquad\qquad\qquad\qquad-\KfT(\omega_2-\lambda_1)\,
\big|\HNN{2}(\lam_1-\lam_2)\big|^2\,\big(\F_{u,\lam_1}
\otimes\F_{u,-\lam_1}\big)\Big]\,d\lambda_1\,d\lambda_2\biggnormr_2\\
&\leq\biggnorml\int_{\Pi^2}
\KfT(\omega_1-\lambda_1)\,\KfT(\omega_2-\lambda_2)\,
\Big[\HNN{2}\big(F_{{u_{\bullet}},\lam_1},\lam_1-\lam_2\big)\,
\otimes\HNN{2}\big(\F_{{u_{\bullet}},-\lam_1},\lam_2-\lam_1\big)\\
&\qquad\qquad\qquad\qquad
-\big|\HNN{2}(\lam_1-\lam_2)\big|^2\,\big(\F_{u,\lam_1}
\otimes\F_{u,-\lam_1}\big)\Big]\,d\lambda_1\,d\lambda_2\biggnormr_2\\ 
&\qquad+\biggnorml\int_{\Pi^2}
\KfT(\omega_1-\lambda_1)
\big[\KfT(\omega_2-\lambda_2)-\KfT(\omega_2-\lambda_1)\big]\\
&\qquad\qquad\qquad\qquad
\times\big|\HNN{2}(\lam_1-\lam_2)\big|^2\,\big(\F_{u,\lam_1}
\otimes\F_{u,-\lam_1}\big)\,d\lambda_1\,d\lambda_2\biggnormr_2.
\end{align*}
Since $\F_{u,\lam}$ is uniformly Lipschitz continuous in $u$, we have
$\snorm{\F_{{u_{r}},\lam}-\F_{u,\lam}}_2\leq C\,\tfrac{N}{T}$ and hence the first term on the right hand side is bounded by
\[
C\,\int_{\Pi^2}
\bfT^2\,L_{\frac{1}{\bfT}}(\omega_1-\lambda_1)^2\,
L_{\frac{1}{\bfT}}(\omega_2-\lambda_2)^2\,
L_N(\lam_1-\lam_2)^2\,\SSS{\frac{N}{T}}\,d\lambda_1\,d\lambda_2
\leq C\,\SSS{\frac{N^2}{\bfT\,T}}.
\]
For the second term, we exploit uniform Lipschitz continuity of the kernel function $\Kf$ to get the upper bound
\[
C\,\int_{\Pi^2}
\KfT(\omega_1-\lambda_1)^2\,\bfT^{-2}\,|\lam_1-\lam_2|\,
L_N(\lam_1-\lam_2)^2\,d\lambda_1\,d\lambda_2
\leq
C\,\SSS{\frac{\log(N)}{\bfT^2}}.
\]
In total we obtain
\[
\bigsnorm{\cov\big(\hatFT_{u,\omega_1},\hatFT_{u,\omega_2}\big)}_2
=O\Big(\SSS{\frac{\log(N)}{\bfT^2\,N^2}}\Big)
+O\Big(\SSS{\frac{1}{\bfT\,T}}\Big)
+O\Big(\SSS{\frac{\log(N)}{N}}\Big)
\]
uniformly in $\omega_1,\omega2\in[-\pi,\pi]$ and $u\in[0,1]$.
\end{proof}

\begin{proof}[Proof of Proposition \ref{sharp}]
A change of variables shows that \eqref{eq:covsmooth} can be written as   
\begin{align*}
\btT\,&\bfT\,T\,\Cov\big(\innerprod{\hatFT_{u,\omega_1}}{ g_1 \otimes g_2}_{H_\cnum\otimes H_\cnum},\innerprod{\hatFT_{u,\omega_2}}{ g_3 \otimes g_4}_{H_\cnum\otimes H_\cnum}\big) \\
& = 2\pi\,\bfT\,\norm{\Kt}^2_2 \int_{\Pi}\KfT(\omega_1-\omega_2-\lambda)\,\KfT(\lambda)\, \innerprod{\F_{u,\omega_2-\lambda}\,g_3}{g_1}\,\innerprod{\F_{u,-\omega_2-\lambda}\,g_4}{g_2}\,d\lambda\\
&\qquad + 2\pi\,\bfT\,\norm{\Kt}^2_2  \int_{\Pi}\KfT(\omega_1+\omega_2-\lambda)\,\KfT(\lambda) \, \innerprod{\F_{u,-\omega_2+\lambda}\,g_4}{g_1}\,\innerprod{\F_{u, \omega_2-\lambda}\,g_3}{g_2}\, d\lambda\\
&\qquad + O\big(\bfT\log(\btT\,T)\big)+O(\btT^2\,\bfT)+O\big((\btT\,\bfT\,T)^{-1}\big).
\tageq \label{covsmoothia}
\end{align*}
The error terms will tend to zero under assumption (A8). Since the product of the two kernels in the first integral is exactly zero whenever $|\lambda-(\omega_1-\omega_2)|>\bfT$ or $\lambda>\bfT$, the first integral vanishes for large enough $T$ unless $\omega_1=\omega_2$. For $\omega_1=\omega_2$, the integral in the first term becomes 
\begin{align*}
\int_{\Pi}&\KfT(-\lambda)\,\KfT(\lambda)\,\innerprod{\F_{u,\omega_1+\lambda}\,g_3}{g_1}\,\innerprod{\F_{u, -\omega_1-\lambda}\,g_4}{g_2}\,d\lambda\\
\intertext{and further by symmetry of the kernel}
&=\int_{\Pi}\KfT(\lambda)^2\,\innerprod{\F_{u,\omega_1+\lambda}\,g_3}{g_1}\,
\innerprod{\F_{u,-\omega_1-\lambda}}{g_4}{g_2}\,d\lambda.
\end{align*}
We note that $\norm{\Kf}^{-2}_2\,\KfT(\lambda)^2$ satisfies the properties of an approximate identity \citep[e.g.,][]{Edwards1967}. Hence application of Lemma F.15 of \citet{PanarTav2013a},
which covers approximate identities in a functional setting, yields that the integral converges to
\begin{align*}
\norm{\Kf}^2_2\,\innerprod{\F_{u,\omega_1}\,g_3}{g_1} \, \innerprod{\F_{u,-\omega_1}\,g_4}{g_2},
\end{align*}
with respect to $\norm{\cdot}_2$. Since the integral in the second term in \eqref{covsmoothia} vanishes unless $\omega_1=-\omega_2$, we can apply a similar argument, which proves the proposition.
\end{proof}

\begin{proof}[Proof of Theorem \ref{IMSE}]
We decompose the difference in terms of its variance and its squared bias. That is,
\begin{align*}
\int_{\Pi} & \mean \snorm{\hatFT_{u, \omega} - \mean\hatFT_{u,\omega} +\mean\hatFT_{u,\omega}-\F_{u,\omega}}^2_2\,d\omega\\
&= \int_{\Pi} \mean \snorm{\hatFT_{u, \omega} - \mean\hatFT_{u,\omega} }^2_2\, d\omega
+\int_{\Pi} \mean\snorm{\mean\hatFT_{u,\omega}-\F_{u,\omega} }^2_2\, d\omega \tageq \label{VBd}.
\end{align*}
The cross term cancels which is easily seen by noting that $\mean\big(\hatFT_{u,\omega}-\mean(\hatFT_{u,\omega})\big)=O_{H_\cnum}$ and hence
\[
\mean\Big(\biginnerprod{\hatFT_{u,\omega}-\mean(\hatFT_{u,\omega})}%
{\mean(\hatFT_{u,\omega})-\F_{u,\omega}}_{H_\cnum\otimes H_\cnum}\Big)=0
\]
for all $u\in[0,1]$ and $\omega\in[-\pi,\pi]$. Consider the first term of \eqref{VBd}. Self-adjointness of $\hatFT_{u,\omega}$ and $\E\|\XT{t}\|^4_2 <\infty$ imply that
\begin{align*}
\trace\big(\Cov(\hatFT_{u,\omega},\hatFT_{u,\omega})\big)&= 
\sum_{n,m=1}^{\infty}\biginnerprod{\E[\big(\hatFT_{u,\omega}-\E(\hatFT_{u,\omega})\big){\otimes}\big(\hat{\F}^{(T)}_{u,\omega}-\E(\hatFT_{u,\omega})\big)\psi_{nm}}{\psi_{nm}}\\& =
\E\sum_{n,m=1}^{\infty}\big|\innerprod{\big(\hatFT_{u,\omega}-\E(\hatFT_{u,\omega})\big)}{\psi_{nm}}\big|^2_2 = \E\snorm{\hatFT_{u,\omega}-\E\hatFT_{u,\omega}}^2_2 < \infty
\end{align*}
for some orthonormal basis $\{\psi_{nm}\}$ of $L^2_\cnum([0,1]^2)$. By Fubini's theorem and Corollary \ref{sharp2}, we thus find
\begin{align*}
\int_{\Pi}\mean\snorm{\hatFT_{u,\omega}-\mean\hatFT_{u,\omega} }^2_2\, d\omega
=\int_{\Pi}\int_{[0,1]^2} \Var(\hatfT_{u,\omega} (\tau, \sigma) )\,d\tau \,d\sigma\,d\omega =O\Big(\SSS{\frac{1}{b_{\text{t}} b_{\text{f}} T}}\Big).
\end{align*}
Theorem \ref{expsmooth} then yields for the second term of \eqref{VBd}
\begin{align*}
\int_{\Pi}\snorm{\F_{u,\omega}-\mean\hatFT_{u,\omega} }^2_2\,d\omega
&=\int_{\Pi}\int_{[0,1]^2} \big|f_{u,\omega}(\tau,\sigma)-\mean\hatfT_{u,\omega}(\tau, \sigma)\big|^2\,d\tau\,d\sigma\,d\omega\\ &=O\Big(b_{\text{t}}^2+b_{\text{f}}^2
+\SSS{\frac{\log b_{\text{t}} T}{b_{\text{t}} T}}\Big)^2. 
\end{align*}
\end{proof}

\begin{proof}[Proof of Proposition \ref{convCLT}]
We have
\begin{equation}
\begin{split}
\cum&\big(\ET_{u,\omega_1}(\psi_{m_1n_1}),\ldots,
\ET_{u,\omega_k}(\psi_{m_kn_k})\big)
=\SSS{\frac{(\btT\bfT T)^{k/2}}{\HNN{2}(0)^k}}\int_{{\Pi}^k}
\lprod_{j=1}^{k}\KfT(\omega_j-\lambda_j)\\
&\times
\cum\Big(\DT_{u,\omega_1}(\psi_{m_1})\,\DT_{u,-\omega_1}(\psi_{n_1}),\ldots,
\DT_{u,\omega_k}(\psi_{m_k})\,\DT_{u,-\omega_k}(\psi_{n_k})\Big)\,
d\lam_1\cdots d\lam_k,
\end{split}
\label{eq:hocum1}
\end{equation}
where $\DT_{u,\omega}(\phi)=\biginnerprod{\DT_{u,\omega}}{\phi}$ for $\phi\in L^2_\cnum([0,1])$. Application of the product theorem for cumulants \citep[e.g.][Theorem 2.3.2]{Brillinger} yields
for the cumulant
\begin{equation}
\begin{split}
\cum\big(\DT_{u,\omega_1}(\psi_{m_1})\,\DT_{u,-\omega_1}(\psi_{n_1}),&\ldots,
\DT_{u,\omega_k}(\psi_{m_k})\,\DT_{u,-\omega_k}(\psi_{n_k})\Big)\\
&=\sum_{i.p.}
\lprod_{l=1}^{M}\cum\big(\DT_{u,\gamma_{p}}(\psi_{r_{p}}),p\in P_l\big),
\end{split}
\label{cumDDD}
\end{equation}
where the summation extends over all indecomposable partitions
$P=\{P_1,..., P_M\}$ of the table
\[
\begin{matrix}
  (1,0) & (1,1)\\
  \vdots  & \vdots \\
  (k,0) & (k,1)
 \end{matrix},
\]
and, for $p=(i,j)$, $\gamma_{p}=\gamma_{ij}=(-1)^j\,\lam_i$ as well as $r_p=r_{ij}=m_i^{1-j}n_i^{j}$ for $i=1,\ldots,k$ and $j\in\{0,1\}$. For the next steps, we further denote the elements of $P_l$ with $|P_l|=d_l$ by $p_{l1},\ldots,p_{ld_l}$. Then, by \eqref{eq:cumDNop}, we obtain further for the above cumulant
\begin{equation}
\label{cumDDDnext}
\begin{split}
\sum_{i.p.}\lprod_{l=1}^M
\int_{\Pi^{d_l-1}}\int_{[0,1]^{d_l}}
&\Big[\lotimes_{s=1}^{d_l}
\HN\big(\AAttT{t_{u,\bullet}}{\alpha_{s}},\gamma_{p_{ls}}-\alpha_{s}\big)
\F^\veps_{\alpha_{1},\ldots,\alpha_{d_l-1}}\Big](\tau_{1},\ldots,\tau_{d_l})\\
&\times\lprod_{s=1}^{d_l}\overline{\psi_{r_{p_{ls}}}(\tau_s)}\,
d\tau_1\cdots d\tau_{d_l}\,\eta(\alpha_1+\ldots+\alpha_{dl})\,
d\alpha_1\cdots d\alpha_{d_l}.
\end{split}
\end{equation}
Noting that the inner integral is a inner product in the tensor product space, we get
\begin{align*}
\Big|&\Biginnerprod{\lotimes_{s=1}^{d_l}
\HN\big(\AAttT{t_{u,\bullet}}{\alpha_{s}},\gamma_{p_{ls}}-\alpha_{s}\big)
\F^\veps_{\alpha_{1},\ldots,\alpha_{d_l-1}}}{
\otimes_{s=1}^{d_l}\psi_{r_{p_{ls}}}}\Big|\\
&\leq\Bigsnorm{\lotimes_{s=1}^{d_l}\HN\big(\AAttT{t_{u,\bullet}}{\alpha_{s}},
\gamma_{p_{ls}}-\alpha_{s}\big)\,
\F^\veps_{\alpha_{1},\ldots,\alpha_{d_l-1}}}_2\,
\Bignorm{\lotimes_{s=1}^{d_l}\psi_{r_{p_{ls}}}}_2\\
&\leq\Bigsnorm{\lotimes_{s=1}^{d_l}\HN\big(\AAttT{t_{u,\bullet}}{\alpha_{s}},
\gamma_{p_{ls}}-\alpha_{s}\big)}_\infty\,
\bigsnorm{\F^\veps_{\alpha_{1},\ldots,\alpha_{d_l-1}}}_2\,
\lprod_{s=1}^{d_l}\bignorm{\psi_{r_{p_{ls}}}}_2.
\end{align*}
Noting that by Lemma \ref{Dh1993A5}
\[
\Bigsnorm{\HN\big(\AAttT{t_{u,\cdot}}{\alpha_{s}},
\gamma_{p_  {ls}}-\alpha_{s}\big)}_\infty
\leq K\,L_N(\gamma_{p_{ls}}-\alpha_{s})
\]
for some constant $K$, we get together with $ \snorm{\F^\veps_{\alpha_{1},\ldots,\alpha_{d_l-1}}}_2\leq K'$
as an upper bound for \eqref{cumDDDnext}
\begin{align*}
K\sum_{i.p.}\lprod_{l=1}^M&
\int_{\Pi^{d_l}}
\lprod_{s=1}^{d_l}L_N(\gamma_{p_{ls}}-\alpha_{s})\,
\eta(\alpha_1+\ldots+\alpha_{dl})\,d\alpha_1\cdots d\alpha_{d_l}\\
\intertext{and further by repeated use of Lemma \ref{Dh1993A4}(v)}
&\leq K\sum_{i.p.}\lprod_{l=1}^M
L_N(\bar\gamma_{l})\,\log(N)^{d_l-1}
\leq K\,\log(N)^{2k-M}\sum_{i.p.}\lprod_{l=1}^{M}L_N(\bar\gamma_{l}).
\end{align*}
Substituting the upper bound for the cumulant in \eqref{eq:hocum1} and noting that $\frac{1}{N}\HNN{2}(0)\to\norm{h}^2_2$ as $N\to\infty$,
we find
\begin{equation}
\begin{split}
\big|\cum&(\ET_{u,\omega_1}(\psi_{m_1n_1}),\ldots,
\ET_{u,\omega_k}(\psi_{m_kn_k})\big)\big|\\
&\leq\SSS{\frac{C\,\bfT^{k/2}\,\log(N)^{2k-M}}{N^{k/2}}}\sum_{i.p.}
\int_{{\Pi}^k}\lprod_{j=1}^{k}\KfT(\omega_j-\lam_j)\,\lprod_{l=1}^{M}L_N(\bar\gamma_{l})\,d\lam_1\cdots d\lam_k.
\end{split}
\end{equation}
It is sufficient to show that for each indecomposable partition $\{P_1,\ldots,P_M\}$ the corresponding term in the above sum tends to zero. First, suppose that $M=k$. Bounding the factors $\KfT(\omega_i-\lam_i)$ by $\norm{\Kf}_\infty/\bfT$ for $i=2,\ldots,k$ and integrating over $\lam_3,\ldots,\lam_{k}$, we obtain by Lemma \ref{EichlerL}(i) as an upper bound
\begin{align*}
&\SSS{\frac{C\,\log(N)^{2k-M}}{\bfT^{k/2-1}N^{k/2}}}\,
\int_{\Pi^2}\KfT(\omega_{1}-\lam_{1})\,
L_N(\lam_{1}\pm\lam_{2})^2\,d\lam_1\,d\lam_2\\
&\qquad\leq
\SSS{\frac{C\,\log(N)^{2k-2}}{\bfT^{k/2-2}N^{k/2}}}\,\int_{\Pi^2}
L_{\bfT^{-1}}(\omega_{1}-\lam_{1})^2\,
L_N(\lam_{1}\pm\lam_{2})^2\,d\lam_2\,d\lam_1\\
&\qquad\leq
\SSS{\frac{C\,\log(N)^{2k-2}}{\bfT^{k/2-2}N^{k/2}}}\,\int_{\Pi}
N\,L_{\bfT^{-1}}(\omega_{1}-\lam_{1})^2\,d\lam_1
\leq
\SSS{\frac{C\,\log(N)^{2k-2}}{(\bfT N)^{k/2-1}}},
\end{align*}
where we have $\KfT(\omega)\leq \bfT\,L_{\bfT^{-1}}(\omega)$ and repeatedly Lemma \ref{Dh1993A4}(iv). Next, if $M<k$ we select variables $\lam_{i_1},\ldots,\lam_{i_{k-2}}$ according to Lemma \ref{EichlerL}(ii) and bound all corresponding factors $\KfT(\omega_{i_j}-\lam_{i_j})$ for $j=1,\ldots,k-2$ by $\norm{\Kf}_\infty/\bfT$. Then integration over the $k-2$ selected variables yields the upper bound
\begin{align*}
&\SSS{\frac{C\,\log(N)^{3k-M-2}}{\bfT^{k/2-2}\,N^{k/2-1}}}\,
\int_{\Pi^2}\KfT(\omega_{i_{k-1}}-\lam_{i_{k-1}})\,\KfT(\omega_{i_k}-\lam_{i_k})
\,d\lam_{i_{k-1}}\,d\lam_{i_k}\\
&\qquad\leq\SSS{\frac{C\,\bfT\,\log(N)^{3k-M-2}}{\bfT^{k/2-1}\,N^{k/2-1}}},
\end{align*}
since $\norm{\KfT}_1=1$. Since $\bfT\,N=\bfT\,\btT\,T\to\infty$ and $k/2-1>0$,
the upper bounds tend to zero as $T\to\infty$, which completes the proof.
\end{proof}

\setcounter{section}{-1}
\setcounter{equation}{0}
\def\theequation{B\arabic{section}.\arabic{equation}}
\def\thesection{B}
\section{}In Appendix B, we provide additional technical material necessary to complete the proofs of the main part of this paper. Section \ref{properties} contains background material on operator theory. Section \ref{cumprops} provide background on the higher order structure of random functions as well as an important result on the existence of a stochastic integral, necessary to define functional Cram{\'e}r representations. Section \ref{taperandl} introduces the necessary background on tapering on function spaces. 
\setcounter{section}{0}
\setcounter{equation}{0}
\def\theequation{B\arabic{section}.\arabic{equation}}
\def\thesection{B\arabic{section}}

\section{Some operator theory} \label{properties}

We start with a general characterization of a tensor product of a finite sequence of vector spaces, which in particular holds for sequences of Hilbert spaces. 
\begin{definition}[Algebraic tensor product of Banach spaces]\label{Algtensor}
Given a finite sequence of vector spaces $V_1,\ldots, V_k$ over an arbitrary field $\mathbb{F}$, we define the {\em algebraic tensor product}
$V_1 \otimes \cdots \otimes V_k$ as a vector space with a multi--linear map $V_1 \times \cdots \times V_k \to W$ given by $(f_1,\ldots, f_k) \to (f_1 \otimes \cdots \otimes f_k)$ such that, for every linear map $\mathcal{T}:  V_1 \times \cdots \times V_k \to W$, there is unique k-linear map $\tilde{\mathcal{T}}: V_1 \times \cdots \times V_k \to W$ that satisfies
\begin{align*}
\mathcal{T}(f_1,\ldots, f_k) =\tilde{\mathcal{T}}(f_1 \otimes \cdots \otimes f_k).
\end{align*}
\end{definition}
Here, uniqueness is meant up to isomorphisms. The tensor product can be viewed as a linearized version of the product space $V_1 \times \cdots \times V_k$ satisfying equivalence relations of the form $a(v_1,v_2) \sim (a v_1,v_2) \sim (v_1,av_2)$ where $a \in \mathbb{K}$ and $v_1 \in V_1 , v_2 \in V_2$, which induce a quotient space. These relationships uniquely identify the points in the product space $V_1 \times \ldots \times V_k$ that yield multi--linear relationships. In a way, the tensor product $\bigotimes_{j=1}^{k} V_j$ can thus be viewed as the `freest' way to put the respective different vector spaces $V_1,\ldots,V_k$ together. We mention in particular that the algebraic tensor product satisfies the associative law, i.e., $(V_1 \otimes V_2 ) \otimes V_3 = V_1 \otimes (V_2  \otimes V_3)$, and hence it will often be sufficient to restrict attention to $k=2$. \\

The algebraic tensor product of two Hilbert spaces $H_1$ and $H_2$ is itself not a Hilbert space. We can however construct a Hilbert space by considering the inner product acting on $H_1 \otimes H_2$ given by
\begin{align*}
\langle x \otimes y, x' \otimes y'  \rangle_{H_1 \otimes H_2} =\langle x, x' \rangle \langle y,y' \rangle, \quad x,x' \in H_1,\, y,y' \in H_2
\end{align*}
and then taking the completion with respect to the induced norm $\|\cdot\|_{H_1 \otimes H_2}$.
The completed space, denoted by $H_1 \widehat{\otimes} H_2$, is identifiable with the Hilbert-Schmidt operators and is referred to as the {\em Hilbert Schmidt tensor product}. Throughout this work, when reference is made to the tensor product space of Hilbert spaces we mean the latter space. When no confusion can arise, we shall moreover abuse notation slightly and denote $H_1 \widehat{\otimes} H_2$ simply by $H_1 \otimes H_2$.

\begin{definition} \label{tensorproductop}
The tensor product $(A \otimes B) \in S_P(H)\otimes S_p(H)\cong S_{p}(S_{p}(H))$  between two operators $A,B \in S_{p}(H)$ is defined as 
\begin{align}\label{eq:tensor1abc}
(A\, {\otimes} \,B)(x \otimes y) =  Ax \otimes By,
\end{align} for $x,y \in H$. It follows straightforwardly from the property
\begin{align}
(x \otimes y)z = \langle z, y\rangle x, \quad z \in H,
\end{align} that for any $C \in S_p(H)$, we have the identity
\begin{align}\label{eq:tensorabc}
(A \otimes B)C =  ACB^{\dagger},
\end{align}
where $ B^{\dagger}$ denote the adjoint operator of $B$. 
\end{definition}

\begin{proposition}[H{\"o}lder's Inequality for operators] \label{holderoperator}
\textit{Let $H$ be a separable Hilbert space and $A,B \in S_{\infty}(H)$. Then the composite operator $AB$ also defines a bounded linear operator over $H$, i.e.,  $AB \in S_{\infty}(H)$. This operation satisfies the associative law. Moreover, let $1 \le p,q,r \le \infty$, such that $\frac{1}{r}=\frac{1}{q}+\frac{1}{p}$. If $A \in S_q(H)$ and $B \in S_p(H)$ then $AB \in S_r(H)$ and 
\begin{align*}
\snorm{AB}_r \le \snorm{A}_q \snorm{B}_p.
\end{align*}}
\end{proposition}

\begin{proposition} \label{propertieskernelsop}
\textit{Let $H = L^2_{\cnum}(T,\mu)$ be a separable Hilbert space, where $(T,\mu)$ is a measure space. The functions $a,b,c \in L^2_\cnum(T\times T,\mu\otimes\mu)$ induce operators $A,B,C$ on $H$ such that for all $x \in H$ }
\begin{align}\label{eq:kernel} 
Ax(\tau) =  \int_\mathcal{D} a(\tau, \sigma) x(\sigma) d\mu(\sigma),
\end{align}
and the composition operator $AB$ has kernel
\begin{align}\label{eq:kernelsingle} 
[AB](\tau, \sigma) =  \int_\mathcal{D} a(\tau, \mu_1) b(\mu_1, \sigma) d\mu_1,
\end{align}
for all $\tau,\sigma \in T$ $\mu$-almost everywhere. The tensor product operator $(A \otimes B) \in  S_{2}(S_{2}(H))$ in composition with $C$ has kernel
\begin{align}\label{eq:kerneldouble}
[(A \otimes B)C](\tau, \sigma) =  \int_\mathcal{D} \int_\mathcal{D} a(\tau, \mu_1) \overline{b(\sigma, \mu_2)} c(\mu_1, \mu_2) d\mu_1 d\mu_2.
\end{align}
Because $(A \otimes B)C$ has a well defined kernel in $L^2_\cnum(T\times T,\mu\otimes\mu)$, it can moreover be viewed as an operator on $H$. Using identity \eqref{eq:tensorabc}, this is the operator $ACB^{\dagger}$, where $B^{\dagger}$ has kernel $b^{\dagger}(\mu_2, \sigma) = \overline{b(\sigma,\mu_2)}$. 
\end{proposition}
\begin{corollary} \label{tensorprodHS}
Let $A_i,  i = 1,\cdots,k$ for k finite belong to $S_p(H)$ and let \begin{align*} {\boldsymbol{\psi}} = (\psi_{1} \otimes \cdots \otimes \psi_{k} ) \end{align*} be an element of $ \bigotimes_{i=1}^{k} H$. Then we have that the linear mapping \begin{align*}\mathcal{A} = \big(A_1 \otimes ...\otimes  A_k\big)\end{align*} satisfies $i)\, \|\mathcal{A}\boldsymbol{\psi}\|_2 <\infty$ and $ii)\, \snorm{ \mathcal{A}}_p <\infty$.
\end{corollary}
\begin{proof}[Proof of Corollary \ref{tensorprodHS}]

For $i)$, we have by proposition \ref{holderoperator},
\begin{align*} 
 \|\mathcal{A}\boldsymbol{\psi}\|_2& =  \|  \big(A_1 \otimes ...\otimes  A_k\big) \boldsymbol{\psi}\|_2 =  \snorm{  \big(A_1 \otimes ...\otimes  A_k\big)}_{\infty} \|\boldsymbol{\psi}\|_2 \\&\displaybreak[0]  \le  \snorm{  \big(A_2 \otimes ...\otimes  A_k\big)}_{\infty}  \snorm{ A_1}_{\infty} \|\boldsymbol{\psi}\|_2 \le \prod_{i=1}^{k}\snorm{A_i}_{\infty}\|\boldsymbol{\psi}\|_2 \le \prod_{i=1}^{k}\snorm{A_i}_{p}\|\boldsymbol{\psi}\|_2  < \infty. 
\end{align*}
In case $p=2$, the latter equals $ \prod_{i=1}^{k}\|a_i\|_{2}\|\boldsymbol{\psi}\|_2$ by proposition \ref{HSkernel}. 
Property $ii)$ holds since for any $A_1, A_2 \in S_p(H)$, we have $\snorm{A_1 \otimes A_2}_p = \snorm{A_1}_p \snorm{A_2}_p$. To illustrate the second property, observe that if $p=2$ we obtain
\begin{align*}
& \|\mathcal{A}\|^2_2 = \|A_1 \otimes ...\otimes  A_k\|^2_2 = \int_{[0,1]^{2k}}| a_1(\tau_1, \mu_1)..a_k(\tau_k,\mu_k)|^2d\tau_1..d\tau_k d\mu_1..d\mu_k 
\\& =  \int_{[0,1]^{2}} a_1(\tau_1, \mu_1) \overline{a_1(\tau_1, \mu_1)}d\tau_1, d\mu_1 .. \int_{[0,1]^{2}}a_k(\tau_k,\mu_k)\overline{a_k(\tau_k,\mu_k)} d\tau_k d\mu_k \\& = \|a_1\|^2_2 ..\|a_k\|^2_2 < \infty. 
\end{align*}
\end{proof}

\begin{proposition}[Neumann series]
\label{Neumann}
\textit{Let $A$ be a bounded linear operator on $H$ and $I_H$ be the identity operator. If $\snorm{ A}_{\infty} < 1$, the operator $I_H-A$ has a unique bounded inverse on $H$ given by }
\begin{align} 
(I_H-A)^{-1} = \sum_{k=0}^{\infty} A^k.
\end{align}
If $A \in S_2(H)$ with $\snorm{A}_2 <1$, then this equality holds in Hilbert-Schmidt norm. 
\end{proposition}

\begin{proof}
We only show the case $A \in S_2(H)$. Note that the space $S_2(H)$ is a Hilbert space. Then for $m<n$,
\begin{align*}
\Bigsnorm{\lsum_{k=0}^{m} A^k - \lsum_{k=0}^{n}A^k}_2
\leq \lsum_{k=m+1}^{n}\bigsnorm{A}^k_2
\leq \frac{\snorm{A}^{m+1}_2}{1-\snorm{A}_2},
\end{align*}
which shows that the partial sum forms a Cauchy sequence and hence has a limit
$A^*$ in $S_2(H)$. Furthermore, we have
\[
(I_H-A)\,A^*
=\lim_{n\to\infty}(I_H-A)\,\lsum_{k=0}^n A^n
=\lim_{n\to\infty}\big(I_H-A^{n+1})=I_H
\]
in $S_2(H)$, which shows that $A^*$ is the inverse of $I_H-A$.
\end{proof}

\begin{proposition}[Hilbert-Schmidt operators as kernel operator] \label{HSkernel}
Let $H = L^2_{\cnum}(T,\mu)$ be a separable Hilbert space, where $(T,\mu)$ is a measure space, and let $A$ be an operator on $H$. Then $A\in S_2(H)$ if and only if it is an integral operator, that is, there exists a function $a\in L^2_\cnum(T\times T,\mu\otimes\mu)$ such that 
\[
A\,x(\tau) = \int a(\tau,\sigma)\, x(\sigma)\,d\mu(\sigma)
\]
for all $\tau\in T$ $\mu$-almost everywhere. Moreover, we have $\snorm{A}_2=\norm{a}_2$.
\end{proposition}
\begin{proof}
First, suppose $A$ is an integral operator on $H$ with kernel $a \in L^2_\cnum(T\times T,\mu\otimes\mu)$. Because $H$ is separable, it has a countable orthonormal basis $\{\psi_n\}_{n\in\nnum}$. For fixed $\tau \in M$, the function $a_{\tau}(\sigma) = a(\tau,\sigma)$ defines a measurable function on $ L^2_\cnum(T,\mu)$. We can therefore write
\begin{align*}
A\psi_n(\tau) = \int a(\tau, \sigma) \psi_n(\sigma) \mu(\sigma) = \langle a_\tau, \overline{\psi_n} \rangle.
\end{align*}
Observe that $\{\overline{\psi_n}\}_{n \ge 1}$ also forms a orthonormal basis of $H$. An application of the Cauchy-Schwarz Inequality gives $|\langle a_{\tau}, \overline{\psi_n} \rangle |^2  \le \| a_{\tau}\|^2 \|  \overline{\psi_n} \|^2 < \infty$ and therefore 
\begin{align*}
\sum_{n=1}^{m}|\langle a_{\tau}, \overline{\psi_n} \rangle |^2 \le  \sum_{n=1}^{\infty}|\langle a_{\tau}, \overline{\psi_n} \rangle |^2 = \| a_{\tau}\|_2^2 < \infty,
\end{align*}
by Parseval's Identity. Hence, as a corollary of the Monotone and Dominated Convergence Theorem we find
\begin{align*}
\snorm{A}^2_2 & = \sum_{n=1}^{\infty}\|A \psi_n \|^2 =  \lim_{m \to \infty} \sum_{n=1}^{m} |\langle a_{\tau},\overline{\psi}_n \rangle |^2 d\tau = \int  \lim_{m \to \infty} \sum_{n=1}^{m}  |\langle a_{\tau},\overline{\psi}_n \rangle |^2 d\tau \displaybreak[0]\\& =\int \|a_{\tau}\|^2 d\tau = \int \int |a(\tau,\sigma)|^2 d\sigma d\tau  = \|a\|^2_2< \infty, 
\end{align*}
showing $A$ is Hilbert Schmidt and $\snorm{A}_2=\|a\|_2$. Now suppose $A$ is Hilbert Schmidt. In this case, we have by definition $\sum_{n=1}^{\infty}\|A \psi_n \|^2 < \infty $ and consequently the series $\sum_{n=1}^{\infty} A \psi_n$ converges in $ L^2_\cnum(T,\mu)$. Therefore the function 
\[a(\tau, \sigma) = \sum_{n=1}^{\infty} A \psi_n(\tau) \overline{\psi_n(\sigma)}\] 
will be well-defined on $L^2_\cnum(T\times T,\mu\otimes\mu)$. Hence, for any element $x \in  L^2_\cnum(T,\mu)$, the Dominated Convergence Theorem yields
\begin{align*}
Ax(\tau) &= A\Big(\lim_{m \to \infty}\sum_{n= 1}^{m} \langle x, \psi_n \rangle \psi_n \Big)(\tau) = \lim_{m \to \infty}\sum_{n= 1}^{m}\langle x, \psi_n \rangle A \psi_n(\tau) \\& =\lim_{m \to \infty}\sum_{n= 1}^{m}\Big( \int x(\sigma) \overline{\psi_n(\sigma) }d\sigma\Big) A \psi_n(\tau) = \lim_{m \to \infty}\Big( \int x(\sigma)\sum_{n= 1}^{m}  \overline{\psi_n(\sigma)} A \psi_n(\tau) d\sigma\Big)\\&
\int x(\sigma)\sum_{n\ge 1} \overline{\psi_n(\sigma)} A \psi_n(\tau) d\sigma = \int x(\sigma)a(\tau,\sigma) d\sigma. 
\end{align*}
\end{proof}

\begin{theorem}[Product Rule on Banach spaces]
\label{ProdBan}
Let $E, F_1, F_2, G$ be Banach spaces and let $U \subset E$ be open. Suppose that $f:U \to F_1$, and  $G:U \in F_2$ are Fr{\'e}chet differentiable of order $k$. Let $Z(\cdot,\cdot) : F_1 \times F_2 \to G$ be a continuous bilinear map. Then, $Z(f,g) :U \to G$ is Fr{\'e}chet differentiable of order $k$ and 
\begin{align}
\frac{\partial Z}{\partial u}(f(u),g(u)) = Z(\frac{\partial f(u)}{\partial u}, g(u))+  Z(f(u),\frac{\partial g(u)}{\partial u}).
\end{align} 
\end{theorem}
For the proof, see for example \citet{Nelson1969}.

\section{Moment and cumulant tensors} 
\label{cumprops}

Let $X$ be a random element on a probability space $(\Omega,\aalg,\prob)$ that takes values in a separable Hilbert space $H$. More precisely, we endow $H$ with the topology induced by the norm on $H$ and assume that $X:\Omega\to H$ is Borel-measurable. Then the mean $\mean(X)$ of $X$ in $H$ exists and is given by
\[
\mean(X)=\lsum_{i\in\nnum}\mean\big(\innerprod{X}{\psi_i}\big)\,\psi_i,
\]
where $(\psi_i)_{i\in\nnum}$ is an orthonormal basis of $H$, provided that $\mean(\norm{X}^2_2)<\infty$.

For higher moments, it is appropriate to consider these as tensors in a tensor product space $H\otimes\cdots\otimes H$ of appropriate dimension. More precisely, let $X_1,\ldots,X_k$ be random elements in $H$. Then we define the moment tensor $\mean(X_1\otimes\cdots\otimes X_k)$ by
\[
\mean(X_1\otimes\cdots\otimes X_k)=\sum_{i_1,\ldots,i_k\in\nnum}
\mean\Big(\lprod_{j=1}^k \innerprod{X_j}{\psi_{i_j}}\Big)\,
\psi_{i_1}\otimes\cdots\otimes\psi_{i_k}.
\]
Similarly, we define the cumulant tensor $\cum(X_1,\ldots,X_k)$ by
\begin{align}\label{eq:CovOP}
\cum(X_1,\ldots,X_k)=\sum_{i_1,\ldots,i_k\in\nnum}
\cum\big(\innerprod{X_1}{\psi_{i_1}},\ldots,\innerprod{X_k}{\psi_{i_k}}\big)\,
\psi_{i_1}\otimes\cdots\otimes\psi_{i_k}.
\end{align}
The cumulants on the right hand side are as usual given by
\[
\cum\big(\innerprod{X_1}{\psi_{i_1}},\ldots,\innerprod{X_k}{\psi_{i_k}}\big)
=\sum_{\nu=(\nu_1,\ldots,\nu_p)}(-1)^{p-1}\,(p-1)!\,\lprod_{r=1}^p
\mean\Big(\lprod_{j\in\nu_r}\innerprod{X_j}{\psi_{i_j}}\Big),
\]
where the summation extends over all unordered partitions $\nu$ of $\{1,\ldots,k\}$. 

More generally, we also require the case where the $X_i$ are themselves tensors, that is, $X_i=\lotimes_{j=1}^{l_i}X_{ij}$, $i=1,\ldots,k$, for random elements $X_{ij}$ in $H$ with $j=1,\ldots,l_i$ and $i=1,\ldots,k$. In this case, the joint cumulant tensor $\cum(X_1,\ldots,X_k)$ is given by an appropriate generalization of the product theorem for cumulants \citep[Theorem 2.3.2]{Brillinger} to the tensor case,
\[
\cum(X_1,\ldots,X_k)=\!\!\!\sum_{r_{11},\ldots,r_{kl_k}\in\nnum}
\sum_{\nu=(\nu_1,\ldots,\nu_p)}\lprod_{n=1}^{p}
\cum\big(\innerprod{X_{ij}}{\psi_{r_{ij}}}|(i,j)\in \nu_n\big)\,
\psi_{r_{11}}\otimes\cdots\otimes\psi_{r_{kl_k}},
\]
where the summation extends over all indecomposable partitions $\nu=(\nu_1,\ldots,\nu_p)$ of the table
\[
\begin{matrix}
(1,1) &\cdots& (1,l_1)\\
\vdots&\ddots& \vdots\\
(k,1) &\cdots& (k,l_k).
\end{matrix}
\]
Formally, we also abbreviate this by
\begin{equation}
\label{prodcum}
\cum(X_1,\ldots,X_k)
=\sum_{\nu=(\nu_1,\ldots,\nu_p)}S_\nu\Big(\lotimes_{n=1}^{p}
\cum\big(X_{ij}|(i,j)\in\nu_n\big)\Big),
\end{equation}
where $S_\nu$ is the permutation that maps the components of the tensor back into the original order, that is,
$S_\nu\big(\otimes_{r=1}^p\otimes_{(i,j)\in\nu_r} X_{ij}\big)
=X_{11}\otimes\cdots\otimes X_{kl_k}$.

Next, let $A_1,\ldots,A_k$ linear bounded operators on $H$. As in Appendix \ref{properties}, let $A_1\otimes\cdots\otimes A_k$ be the operator on $H\otimes\cdots\otimes H$ given by
\[
(A_1\otimes\cdots\otimes A_k)(x_1\otimes\cdots\otimes x_k)
=(A_1\,x_1)\otimes\cdots\otimes(A_k\,x_k)
\]
for all $x_1,\ldots,x_k\in H$. The next proposition states that moment tensors---and hence also cumulant tensors by the above definitions---transform linearly.

\begin{proposition} \label{momlin}
Let $A_1,\ldots, A_k$ be bounded linear operators on $H$ and $X_1, \ldots, X_k$ be random elements in $H$. Then
\begin{equation}
\big(A_1\otimes\cdots\otimes A_k\big)\,
\mean\big(X_1 \otimes \cdots \otimes X_k\big)
= \mean\big((A_1\,X_1) \otimes \cdots \otimes (A_k\,X_k)\big).
\end{equation}
\end{proposition}

\newcommand{\bigginnerprod}[2]{\bigg\langle{#1},{#2}\bigg\rangle}

\begin{proof}
Let $\{\psi_i\}_{i\in\nnum}$ be an orthonormal basis of $H$. Using the definition of a moment tensor, we get
\begin{align*}
(A_1\otimes&\cdots\otimes A_k)\,\mean(X_1\otimes\cdots\otimes X_k)\\
&=\sum_{i_1,\ldots,i_k\in\nnum}
\mean\Big(\lprod_{j=1}^{k}\innerprod{X_j}{\psi_{i_j}}\Big)\,
(A_1\,\psi_{i_1})\otimes\cdots\otimes(A_k\,\psi_{i_k})\\
\intertext{and further, by representing $A_j\,\psi_{i_j}$ with respect to the chosen orthonormal basis,}
&=\sum_{i_1,\ldots,i_k\in\nnum}\sum_{n_1,\ldots,n_k\in\nnum}
\mean\Big(\lprod_{j=1}^{k}\innerprod{X_j}{\psi_{i_j}}\Big)
\lprod_{j=1}^{k}\innerprod{A_j\psi_{i_j}}{\psi_{n_j}}
(\psi_{n_1}\otimes\cdots\otimes\psi_{n_k})\\
&=\sum_{n_1,\ldots,n_k\in\nnum}\mean\bigg[\lprod_{j=1}^{k}
\bigginnerprod{A_j\bigg(
\lsum_{i_j\in\nnum}\innerprod{X_j}{\psi_{i_j}}\bigg)}{\psi_{n_j}}\,\bigg]
\big(\psi_{n_1}\otimes\cdots\otimes\psi_{n_k}\big)\\
&=\sum_{n_1,\ldots,n_k\in\nnum}\mean\Big[\lprod_{j=1}^{k}
\biginnerprod{A_j\,X_j}{\psi_{n_j}}\Big]\,
\big(\psi_{n_1}\otimes\cdots\otimes\psi_{n_k}\big)\\
&=\mean\big((A_1\,X_1)\otimes\cdots\otimes(A_k\,X_k)\big),
\end{align*}
where we have used linearity of the operators, of the inner product, and of the ordinary mean.
\end{proof}

As a direct consequence of the above proposition, we also have linearity of cumulant tensors. More precisely, for $i=1,\ldots,k$, let $X_i$ be a random tensor in $\bigotimes_{j=1}^{k} H$ and let $A_i$ be a linear bounded operator on the same tensor product space. Then
\begin{equation}
\label{cumlin}
\big(A_1\otimes\cdots\otimes A_k\big)\,\cum\big(X_1,\ldots,X_k\big)
=\cum\big(A_1\,X_1,\ldots,A_k\,X_k\big).
\end{equation}

\subsection{Higher order dependence under functional stationarity} \label{Highorderfuncstat}
For stationary processes, we follow convention and write the cumulant tensor \eqref{eq:CovOP} as a function of $k-1$ elements, i.e., we denote it by $\mathcal{C}_{t_1,\ldots,t_{k-1}}$. Under regularity conditions this operator is in $S_2(H)$ and we can define the corresponding $k$-th order cumulant kernel  $\cumk{k}$ of the process $X$  by
\begin{align}\label{eq:cumker}
\mathcal{C}_{t_1,\ldots,t_{k-1}} = \sum_{i_1,\ldots,i_k\in\nnum} \int_{[0,1]^k} \cumker{k} \lprod_{j=1}^k \psi_{i_j}(\tau_j) d\tau_1 \cdots d\tau_k \psi_{i_1} \otimes \cdots \otimes \psi_{i_k}
\end{align}
A sufficient condition that is often imposed for this to hold is $\mean\norm{X_0}^k_2 < \infty$. Similar to the second-order case, the tensor \eqref{eq:CovOP} will form a Fourier pair with a $k$-th order cumulant spectral operator given summability with respect to $\snorm{ \cdot }_p$ is satisfied. The $k$-th order cumulant spectral tensor is specified as
\begin{align}
\mathcal{F}_{\omega_1,..,\omega_{k-1}}
=(2\pi)^{1-k} \sum_{t_1,..,t_{k-1}\in\znum} \mathcal{C}_{t_1,..,t_{k-1}}\,
\exp\Big(-\im\lsum_{j=1}^{k-1} \omega_j\,t_j\Big),
\end{align}
where the convergence is in $\snorm{ \cdot }_p$. Properties on the kernels that are relevant in the time-dependent framework are discussed in section \ref{section3} of the main paper.

\begin{theorem}\label{cumWN}
Let $\{X_t\}_{t\in\znum}$ be a stationary stochastic process in $L^{2}([0,1])$ such that $\mean\norm{X_0}_2^k<\infty$ for all $k\in\nnum$ and $\sum_{t_1,\ldots,t_{k-1}=-\infty}^{\infty}\snorm{\cumop{k}}_2<\infty$.
Furthermore let
\[
\ZN{\omega}=\SSS{\frac{1}{2\pi}}\sum_{t=-N}^{N}X_t \int_{-\pi}^{\omega}  e^{-\im\lambda t}\,d\lambda.
\]
Then there exists a $2\pi$-periodic stochastic process $\{Z_\omega\}_{\omega\in\rnum}$ taking values in $L^2_\cnum([0,1])$ with $\overline{Z}_{\omega}= Z_{-\omega}$ such that
$\lim_{N\to\infty}\mean\norm{\ZN{\omega}-Z_{\omega}}^2_2=0$. Furthermore, $\{Z_\omega\}$ equals  almost everywhere the functional orthogonal increment process of the Cramer representation of $\{X_t\}$, that is,
\[
X_t = \int_{-\pi}^{\pi} e^{\im\omega t}\,dZ_{\omega}\qquad\text{a.e. in $\Hspace$.} 
\]
Finally, we have for $k\geq 2$
\begin{align*}
\cum\big(Z_{\omega_1},\ldots,Z_{\omega_k}\big)
=\int_{-\pi}^{\omega_1}\cdots \int_{-\pi}^{\omega_k}
   \eta\Big(\lsum_{j=1}^{k} \lambda_j\Big) \F_{\alpha_1,\ldots,\alpha_{k-1}}
    \,d\alpha_1 \cdots d\alpha_k, \tageq \label{WNcum}
\end{align*}
which holds almost everywhere and in $L^2$.
\end{theorem}
The final statement of the above theorem suggests the use of the differential notation
\begin{align*}
\cum\big(dZ_{\omega_1}(\tau_1),\ldots,dZ_{\omega_k}(\tau_k)\big)
&=\eta(\omega_1+\ldots+\omega_k)\,
     f_{\omega_1,\ldots,\omega_{k-1}}(\tau_1,\ldots,\tau_k)\,
     d\omega_1\cdots d\omega_k. 
\end{align*} 
\begin{proof}[Proof of Theorem \ref{cumWN}] 
The theorem generalizes Theorem 4.6.1 of \citet{Brillinger}. Let $\mu$ be the measure on the interval $[-\pi,\pi]$ given by 
\[
\mu(A)=\int_{A}\snorm{\mathcal{F}_{\omega}}_1\,d\omega,
\]
for all Borel sets $A\subseteq[-\pi,\pi]$. Similar to the time series setting, it has been shown \citep{Panar2013b} that there is a unique isomorphism $\Iso$ of $\overline{\text{sp}}\{X_t \}_{t \in \znum}$ onto $L^2_\cnum([-\pi,\pi],\mu)$ such that
\[
\Iso X_{t}= e^{\im t \cdot}
\]
for all $t\in\znum$. The process defined by $Z_{\omega} = \Iso^{-1}\big(1_{(-\pi,\omega]}(\cdot)\big)$ is then a functional orthogonal increment process of which the second order properties are completely determined by the spectral density operator $\F$. We have
\begin{align*}
\Iso(Z_{\omega}-Z_{\nu}) &= 1_{(\nu,\omega]}(\cdot), \qquad -\pi < \nu < \omega < \pi,
\intertext{and for $b_j\in\cnum$, $j=1,\ldots,N$}
\Iso \Big(\lsum_{j=1}^{N} b_j X_{t_j}\Big) &= \lsum_{j=1}^{N} b_j e^{\im t_j (\cdot)}.
\end{align*}
For the first part of the proof, we shall use that the function $1_{(-\pi,\omega]}(\cdot)$ can be approximated by the $N$-th order Fourier series approximation
\begin{align*}
b_{N}(\lambda) & =\lsum_{|t|\leq N}\tilde{b}_{\omega, t}\,e^{\im t \lambda},
\end{align*}
where the Fourier coefficients are given by
\begin{align}
\label{eq:coefh}
\tilde{b}_{\omega,t} =\frac{1}{2\pi} \int^{\pi}_{-\pi} 1_{(-\pi,\omega]}(\lambda)\,e^{-\im t \lambda}\,d\lambda.
\end{align}
The approximation satisfies  the properties listed in the following proposition \citep[][Proposition 4.11.2]{BrockwellDavis}.
\begin{proposition} \label{propidprox}
Let $\{b_{N}\}_{N \ge 1}$ be the sequence of functions defined in \eqref{eq:coefh}. Then for $-\pi < \nu < \omega < \pi$,
\begin{romanlist}
\item \label{prop1FS}
$\DS\sup_{\lambda \in [-\pi,\pi] \setminus \mathcal{E}} | b_{N}(\lambda) - 1_{(\nu,\omega]}(\lambda)\big| \to 0$ as $N \to \infty$, where $\mathcal{E}$ is an open subset of $[-\pi,\pi]$ containing both $\nu$ and $\omega$; 
\item \label{prop2FS}
$\DS\sup_{\lambda \in [-\pi,\pi] } |b_{N}(\lambda) | \leq C < \infty$ for all $N \geq 1$.
\end{romanlist}
\end{proposition}

Note then that we can write
\begin{align*}
Z^{(N)}_{\omega} = \frac{1}{2\pi}\sum_{|t|\leq N}X_{t}\int^{\pi}_{-\pi} 1_{(-\pi,\omega]}(\lambda)\,e^{\im t \lambda}\,d\lambda
=\sum_{|t|\leq N}\tilde{b}_{\omega,t}\,X_{t},
\end{align*}
where $\{\tilde{b}_{\omega,t}\}_{t\in\nnum}$ are the Fourier coefficients of the indicator function $1_{(-\pi,\omega]}$. Therefore,
\begin{align*}
\cum &\big(Z^{(N)}_{\omega_1},\ldots,Z^{(N)}_{\omega_k}\big)\\
& = \sum_{|t_1|,\ldots, |t_k|\leq N} \tilde{b}_{\omega_1, t_1}\cdots \tilde{b}_{\omega_k, t_k}\,
\cum\big(X_{t_1}\ldots , X_{t_k}\big)
\intertext{and by stationarity of the process $X_t$}
&= \sum_{|t_1|,\ldots, |t_k|\leq N} \tilde{b}_{\omega_1, t_1}\cdots \tilde{b}_{\omega_k, t_k}\,\int_{{\Pi}^{k}} e^{\im (\alpha_1 t_1 +\ldots + \alpha_{k} t_{k})}  \eta\Big(\lsum_{j=1}^{k} \alpha_j\Big)\, \F_{\alpha_1 \ldots \alpha_{k-1}}\,d\alpha_1 \cdots d\alpha_{k}\\
& = \int_{{\Pi}^{k}}   \eta\Big(\lsum_{j=1}^{k} \alpha_j\Big)\,\F_{\alpha_1 \ldots \alpha_{k-1}}    \prod_{i=1}^{k}\sum_{|t_i|<N} \bigg(\int_{{\Pi}^{k}}  1_{(-\pi,\omega_i]}(\lambda_i) e^{-\im t_i \lambda_i} d\lambda_{i} \bigg)
\,e^{\im \alpha_i t_i} d\alpha_1 \cdots d\alpha_{k} \\
&=  \int_{{\Pi}^{k}}   \eta\Big(\lsum_{j=1}^{k} \alpha_j\Big)\, \F_{\alpha_1 \ldots \alpha_{k-1}}\, b_{\omega_1,N}(\alpha_1 ) \cdots b_{\omega_k,N}(\alpha_k)\,d\alpha_1 \cdots d\alpha_{k} .
\end{align*}
To show convergence, recall that the kernel function $\F_{\alpha_1 \ldots \alpha_{k-1}}$ is bounded and uniformly continuous in the manifold $\sum_{j=1}^{k} \alpha_j \equiv 0 \mod(2\pi)$ with respect to $\| \cdot\|_2$. An application of H{\"o}lder's inequality yields
\begin{align*}
\Bignorm{ & \int_{{\Pi}^{k}}   \eta(\lsum_{j=1}^{k} \alpha_j) \F_{\alpha_1 \ldots \alpha_{k-1}}  \Big[ b_{\omega_1,N}(\alpha_1 ) \cdots b_{\omega_k,N}(\alpha_k) -1_{(-\pi,\omega_1]}(\alpha_1) \cdots 1_{(-\pi,\omega_k]}(\alpha_k)  \Big]d\alpha_1 \cdots d\alpha_{k} }_2  
\\& \le \sup_{\alpha_1,\ldots,\alpha_{k-1}}\|  \F_{\alpha_1 \ldots \alpha_{k-1}}\|_2 \int_{{\Pi}^{k}} \Big| b_{\omega_1,N}(\alpha_1 ) \cdots b_{\omega_k,N}(\alpha_k) -1_{(-\pi,\omega_1]}(\alpha_1) \cdots 1_{(-\pi,\omega_k]}(\alpha_k)  \Big| d\alpha_1 \cdots d\alpha_{k} 
\end{align*}
A standard telescoping argument together with Proposition \ref{propidprox} gives \begin{align*}  
& \le K\int_{{\Pi}^{k}} \sum_{j=1}^{k} \lprod_{l=1}^{j-1} \big|b_{\omega_l,N}(\alpha_{l})\big|
\lprod_{l=j+1}^{k} \big| 1_{(-\pi,\omega_1]}(\alpha_{l} )\big|\,  \big|b_{\omega_j,N}(\alpha_{j})-1_{(-\pi,\omega_{j}]}(\alpha_{j})\big|
\,d\alpha_1 \cdots d\alpha_{k} \\
& \le K\,k\,\big(\sup_{1\leq j\leq k}\sup_{\alpha} |b_{\omega_j,N}(\alpha)|\big)^{k-1}\,\sup_{\omega} \int_{{\Pi}}   \big|b_{\omega,N}(\alpha)- 1_{(-\pi,\omega]}(\alpha)  \big| d\alpha \to 0
\end{align*}
as $N\to\infty$.
Hence, the dominated convergence theorem implies
\begin{align*}
& \lim_{N \to \infty} \cum \big(Z^{(N)}_{\omega_1},\ldots,Z^{(N)}_{\omega_k}\big)
\\&  =
 \frac{1}{({2\pi})^{k}}  \int_{{\Pi}^{k}}1_{(-\pi,\omega_1]}(\alpha_1) \cdots 1_{(-\pi,\omega_k]}(\alpha_k) \F_{\alpha_1 \ldots \alpha_{k-1}} \eta\Big(\lsum_{j=1}^{k} \alpha_j\Big) d\alpha_1 \cdots d\alpha_{k} \\& 
= \frac{1}{({2\pi})^{k}} \int_{-\pi}^{\omega_1} \cdots  \int_{-\pi}^{\omega_k} \eta\Big(\lsum_{j=1}^{k} \lambda_j\Big)\,\F_{\alpha_1 \ldots \alpha_{k-1}}d\lambda_1 \cdots d\lambda_{k},  \\&
=  \cum \big(Z_{\omega_1},\ldots,Z_{\omega_k}\big) \tageq \label{eq:cumZlim}
\end{align*}
which establishes the $L^2$ convergence in \eqref{WNcum}. The almost everywhere convergence is proved similarly by replacing $\F$ by $f(\tau_1,\ldots,\tau_k)$. 
In order to show that $X_t  = \int_{-\pi}^{\pi} e^{\im \omega t} dZ_{\omega}$ with probability 1, it remains to show that 
\begin{align}\label{eq:sprepZ}
\E \Bignorm{ X_t - \int_{-\pi}^{\pi} e^{\im \omega t} dZ_{\omega}}_2^2 = 0. \end{align}
We refer to \citet{Panar2013b} for a proof.

\end{proof}


\subsection{Higher order dependence for time-dependent linear models} \label{MAfilters}
\begin{proposition} \label{MAcumrep}
Let $\{\veps_t\}_{t \in \znum}$ be a functional \IID process in $H$ with $\E\|\varepsilon_0 \|^k_2 < \infty$, $k \in \nnum$ and let $\{\AtT{s}\}_{s\in\znum}$ be a sequence of operators in $S_{\infty}(H)$ satisfying $\sum_{s}\snorm{\AtT{s}}_{\infty}<\infty$ for all $t=1,\ldots,T$ and $T \in\nnum$. Then the process $\XNT{t}=\sum_{|s| \le N } \AtT{s} \varepsilon_{t-s}$ has the following properties:
\begin{romanlist}
\item\label{item:Macumone}
$\XNT{t}$ converges to a process $\XT{t}$ in $L^k_H(\Omega,\prob)$;
\item\label{item:Macumtwo}
$\DS\cum(\XT{t_1},\ldots,\XT{t_k}) = \big(\lsum_{s_1 \in \znum} \AttT{t_1}{s_1} \otimes \cdots \otimes  \lsum_{s_k \in \znum} \AttT{t_k}{s_k} \big) \cum(\varepsilon_{{t_1}-s_1},\ldots, \varepsilon_{{t_k}-s_k} )$, where the convergence is with respect to $\norm{\cdot}_2$.
\end{romanlist}
\end{proposition}

\begin{proof}[Proof of Proposition \ref{MAcumrep}]
For the first equality, we need to show that 
\[\lim_{N \to \infty} \E \|\XNT{t}-\XT{t} \|^k_2=0.\] 
We will do this by demonstrating that the the tail series $X^{-(N)}_{t,T}= \sum_{ s =N+1 }^{M} \AtT{s} \varepsilon_{t-s}$ converges. Since $\|\AtT{s} \varepsilon_t\|_2 \le \snorm{\AtT{s}}_{\infty} \|\varepsilon_t\|_2$, an application of the generalized H{\"o}lder's Inequality yields
\begin{align*}
\E\|X^{-(N)}_{t_1,T}\|^k_2 & \le \sum_{s_1,\ldots,s_k = N+1 }^{M} \snorm{\AttT{t_1}{s_1}}_{\infty} \cdots  \snorm{\AttT{t_k}{s_k}}_{\infty} \E \big[\|\varepsilon_{t_1-s_1}\|_2 \cdots \|\varepsilon_{t_k-s_k}\|_2\big] \\& \le \sum_{|s_1|,\ldots,|s_k| > N } \snorm{\AttT{t_1}{s_1}}_{\infty} \cdots  \snorm{\AttT{t_k}{s_k}}_{\infty} \big[\E \|\varepsilon_{t_1-s_1}\|^k_2 \cdots \E \|\varepsilon_{t_k-s_k}\|^k_2 \big]^{1/k} \\&  \le \big(\sum_{|s| > N } \snorm{\AtT{s}}_{\infty}\Big)^k \E\|\varepsilon_{0}\|^k_2 < \infty,  
\end{align*}
uniformly in $M$. Hence, $\lim_{N \to \infty} \big(\mean\norm{\XT{t}^{-(N)}}^k_2\big)^{1/k}=0$.

We now prove (\romannumeral 0\ref{item:Macumtwo}). By Proposition \ref{momlin} and (\romannumeral 0\ref{item:Macumone}), we have 
\begin{align*}
\cum(\AttT{t_1}{s_1}\varepsilon_{{t_1}-s_1},\ldots,  \AttT{t_k}{s_k}  \varepsilon_{{t_k}-s_k} ) = \Big(\AttT{t_1}{s_1} \otimes \cdots \otimes  \AttT{t_k}{s_k} \Big)  \cum(\varepsilon_{{t_1}-s_1},\ldots, \varepsilon_{{t_k}-s_k} ). 
\end{align*}
It is therefore sufficient to show that 
\begin{align*}
 \cum\Big(\lsum_{s_1 \in \znum} \AttT{t_1}{s_1}\varepsilon_{{t_1}-s_1},\ldots, \lsum_{s_k \in \znum} \AttT{t_k}{s_k} \varepsilon_{{t_k}-s_k})= \sum_{s_1,\ldots,s_k \in \znum}\cum(\AttT{t_1}{s_1}\varepsilon_{{t_1}-s_1},\ldots, A_{s_k,t_{k}, T} \varepsilon_{{t_k}-s_k} ).
\end{align*}
Let $\{\psi_l\}_{l \in \nnum}$ be an orthonormal basis of $H$. Then $\{\psi_{l_1} \otimes \cdots \otimes \psi_{l_k} \}_{l_1,\ldots,l_k \ge 1}$ forms an orthonormal basis  $\bigotimes_{j=1}^{k} H$.  For the partial sums 
\begin{align*}
\lsum_{s_j=1 }^{N} \AttT{t_j}{s_j}\varepsilon_{{t_j}-s_j}, \qquad j=1,\cdots, k,
\end{align*}
we obtain by virtue of the triangle inequality, the Cauchy-Schwarz Inequality and generalized H{\"o}lder Inequality
\begin{align*}
\E\| \lprod_{j=1}^{k} \lsum_{s_j=1 }^{N} \AttT{t_j}{s_j}\varepsilon_{{t_j}-s_j}\psi_{l}\|_1
&\leq  \lprod_{j=1}^{k} \E \Bignorm{\lsum_{s_j=1 }^{N} \AttT{t_j}{s_j}\varepsilon_{{t_j}-s_j}\psi_{l}}_1\\
&\leq \big(\sup_{t,T}\lsum_{s \in \znum}\snorm{\AttT{t}{s}}_{\infty}\big)^k\, \E\|\varepsilon_0\|^k_2 < \infty.
\end{align*}
The result now follows by the dominated convergence theorem.
\end{proof}
\subsection{Existence of the stochastic integral} \label{proofStochInt}
In order to provide sufficient conditions for local stationarity of functional processes in terms of spectral representations, we turn to investigating the conditions under which stochastic integrals
$\int_{-\pi}^{\pi} U_{\omega}\,dZ_{\omega}$ for $S_{\infty}(H_\cnum)$-valued functions $U_{\omega}$ are well-defined. For this, let $\mu$ be a measure on the interval $[-\pi,\pi]$ given by 
\[
\mu(A)=\int_{A}\snorm{\mathcal{F}_{\omega}}_1\,d\omega, \tageq \label{eq:measure}
\]
for all Borel sets $A\subseteq[-\pi,\pi]$ and let $\mathcal{B}_{\infty} =L^2_{S_{\infty}(H_\cnum)}([-\pi,\pi],\mu)$ be the corresponding Bochner space of all strongly measurable functions $U:[-\pi,\pi]\to S_{\infty}(H_\cnum)$ such that 
\[
\|U\|^2_{\mathcal{B}_\infty} =\int_{-\pi}^{\pi} \snorm{U_{\omega}}^2_{\infty} d\mu(\omega)<\infty.   \tageq \label{eq:Bnorm}
\]
\citet{Panar2013b} showed that the stochastic integral is well defined in $\mathbb{H}_\cnum$ for operators that belong to the Bochner space $\mathcal{B}_{2}=L^2_{S_{2}(H_\cnum)}([-\pi,\pi],\mu)$, which is a subspace of $\mathcal{B}_\infty$. In particular, it contains all functions $U:[-\pi,\pi]\to S_{2}(H_\cnum)$ of the form
\[
U_\omega=g(\omega)\,I + A_{\omega},
\]
where $g$ and $A$ are, respectively, $\cnum$ and $S_2(H_\cnum)$-valued functions that are both c{\`a}dl{\`a}g with a finite number of jumps and $A$ additionally satisfies $\int_{-\pi}^{\pi} \snorm{A_{\omega}}^2_{2}\, \snorm{\mathcal{F}_{\omega}}_1\,d\omega < \infty$. Here, continuity in $S_2(H_\cnum)$ is meant with respect to the operator norm $\snorm{\cdot }_{\infty}$. Because the space $\mathcal{B}_{2}$ is too restrictive to include interesting processes such as general functional autoregressive processes, we first show that the integral is properly defined in $\mathbb{H}_\cnum$ for all elements of $\mathcal{B}_\infty$. To do so, consider the subspace $\mathcal{Q}_0 \subset \mathcal{B}_{\infty}$ of step functions spanned by elements $U\,\mathbf{1}_{[\alpha,\beta)}$ for $U \in S_{\infty}(H_\cnum)$ and $\alpha < \beta \in [-\pi, \pi]$. Additionally, denote its closure by $\mathcal{Q} = \overline{\mathcal{Q}_0}$. Define then the mapping $\mathcal{T}: \mathcal{Q}_0 \mapsto \mathbb{H}_\cnum$ by linear extension of
\begin{align}\label{eq:chap2isomap}
\mathcal{T}(U\,\mathbf{1}_{[\alpha,\beta)}) = U(Z_{\beta}-Z_{\alpha}).
\end{align}
The following lemma shows that the image of $\mathcal{T}$ is in $\mathbb{H}_\cnum$.

\leqnomode
\begin{lemma}\label{bochnerlemma}
Let $X_t$ be a functional process with spectral representation $X_t= \int_{-\pi}^{\pi} e^{\mathrm{i}\omega t} d Z_{\omega}$ for some functional orthogonal increment process $Z_{\omega}$ that satisfies $\E \| Z_{\omega}\|_2^2 = \int_{-\pi}^{\omega}\snorm{ \mathcal{F}_{\lambda}}_1 d \lambda$. Then for $U_1, U_2 \in S_{\infty}(H_\cnum)$ and $\alpha,\beta\in[-\pi,\pi]$
\begin{align}
\tag{$i$} \langle U_1Z_{\alpha}, U_2 Z_{\beta}\rangle_{\mathbb{H}_\cnum} &= \trace\Big(U_1 \Big[\int_{-\pi}^{\alpha\wedge\beta}\mathcal{F}_{\omega}\,d\omega \Big] U_2^{\dagger}\Big)
\intertext{and}
\tag{$ii$} \|U_1 Z_{\alpha}\|^2_{\mathbb{H}_\cnum} &\le \snorm{U_1}^2_{\infty}   \int_{-\pi}^{\alpha}\snorm{ \mathcal{F}_{\lambda}}_1\,d\lambda. 
\end{align}
\end{lemma}
\reqnomode

\begin{proof}[Proof of Lemma \ref{bochnerlemma}]
Firstly, we note that by Cauchy-Schwarz inequality
\begin{align*} 
\E\int_{0}^1 \ |U_1 Z_{\alpha}(\tau) U_2 Z_{\beta}(\tau)|d\tau
&\le  \E\| U_1Z_{\alpha}\|_2 \|U_2 Z_{\beta}\|_2
\le \snorm{U_1}_{\infty} \snorm{U_2}_{\infty} \E\|Z_{\alpha}\|_2\|Z_{\beta}\|_2\\
&\le \snorm{U_1}_{\infty} \snorm{U_2}_{\infty}   \int_{-\pi}^{\alpha\wedge\beta}\snorm{ \mathcal{F}_{\lambda}}_1 d \lambda < \infty. \tageq \label{eq:firstlemma1}
\end{align*}
Secondly, $U_1Z_{\alpha}$ and $U_2 Z_{\beta}$ are elements in $H_\cnum$ and therefore the (complete) tensor product $U_1 Z_{\alpha} \otimes U_2 \overline{Z_{\beta}}$ belongs to $S_2(H_\cnum)$. By Proposition \ref{HSkernel}, it is thus a kernel operator with kernel [$U_1 Z_{\alpha} \otimes U_2\overline{Z_{\beta}}](\tau,\sigma) = U_1 Z_{\alpha}(\tau) \overline{U_2 Z_{\beta}}(\sigma)$. An application of Fubini's Theorem yields
\begin{align*}
\E\int_{0}^1
&\ U_1 Z_{\alpha}(\tau) \overline{U_2 Z_{\beta}}(\tau) d\tau
=\int_{0}^1 \ E \big(U_1 Z_{\alpha} \otimes U_2\overline{Z_{\beta}}\big)(\tau, \tau) d\tau \\
&= \int_{0}^1 \  \big(U_1 \otimes U_2)\E(Z_{\alpha} \otimes \overline{Z_{\beta}}\big)(\tau, \tau) d\tau 
= \int_{0}^1 \  \big(U_1 \otimes U_2)\int_{-\pi}^{\alpha\wedge\beta}\mathcal{F}_{\omega}\, d\omega (\tau, \tau) d\tau \\
& =\int_{0}^1 \ U_1 \int_{-\pi}^{\alpha\wedge\beta}\mathcal{F}_{\omega}\,d\omega (\tau, \tau)\,U_2^{\dagger}\,d\tau ,
\end{align*}
where the second equality follows because the expectation commutes with bounded operators for integrable random functions (Proposition \ref{momlin}) and the last equality follows from the identity \eqref{eq:tensorabc} of definition \ref{tensorproductop}. This shows the first result of Lemma \ref{bochnerlemma}. The second result follows straightforwardly from \eqref{eq:firstlemma1}. 
\end{proof}
It is easily seen from the previous lemma that for $\lambda_1 > \lambda_2 \ge \lambda_3 > \lambda_4$
\begin{align*}
\langle U_1(Z_{\lambda_1}-Z_{\lambda_2}), U_2(Z_{\lambda_3}-Z_{\lambda_4})\rangle_{\mathbb{H}_\cnum}= 0, 
\end{align*}
demonstrating orthogonality of the increments is preserved. Since every element $U_n \in \mathcal{Q}_0$ can be written as $\sum_{j=1}^{n}U_j\mathbf{1}_{[\lambda_{j},\lambda_{j+1})}$ the lemma moreover implies
\begin{align*}
\|\mathcal{T}(U)\|^2_{\mathbb{H}_\cnum}& = \sum_{j,k=1}^{n}\langle U_j (Z_{\lambda_{j+1}}-Z_{\lambda_j}), U_k (Z_{\lambda_{k+1}}-Z_{\lambda_k})\rangle_{\mathbb{H}_\cnum}=  \sum_{j=1}^{n}\| U_j (Z_{\lambda_{j+1}}-Z_{\lambda_j})\|^2_{\mathbb{H}_\cnum} \\& \le \sum_{j=1}^{n} \snorm{U_j}^2_{\infty} \int_{\lambda_j}^{\lambda_{j+1}} \snorm{ \mathcal{F}_{\alpha}}_1 d \alpha  = \|U\|^2_{\mathcal{B}_{\infty}}.
\end{align*}
The mapping $\mathcal{T}: \mathcal{Q}_0 \mapsto \mathbb{H}_\cnum$ is therefore continuous. Together with the completeness of the space $\mathbb{H}_\cnum$ this establishes that, for every sequence $\{U_n\}_{n \ge 1} \subset \mathcal{Q}_0$ converging to some element $U \in \overline{\mathcal{Q}}$, the sequence $\{\mathcal{T}(U_n)\}_{n \ge 1}$ forms a Cauchy sequence in $\mathbb{H}_\cnum$ with limit $\mathcal{T}(U) = \lim_{n \to \infty}\mathcal{T}(U_n)$.
By linearity and continuity of the mapping $\mathcal{T}$, the limit is independent of the choice of the sequence. Furthermore, since $\mathcal{Q}_0$ is the subspace spanned by step functions that are square integrable on $[-\pi,\pi]$ with respect to the finite measure $\mu$ and hence is dense in $L^2_{S_{\infty}(H_\cnum)}([-\pi,\pi],\mu)$, we have $\mathcal{B}_{\infty}\subseteq\overline{\mathcal{Q}}$. Since $\|\mathcal{T}(U)\|_{\mathbb{H}_\cnum} \le \|U \|_{\mathcal{B}_{\infty}}$, the above extension is well-defined for all $U\in\mathcal{B}_{\infty}$.

\section{Data taper}
\label{taperandl}

In order to show convergence of the higher order cumulants of the estimator in \eqref{eq:smoothest}, we will make use of two lemmas from \citet{Dahlhaus1993} (Lemma A.4 and A.5 resp.). Both rely on the function $L_T:\mathbb{R} \to \mathbb{R}, T \in \mathbb{R}^{+}$, which is the $2 \pi$-periodic extension of
\begin{align} L_T(\lambda) =
  \begin{cases}
    T       & \quad \text{if } |\lambda| \le 1/T, \\
    1/|\lambda|  & \quad \text{if } 1/T \le |\lambda| \le \pi.\\
  \end{cases}
\end{align}
The function $L_T$ satisfies some nice properties. The following lemma lists those required in the current paper:
\begin{lemma} \label{Dh1993A4}
Let $k, l, T \in \mathbb{N}, \lambda, \alpha, \omega, \mu, \gamma \in \mathbb{R}$ and $\Pi: (-\pi, \pi]$. The following inequalities then hold with a constant $C$ independent of $T$.
\begin{romanlist}
\item
$\DS L_T(\lambda)$ is monotone increasing in $T$ and decreasing in $\lambda \in [0,\pi]$;
\item
$\DS|\lambda| L_T(\lambda) \le C$ for all $|\lam|\leq \pi$;
\item
$\DS\int_{\Pi} L_T(\lambda) d\lambda \le C\,\log T$;
\item
$\DS\int_{\Pi} L_T(\lambda)^k\,d\lambda \le C\,T^{k-1}$ for $k>1$;
\item
$\DS\int_{\Pi} L_T(\alpha-\lambda)\,L_T(\lambda+\gamma)\, d\lambda \le C\, L_T(\alpha+\gamma) \log T$.
\end{romanlist}
\end{lemma}
In addition, we also make use of Lemma 2 from \citet{Eichler2007}.
\begin{lemma}
\label{EichlerL}
Let $\{P_1,\ldots,P_m\}$ be an indecomposable partition of the table
\[
\begin{matrix}
\alpha_1&-\alpha_1\\
\vdots&\vdots\\
\alpha_n&-\alpha_n
\end{matrix}
\]
with $n\geq 3$. For $P_j=\{\gamma_{j1},\ldots,\gamma_{jd_j}\}$, let $\bar\gamma_j=\gamma_{j1}+\ldots+\gamma_{jd_j}$.
\begin{romanlist}
\item
If $m=n$ then for any $n-2$ variables
$\alpha_{i_1},\ldots,\alpha_{i_{n-2}}$ we have
\[
\int_{\Pi^{k-2}}\lprod_{j=1}^{n}
L_T(\bar\gamma_j)\,d\alpha_{i_1}\cdots d\alpha_{i_{n-2}}
\leq C\,L_N(\alpha_{i_{n-1}}\pm\alpha_{i_n})^2\,\log(T)^{n-2}.
\]
\item
If $m<n$ then there exists $n-2$ variables $\alpha_{i_1},\ldots,\alpha_{i_{n-2}}$ such that
\[
\int_{\Pi^{k-2}}\lprod_{j=1}^{n}
L_T(\bar\gamma_j)\,d\alpha_{i_1}\cdots d\alpha_{i_{n-2}}
\leq C\,T\,\log(T)^{n-2}.
\]
\end{romanlist}
\end{lemma}

The usefulness of the $L_T$ function stems from the fact that it gives an upperbound for the function $H_{k,N}$ which was defined in section \ref{expandcov}. Namely, we have 
\begin{align} \label{Hbound} |H_{k,N}^{(\lambda)}| \le L_N(\lambda), \forall k \in \mathbb{N}.
\end{align}
We also require an adjusted version of Lemma A.5 of \citet{Dahlhaus1993}:
\begin{lemma} \label{Dh1993A5}
Let $N,T \in \mathbb{N}$. Suppose $h$ is a data-taper of bounded variation and let the operator-valued function $G_{u}:[0,1] \to S_{p}(H)$ be continuously differentiable in $u$ such that $\bigsnorm{\frac{\partial G_{u}}{\partial u}}_{p} < \infty$ uniformly in $u$. Then we have for $0 \le t \le N,$
\begin{align*}
\HN(G_{\frac{\bullet}{T}},\omega) & =  \HN(\omega)G_{\frac{t}{T}} + O\Big(\sup_{u}\Bigsnorm{\frac{\partial }{\partial u}G_{u}}_{p}\frac{N}{T}L_N(\omega)\Big) \\& 
= O\Big(\sup_{u \le N/T}\snorm{G_{u}}_{p}\frac{N}{T}L_N(\omega) +\sup_{u}\Bigsnorm{\frac{\partial }{\partial u}G_{u}}_{p}\frac{N}{T}L_N(\omega)\Big), \tageq \label{eq:LboundHn}
\end{align*}
where $\HN(G_{\bullet},\omega)$ is as in \eqref{Hgs}
The same holds if $G_{\frac{\bullet}{T}}$ on the left hand side is replaced by operators $G^{(T)}_{{\bullet}}$ for which $\sup_s\snorm{G^{(T)}_{{\bullet}}-G_{\frac{\bullet}{T}}}_p=O(\frac{1}{T})$.
\end{lemma}
\begin{proof}
Summation by parts gives 
\begin{align*}
\HN(G_{\frac{\bullet}{T}},\omega) -\HN(\omega)G_{\frac{t}{T}} & 
= \sum_{s = 0}^{N-1} [G_{\frac{s}{T}}-G_{\frac{t}{T}}]\hN{s} e^{-\mathrm{i}\omega s } \\& 
= - \sum_{s=0}^{N-1}[G_{\frac{s}{T}}-G_{\frac{s-1}{T}}]H_s(\hN{\bullet},\omega)  + [G_{\frac{N-1}{T}}-G_{\frac{t}{T}}]\HN(\omega).
\end{align*}
It has been shown in \citet{Dahlhaus1988} that $|H_s(\hN{\bullet},\omega) | \le KL_s(\omega) \le K L_N(\omega)$. The result in \eqref{eq:LboundHn} then follows since
\[\snorm{G_b - G_a}_{p} \le \sup_{a< \xi <b} \Bigsnorm{\frac{\partial }{\partial u}G_{u}\big|_{u = \xi}}_{p}|b-a|, \qquad a,b\in\rnum,\]
by the Mean Value Theorem. The lemma holds additionally for operators $G^{(T)}_{{\bullet}}$ that satisfy $\sup_s\snorm{G^{(T)}_{{\bullet}}-G_{\frac{\bullet}{T}}}_p=O(\frac{1}{T})$. This is a consequence of Minkowski's inequality since
\begin{align*}
& \bigsnorm{\HN(G^{(T)}_{{\bullet}}-G_{\frac{\bullet}{T}},\omega) +  \HN(G_{\frac{\bullet}{T}},\omega)}_p 
\\& = \bigsnorm{\HN(G^{(T)}_{{\bullet}}-G_{\frac{\bullet}{T}},\omega)}_p +\snorm{ \HN(G_{\frac{\bullet}{T}},\omega)}_p 
 \\& = O\big(\frac{N}{T}+ L_N(\lambda)\big) = O(L_N(\lambda)). \tageq
\end{align*}
Hence, the replacement error is negligible compared to the error of \ref{eq:LboundHn}.
\end{proof}
If $p=2$, the above implies that the kernel function $g_u \in H^2_\cnum$ of $G_u$ satisfies
\begin{align*}  
&\|\HN(g_{\frac{\bullet}{T}},\omega) - \HN(\omega)g_{\frac{t}{T}}\|_2 =R_{1,N},\\
& \|\HN(g_{\frac{\bullet}{T}},\omega)\|  = R_{2,N}+R_{1,N},
\end{align*}
where 
\begin{align*}
\|R_{1,N}\|_2 &= O\Big(\sup_{u} \|\frac{\partial }{\partial u}g_{u}\|_{p}\frac{N}{T}L_N(\omega)\Big), \\
\|R_{2,N}\|_2& = O\Big(\sup_{u \le N/T}\|g_{u}\|_{p}\frac{N}{T}L_N(\omega). \tageq 
\end{align*}
Similarly if $g_{\frac{\bullet}{T}}$ on the left hand side is replaced by the kernel function $g^{(T)}_{\bullet} \in H^2_\cnum$ of $G^{(T)}_{\bullet}$. If the kernels are bounded uniformly in their functional arguments, Lemma A.5 of \citet{Dahlhaus1993} is pointwise applicable.

\bibliographystyle{stat}


\end{document}